\documentclass[preprintnumbers,floatfix,superscriptaddress,pra,onecolumn,showpacs,notitlepage]{revtex4-1}
\usepackage{amsmath,amsthm,amssymb}
\usepackage{graphicx}
\usepackage{dcolumn}
\usepackage{bbold}
\usepackage{hyperref} 
\usepackage{braket}
\usepackage{color}
\usepackage{dsfont}
\usepackage{pifont}
\usepackage{setspace}
\usepackage{multirow}
\usepackage{booktabs}
\usepackage[font=footnotesize,justification=raggedright]{caption}
\usepackage[normalem]{ulem}
\usepackage{enumerate}
\usepackage{algorithm}
\usepackage{physics}
\usepackage{lipsum}
\usepackage{epstopdf}
\usepackage[noend]{algpseudocode}
\usepackage{enumitem}
\usepackage{changepage}

\definecolor{cinnabar}{rgb}{0.89, 0.26, 0.2}
\newcommand{\vl}[1]{\textcolor{black}{#1}}

\newcommand{\cmark}{\textnormal{\ding{51}}}
\newcommand{\xmark}{\textnormal{\ding{55}}}
 
\newcommand{\tn}{\textnormal}
\newcommand{\inn}{\textnormal{in}}
\newcommand{\out}{\textnormal{out}}
\newcommand{\epr}{\textnormal{EPR}}
\newcommand{\kt}{{\kappa}}
\newcommand{\pavg}{{P}_{\cmark}}
\newcommand{\prep}{\textsc{Prep~}}

\theoremstyle{definition}
\newtheorem{definition}{Definition}

\newtheorem{corollary}{Corollary}

\theoremstyle{theorem}
\newtheorem{theorem}{Theorem}
\newtheorem*{theorem*}{Theorem}
\newtheorem{lemma}{Lemma}
\newcolumntype{C}[1]{>{\centering\arraybackslash}m{#1}}

\makeatletter
\renewcommand{\ALG@name}{Test}
\makeatother

\begin{document}
\title{\vl{Certification of a functionality in a quantum network stage}}

\author{Victoria Lipinska}\email{v.lipinska@tudelft.nl}
\affiliation{QuTech, Delft University of Technology, Lorentzweg 1, 2628 CJ Delft, The Netherlands}
\affiliation{Kavli Institute of Nanoscience, Delft University of Technology, Lorentzweg 1, 2628 CJ Delft, The Netherlands}
\author{L\^{e} Phuc Thinh}
\affiliation{QuTech, Delft University of Technology, Lorentzweg 1, 2628 CJ Delft, The Netherlands}
\affiliation{Institut f\"{u}r Theoretische Physik, Leibniz Universit\"{a}t Hannover, Appelstr.~2, D-30167 Hannover, Germany}
\author{J\'{e}r\'{e}my Ribeiro}
\affiliation{QuTech, Delft University of Technology, Lorentzweg 1, 2628 CJ Delft, The Netherlands}
\affiliation{Kavli Institute of Nanoscience, Delft University of Technology, Lorentzweg 1, 2628 CJ Delft, The Netherlands}
\author{Stephanie Wehner}
\affiliation{QuTech, Delft University of Technology, Lorentzweg 1, 2628 CJ Delft, The Netherlands}
\affiliation{Kavli Institute of Nanoscience, Delft University of Technology, Lorentzweg 1, 2628 CJ Delft, The Netherlands}

\date{\today}
\begin{abstract}
We consider testing the ability of quantum network nodes to execute multi-round quantum protocols. Specifically, 
we examine protocols in which the nodes are capable of performing quantum gates, storing qubits and exchanging said qubits over the network a certain number of times. 
We propose a simple ping-pong test, which provides a certificate for the capability of the nodes to run certain multi-round protocols. 
We first show that in the noise-free regime the only way the nodes can pass the test is if they do indeed possess the desired capabilities. We then proceed to consider the case where operations are noisy, and provide
an initial analysis showing how our test can be used to estimate parameters that allow us to draw conclusions about the actual performance of such protocols on the tested nodes. 
Finally, we investigate the tightness of this analysis using example cases in a numerical simulation.
\end{abstract}

\maketitle

\section{Introduction}
Quantum communication allows us to solve tasks that are impossible to achieve using classical communication alone. The most well known example of such a task is quantum key distribution 
(QKD)~\cite{Bennett1984, Ekert1991}, but many more application protocols are already known (see e.g.~\cite{Wehner2018}). Such application protocols run on the end nodes of a quantum network. These 
may range from simple photonic devices capable of preparing and measuring qubits, to sophisticated quantum processors. 
Recently, stages of development for a quantum internet were identified~\cite{Wehner2018}, where each stage is distinguished by a specific functionality that is offered to a user
wishing to execute quantum network applications. Higher stages bring an increase of functionality -- and thus a richer set of possible application protocols -- at the expense of increased experimental difficulty. 

Given such stages of development, one can ask whether there exists an efficient test to certify that a network offers the capabilities of a specific stage, and with what quality parameters. Here, we will examine this question with a focus on a specific set of protocols in the stage called a \emph{quantum memory network} \cite{Wehner2018}:

\begin{adjustwidth}{2cm}{2cm}
\vspace{1em}
\vl{\textit{
``For any two end nodes $A$ and $B$ the network allows
the execution entanglement generation
and the following additional tasks in any
order: 
(i) preparation of a single-qubit
ancilla state $\ket{\psi}$ by end node $A$ or $B$,
(ii) measurements of any subset of the
qubits at node $A$ or $B$,
(iii) application of
an arbitrary unitary gate $G$ at node $A$ or $B$.
(iv) storage of the qubits for a minimum
time $k \cdot (\ell_z + \tau)$, where $\tau$ is defined as the
time that is required to generate one
Einstein–Podolsky–Rosen (EPR) pair
and send a classical message from node $A$ to
$B$ maximized over all pairs of nodes, $\ell_z$ is the time that it takes for the execution of a depth $z$ quantum circuit at the
end node.''}}
\vspace{1em}
\end{adjustwidth}

\vl{Note that to realize useful application protocols, the storage time $\tau$ needs to be understood as the communication time in the network. In particular, the nodes that are far apart must exhibit longer storage capabilities to achieve this stage of development. Moreover, the stage is only attained if \emph{any} two nodes in the network can realize the functionality, even those that are farthest apart. Therefore, time $\tau$ can be thought of as the maximum time which takes any two nodes to communicate.}

\vl{To certify that a quantum network achieves a functionality defined by this stage of development, we will consider a set of protocols which pass a qubit state $\ket{\psi}$ a number $k$ of times between the nodes $A$ and $B$, apply the gates and measure at the end. We will choose the testing nodes $A$ and $B$ to be farthest apart in the network.}

Many existing tests are known that can be used to estimate whether the operations above can each be performed individually with high accuracy. Examples include quantum state~\cite{Chuang1997} and process tomography~\cite{Poyatos1997}, gate set tomography~\cite{Merkel2013,BlumeKohout2013}, randomized benchmarking~\cite{Emerson2005,Knill2008,Magesan2011} or capacity estimation to verify the quality of qubit transmission~\cite{Pfister2018}. The concept of self-testing even allows such estimates to be made with only partial trust in the devices \vl{(so called device-independent setting)}~\cite{Montanaro2016,Bancal2018}. Having estimated the quality of each individual operation with metrics such as the diamond distance, it is straightforward to derive an overall estimate on how well protocols in this stage may be executed~\cite{Wehner2018}. Yet, running many individual tests is rather inefficient, and one may wonder whether there might exist an integrated test that instills confidence that we are capable of performing protocols up to a certain number of rounds using the quantum memory network.

Another approach to testing quantum devices comes from the literature of (interactive) proof systems where a verifier interacts \vl{with one or more provers,} who are trying to convince the verifier that a certain assertion is true, or indeed that they possess certain capabilities. A well known example of such work is the question whether a classical polynomially bounded verifier can convince herself that \vl{(two non-communicating) provers holding a quantum computer do indeed have full quantum computing capabilities~\cite{Reichardt2013}. Restricting to only a single prover, there exists also a verification protocol under complexity theoretic assumptions~\cite{Mahadev2018}.} This line of research is not concerned with the quality of specific operations, but rather aims to obtain a certificate of the provers' general abilities to solve certain tasks. Such tests are appealing as they measure general aptitude -- for example in the domain of quantum computation the ability to execute quantum algorithms -- but do not typically make specific statements such as the actual number of physical qubits involved.  Consequently, such tests usually require large amount of resources to be executed.

\section{Results}
Here we take a first step towards finding \vl{effective tests to certify that a network has reached the quantum memory stage of development (see Def.~\ref{DEF:class_of_protocols})}. We propose a test which can be interpreted from two different angles. First, we interpret it as a prover-verifier type protocol inspired by interactive proof literature, to certify that the network has certain capabilities. Second, we interpret it as a tomography-type protocol where we estimate certain properties of operations. 

\begin{itemize}

\item \textbf{Ping-pong test.} We formulate our test in a bipartite scenario where nodes $A$ and $B$ exchange quantum registers according to a defined set of rules.  We call our test the ping-pong game as it is executed by passing qubits back and forth between two nodes. Additionally, the nodes apply gates specified by a gate set $\mathsf{G}$. An important parameter of our test is the number of times $k$ that the nodes pass (ping-pong) the state around. 

\item \textbf{Prover-verifier view.} Our protocol can be viewed as a simple game that the provers (the nodes) play against the verifier with the objective of convincing the verifier that they are capable of executing any protocol in the 
quantum memory stage, which has a specific form. 
In particular, we show that the provers win the $k$-round ping-pong game
with probability one if and only if they are capable of executing perfectly any protocol of the following form: for any possible starting state $\ket{\psi}$, each 
node is capable of executing one possible gate $G \in \mathsf{G}$, before sending 
the resulting state to the other node. The nodes continue in this form for $k$ rounds, before measuring at the end. Moreover, in the case when the winning probability is strictly less than one, we certify that the nodes sent information about the state at least a certain number $m<k$ of times. 

\item \textbf{Estimation view.} In the estimation view we take on a different perspective with the objective to estimate the quality of the operations performed by the nodes, as opposed to certifying their capabilities.
We use the statistics of the ping-pong test to assess a measure of the overall quality of the network. We then compare this to the quality one would expect from combining the estimates of the individual devices used in the network.
What is more, we estimate the performance of $k$-round protocols based on our ping-pong test. In order to evaluate the accuracy of our analytical results, we compare our analytical estimate with numerical estimates for a specific example of a $k$-round protocol influenced by noise.

\end{itemize}

This paper is organized as follows. In Sec.~\ref{sec:thetest} we define the $k$-round protocols and introduce our test. Then, inspired by the interactive proof literature, in Sec.~\ref{sec:prover_verifier} we view our test in the prover-verifier setting. In Sec.~\ref{sec:estimation} we view our test in the context of estimation.

\newpage
\section{Ping-pong test} \label{sec:thetest}
\vl{\subsection{Assumptions}}

\vl{\textbf{IID.} Protocol \ref{PROT:generic} (see Section \ref{SEC:test}) implements $n$ executions of a ping-pong procedure, each of them containing at most $k$ rounds. In Sec. \ref{sec:prover_verifier} (prover-verifier view) we assume that the execution of each of these ping-pong procedures is independent of the others and identical. In particular, this means that the provers' strategy will be the same in every execution of the ping-pong procedure, i.e. their strategy is independent and identically distributed (IID) across executions. However, within one execution, the provers' strategy can involve arbitrary correlation across rounds. 
Furthermore, in Sec. \ref{sec:estimation} (estimation view), we assume that every round of the ping-pong procedure is independent of the others, although does not have to be identical.}

\vl{\textbf{EPR pairs.} The main objective in the quantum memory network stage is using quantum memory in the presence of local gates. Therefore, for simplicity, we assume that any pair of nodes can generate a perfect EPR pair between them. This assumption is strictly speaking not necessary, but merely serves as an aid in understanding our test. In Sec. \ref{SEC:more_imperf} we show how to remove this assumption and how the noise associated with an EPR pair can be absorbed into the noise of the quantum memory.}

\vl{\textbf{State preparation and measurement.} For the same reasons as above, we also assume that any node can perfectly prepare a local qubit state and perfectly measure at the end of a protocol. In Sec. \ref{SEC:more_imperf} we also discuss how to relax this assumption.}

\vl{\textbf{Hilbert space dimension.} For the sake of clarity, throughout the rest of the manuscript we will assume that protocols run on a single qubit. We remark that the results we present generalize for any number $Q$ of qubits (for details see App. \ref{APP:SEC:q_qubit_ext}).}

\vl{\textbf{Device stability.}  In the estimation view Sec. \ref{sec:estimation} (in particular in Theorem \ref{THM:bound_protocols}) we assume that after the devices were tested with the ping-pong test, their behavior does not change. That is, the devices used during the test and in a $k$-round protocol are identical. Note that this can be understood as a consequence of the above IID assumption for the estimation view.}

\subsection{$k$-round protocols}\label{SEC:functionality}

We start with formally describing $k$-round protocols. 
A bipartite $k$-round protocol between any two nodes $A$ and $B$ consists of the following consecutive operations:
\begin{enumerate}
\item Local preparation \prep of a perfect qubit state $\ket{\psi}$ by node $A$.
\item Sending deterministically the local qubit from node $A$ to node $B$ and vice versa, using a quantum channel $\mathcal{E}_{A\rightarrow B} $. \vl{Note that the time $t_{\tn{send}}$ it takes to send a qubit (or a classical bit) from node $A$ to $B$ is upper-bounded by the distance between them and the transmission speed for the qubit carrier. For example, for optical qubits the transmission speed can be understood as the speed of photons in a fiber \cite{Wehner2018}.}
\item Storing the local qubit by nodes $A$ or $B$, denoted by $M_A$ and $M_B$ respectively. Storage of the qubit takes time $t_M$. 
\item Applying an arbitrary local operation by a node on the local qubit. We describe this operation by a gate $G_A \in \mathsf{G}$ and $G_B \in \mathsf{G}$, where $\mathsf{G}$ is an arbitrary set of gates, for example the single-qubit Clifford gates.
Executing a circuit of depth $z$ takes time $\ell_z$.
\item Perfect local measurement of the local qubit at the end of the protocol. The measurements are specified by operators $\Pi_A$ and $\Pi_B$ for nodes $A$ and $B$ respectively. 
\end{enumerate}
Steps 2. -- 4. are performed in rounds $j=1,...,k$ a total number of $k$ times. We call $k$ the depth of the protocol. Each round takes time $\Delta t = t_\tn{send} +  t_M + \ell_z$, so that $t_{j+1}-t_j = \Delta t$, for all $j$. Without loss of generality we assume here that the protocols always start at node $A$. Note that the parity of $j$ indicates at which node the single qubit is located, i.e., for odd $j$ the qubit is held (sent) by $A$ and for even $j$ -- by $B$. 
Therefore, we denote the local operations performed by $A$ or $B$ at a $j$-th round by simply putting $M_{j}$, $G_{j}$. In particular, in this notation $\mathcal{E}_j$ means that a qubit is sent by $A$ and received by $B$ for odd $j$ ($\mathcal{E}_j \equiv \mathcal{E}_{A_j \rightarrow B_j}$), and vice-versa for even $j$.  

\begin{definition}[$k$-round protocols]\label{DEF:class_of_protocols}
We define a $k$-round protocol as a map of the form $\Pi \circ {\mathcal{P}}^{k}\circ \tn{\prep}$, where:
\begin{itemize}[noitemsep,topsep=0pt]
\item \prep is a preparation of a local qubit $\ket{\psi}$ (Step 1).

\item ${\mathcal{P}}^{k}$ is a map describing $k$ rounds of local operations -- memories $M_j$ and gates $G_j$, as well as sending a qubit between $A$ and $B$ (Steps 2 -- 4),
\begin{align}
{\mathcal{P}}^{k} = \bigcirc_{j = 1}^{k} ~ G_{j} \circ M_{j} \circ \mathcal{E}_j.
\end{align}

\item $\Pi$ is a local measurement of all the local qubits (Step 5). Note that depending on the parity of $k$ the measurement is performed either on $A$'s or $B$'s side.
\end{itemize}
\end{definition}

\subsection{Test}\label{SEC:test}

In this section we describe our ping-pong test. The test is a simple instance of a $k$-round protocol as in Def. \ref{DEF:class_of_protocols}. As we will see in next sections, passing the test will allow us to draw conclusion about the whole class of $k$-round protocols.

Since our test will be later on viewed from two different angles, we introduce a node $V$ which will interact with the nodes $A$ and $B$. In the prover-verifier view, Sec.~\ref{sec:prover_verifier}, the node $V$ will act as a verifier. Whereas, in the estimation view, Sec.~\ref{sec:estimation}, the nodes $A$ and $B$ can take up the role of $V$. We choose the testing nodes $A$ and $B$ to be farthest apart in the network. \vl{For those nodes it is the hardest to fulfill the test, since they must account for the longest communication delays.}

\begin{figure}[h]
  \centering
  \includegraphics[width=0.45\linewidth]{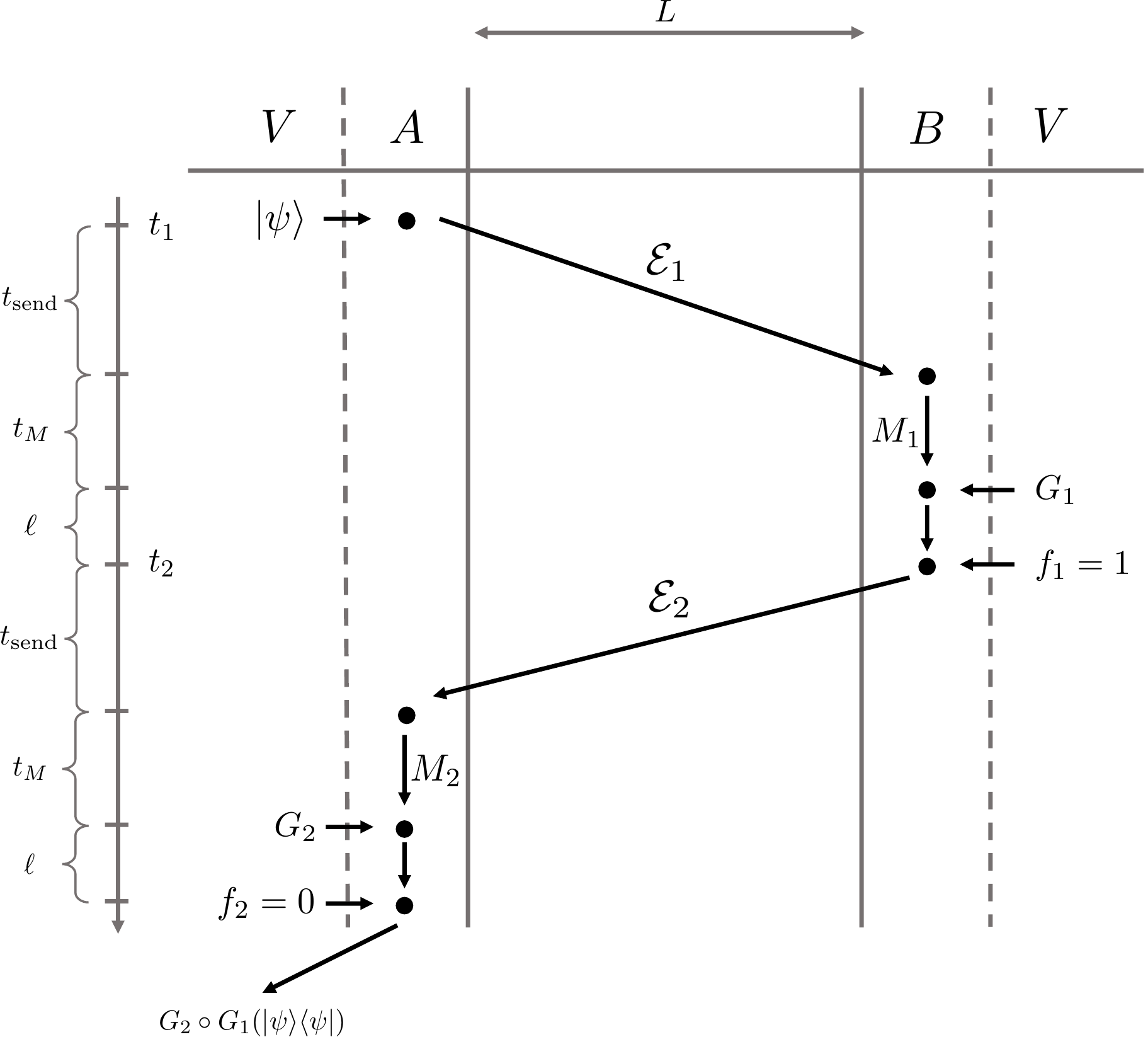}  
\caption{A schematic illustrating a single execution of depth $\kt=2$ of the ping-pong test, Test \ref{PROT:generic}.}\label{FIG:test_figure}
\end{figure}

\subsubsection{General ping-pong test}

In our test, the task of the nodes is to send (``ping-pong'') an unknown state an unknown number of times and at every ping-pong \emph{round} apply a quantum operation given by $V$, see Fig.~\ref{FIG:test_figure}. Additionally, at every round $V$ gives the nodes a challenge \vl{denoted by $f$} -- either to output the quantum state or continue the ping-pong. 
At the end of each \emph{execution} of the test, $i = 1,\dots,n$, the nodes output a state. $V$ measures this output and produces a single classical bit $v^i$: 1 means ``accept'' and 0 means ``reject'', see Test \ref{PROT:generic}. \vl{As stated before}, we assume that the nodes' operations are independent and identical across executions $i$ of the test. This implies that $v^i$ are independent and identically distributed (IID) random variables. We define a winning rate in such a game as the ratio of wins to the total number of executions:
\begin{align}\label{EQ:rate}
R = \frac{1}{n} \sum_{i = 1}^n v^i.
\end{align}

\begin{algorithm}[H]
\linespread{1.15}
\begin{minipage}{\linewidth}
\hspace{5pt}
\normalsize
\algrenewcommand\algorithmicloop{\textbf{repeat }}
  \caption{General ping-pong test (k,$\mathsf{G}$,$\mathcal{X}$)}\label{PROT:generic}
  \label{EPSA}
Fix maximum depth $k$, gate set $\mathsf{G}$ and set of states $\mathcal{X}$. Fix a total number of executions $n$.

\begin{algorithmic}[1]
\For{$i = 1,\dots,n$}
	\State $V$ chooses depth $\kt$ uniformly at random and constructs a challenge string $\vec{f}_\kt = 1 \dots 110 $ of length $\kt$
	\State $V$ samples independently $\kt$ gates from the set $\mathsf{G}$ and creates a sequence $\vec{g}_\kt = G_1 \dots G_\kt$
   \State $V$ samples a state $\ket{\psi} \in \mathcal{X}$ and distributes it to $A$
   \Comment{$t_1 = 0$}
   \For{$j = 1,..., \kt$} 
   \If{$j$ odd}
   		\State$A$ sends $\ket{\psi}$ to $B$ using $\mathcal{E}_j$
   		\Comment{ $t_j$}
   		\State $B$ stores the received state in memory $M_j$
   		\Comment{$t_j + t_\tn{send}$}
   		\State $V$ gives a classical description of $G_j$ to $B$
   		\Comment{$t_j + t_\tn{send} + t_M$}
   		\State $B$ applies $G_j$ to the state in the memory
   		\State $V$ distributes a challenge bit $f_j \in \{0,1 \}$ to $B$ according to the string $\vec{f}_\kt$
   		\Comment{$t_j + t_\tn{send} + t_M + \ell$}
   		\If{$f_j = 1$}
   			\State $j = j+1$
   			\State \textbf{continue}
   		\Else
   			\State $B$ outputs his state
   			\State $V$ measures $\{\Pi^\cmark_{\kt} =  \bigcirc_{j = 1}^{\kt} G_j (\ketbra{\psi}{\psi}) , \Pi^\xmark_{\kt} = \mathds{1} - \bigcirc_{j = 1}^{\kt} G_j(\ketbra{\psi}{\psi})\}$ 
   			\State $V$ decides on the value $v^i$ ('0' accept, '1' reject)
   			\State \textbf{break}
   		\EndIf
   	\ElsIf{$j$ even}	
   	   	\State$B$ sends $\ket{\psi}$ to $A$ using $\mathcal{E}_j$
   		\Comment{ $t_j$}
   		\State $A$ stores the received state in memory $M_j$
   		\Comment{$t_j + t_\tn{send}$}
   		\State $V$ gives a classical description of $G_j$ to $A$
   		\Comment{$t_j + t_\tn{send} + t_M$}
   		\State $A$ applies $G_j$ to the state in the memory
   		\State $V$ distributes a challenge bit $f_j \in \{0,1 \}$ to $A$ according to the string $\vec{f}_\kt$
   		\Comment{$t_j + t_\tn{send} + t_M + \ell$}
   		\If{$f_j = 1$}
   			\State $j = j+1$
   			\State \textbf{continue}
   		\Else
   			\State $A$ outputs her state
   			\State $V$ measures $\{\Pi^\cmark_{\kt} =  \bigcirc_{j = 1}^{\kt} G_j (\ketbra{\psi}{\psi}) , \Pi^\xmark_{\kt} = \mathds{1} - \bigcirc_{j = 1}^{\kt} G_j(\ketbra{\psi}{\psi})\}$ 
   			\State $V$ decides on the value $v^i$ ('0' accept, '1' reject)
   			\State \textbf{break}
   		\EndIf
   	\EndIf
   \EndFor
   \EndFor
   \end{algorithmic}
   \end{minipage}
\end{algorithm}

\newpage
The ping-pong test of depth $\kt$ for a sequence of chosen gates $\vec{g}_\kt = G_1,\dots,G_\kt$ can be associated with the following operator 
\begin{align}\label{EQ:teleport-based_test}
\mathcal{S}_\kt = \bigcirc_{j = 1}^{\kt} ~ G_{j} \circ M_{j} \circ \mathcal{E}_j.
\end{align} 
In a single execution of Test \ref{PROT:generic}, the test can succeed with a certain probability. For all executions $i$, we define such probability, conditioned on a specific input state $\ket{\psi}$, a fixed depth $\kt$ and a fixed sequence of gates $\vec{g}_\kt$ as 
\begin{align} \label{EQ:p_win}
p_{\cmark|\psi,\vec{g}_\kt, \kt} = \Tr [\mathcal{S}_\kt(\ketbra{\psi}{\psi})\cdot  \Pi^\cmark_{\kt}]
\end{align}
and similarly the probability of failure,
$
p_{\xmark|\psi,\vec{g}_\kt, \kt} = \Tr [\mathcal{S}_\kt(\ketbra{\psi}{\psi})\cdot  \Pi^\xmark_{\kt}]
$. Note that $p_{\cmark|\psi,\vec{g}_\kt, \kt}$ does not depend on the execution $i$, since we assume that executions are IID. Here $\{\Pi^\cmark_{\kt},\Pi^\xmark_{\kt}\}$ denotes the measurement performed by $V$ at the end of each execution $i$. We fix the figure of merit to be the average probability $\pavg$ that the nodes succeed ($v_i=1$) in the test. 
\begin{definition}[average probability of success for Test \ref{PROT:generic}]\label{DEF:prob_succ_test1}
The probability of success in the \vl{general} ping-pong test, Test \ref{PROT:teleport}, averaged over depths $\kt$, strings of gates $\vec{g}_\kt$ of length $\kt$, and states $\ket{\psi} \in \mathcal{X}$ 
is defined as
\begin{align}
\begin{split}
\pavg & = \frac{1}{n} \sum_{i} \frac{1}{k}  \sum_\kt \frac{1}{|\mathcal{X}|}\sum_{\psi} \frac{1}{|\mathsf{G}|^\kt} \sum_{\vec{c}_\kt} 
p_{\cmark|\psi,\vec{g}_\kt, \kt} \\
& = \frac{1}{k}  \sum_\kt \frac{1}{|\mathcal{X}|}\sum_{\psi} \frac{1}{|\mathsf{G}|^\kt} \sum_{\vec{c}_\kt} 
p_{\cmark|\psi,\vec{g}_\kt, \kt},
\end{split}
\end{align}
where the last equality holds due to the IID assumption. 
Here $k$ is the maximum depth of the test, $\mathcal{X}$ is the chosen set of states and $\mathsf{G}$ is the chosen set of gates. 
\end{definition}

\subsubsection{Teleportation-based ping-pong test}

In the case when $\mathcal{X}$ is the set of all single-qubit states, the average probability of success gives us an estimate on the average fidelity of the test, see Sec.~\ref{sec:estimation}. This would require sampling from $\mathcal{X}$ according to the Haar measure in the test. However, the same can be achieved more efficiently, by using sampling from the finite set of the six Pauli states $\mathsf{X}$. The reason for this is that $\mathsf{X}$ has a property of a 2-design, meaning that discrete uniform averaging over states \vl{(polynomials of degree 2)} in $\mathsf{X}$, reproduces the Haar average over the full state space. A similar observation holds for Haar sampling from a set of gates $\mathsf{G}$ in the case when $\mathsf{G}$ is a full unitary group. Then, it is enough to consider sampling from the Clifford group of single-qubit gates $\mathsf{Cliff}$ to reproduce the average probability of success. Note that this allows us to estimate the average fidelity of the test, even in the case when one is not able to implement the full unitary group. \vl{Lastly, we remark that any set of states and unitary gates with 2-design properties can be used in place of the Pauli states and Clifford gates.} For more details on 2-design properties of the above sets see App. \ref{APP:SEC:2designs}.

Therefore, we consider a more efficient version of the ping-pong test, Test \ref{PROT:teleport}. Motivated by the above and the fact that for a quantum network quantum channels between the nodes are realized by quantum teleportation, we choose:
\begin{enumerate}[noitemsep,topsep=0pt]
\item the set of states is the set of six Pauli eigenstates, $\ket{\psi} \in \mathsf{X}$ with a uniform probability distribution $\tfrac{1}{|\mathsf{X}|} = \tfrac{1}{6}$;
\item the set of gates is the Clifford set for a single qubit, $C_j \in \mathsf{Cliff}$ with a uniform probability distribution $\tfrac{1}{|\mathsf{Cliff}|}$;
\item sending a qubit from node $A$ to $B$ is done with perfect deterministic teleportation.
\end{enumerate}
We describe the teleportation-based ping-pong test with a triple $(k,\mathsf{Cliff},\mathsf{X})$.
Note that in this case the quantum channel at round $j$, $\mathcal{E}_j$, is equivalent to applying a quantum memory $M_j^\mathcal{T}$ to a half of the EPR pair by one of the provers. We can put $\tau = t_M + t_\tn{send}$, which is the time required to generate one maximally entangled state and send over a classical message from node $A$ to $B$. Hence, a teleportation-based ping-pong test of depth $\kt$ for a sequence of chosen Clifford gates $\vec{c}_\kt = C_1,\dots,C_\kt$ can be associated with the following operator 
\begin{align}\label{EQ:teleport-based_test}
\mathcal{T}_\kt = \bigcirc_{j = 1}^{\kt} ~ C_{j} \circ M_{j}^\mathcal{T}.
\end{align} 
For detailed mathematical description of the test, we refer the reader to App. \ref{APP:test_description}.

\begin{algorithm}[H]
\linespread{1.155}
\begin{minipage}{\linewidth}
\hspace{5pt}
\normalsize

\algrenewcommand\algorithmicloop{\textbf{repeat }}
  \caption{Teleportation-based ping-pong test $(k,\mathsf{Cliff},\mathsf{X})$}\label{PROT:teleport}
  \label{EPSA}

Fix maximum depth $k$, fix the gate set to Clifford set $\mathsf{Cliff}$ and the set of states to the set of six Pauli states $\mathsf{X}$. Fix the total number of executions $n$.

\begin{algorithmic}[1]
\For{$i = 1,\dots,n$}
	\State $V$ chooses depth $\kt$ and constructs a challenge string $\vec{f}_\kt = 1 \dots 110 $ of length $\kt$
	\State $V$ samples independently and uniformly at random $\kt$ gates from the set $\mathsf{Cliff}$ and creates a sequence $\vec{c}_\kt = C_1 \dots C_\kt$
   \State $V$ samples independently and uniformly at random a state $\ket{\psi} \in \mathsf{X}$ and distributes it to $A$
   \Comment{$t_1 = 0$}
   \For{$j = 1,..., \kt$} 
   \If{$j$ odd}
   		\State $A$ sends $\ket{\psi}$ to $B$ using deterministic teleportation
   		\Comment{ $t_j$}
   		\State $B$ stores half of his teleportation EPR pair in memory $M^\mathcal{T}_j$ for time $\tau$
   		\State $V$ gives a classical description of $C_j$ to $B$
   		\Comment{$t_j + \tau$}
   		\State $B$ applies $C_j$ to the state in the memory
   		\State $V$ distributes a challenge bit $f_j \in \{0,1 \}$ to $B$ according to the string $\vec{f}_\kt$
   		\Comment{$t_j + \tau + \ell$}
   		\If{$f_j = 1$}
   			\State Set $B=A$ and $A=B$
   			\State \textbf{continue}
   		\Else
   			\State $B$ outputs his state
   			\State $V$ measures $\{\Pi^\cmark_{\kt} =  \bigcirc_{j = 1}^{\kt} C_j (\ketbra{\psi}{\psi}) , \Pi^\xmark_{\kt} = \mathds{1} - \bigcirc_{j = 1}^{\kt} C_j(\ketbra{\psi}{\psi})\}$
   			   			\State $V$ decides on the value $v^i$ ('0' reject, '1' accept)
   			\State \textbf{break}
   		\EndIf
   	\ElsIf{$j$ even}
   		\State $B$ sends $\ket{\psi}$ to $A$ using deterministic teleportation
   		\Comment{ $t_j$}
   		\State $A$ stores half of her teleportation EPR pair in memory $M^\mathcal{T}_j$ for time $\tau$
   		\State $V$ gives a classical description of $C_j$ to $A$
   		\Comment{$t_j + \tau$}
   		\State $A$ applies $C_j$ to the state in the memory
   		\State $V$ distributes a challenge bit $f_j \in \{0,1 \}$ to $A$ according to the string $\vec{f}_\kt$
   		\Comment{$t_j + \tau + \ell$}
   		\If{$f_j = 1$}
   			\State \textbf{continue}
   		\Else
   			\State $A$ outputs her state
   			\State $V$ measures $\{\Pi^\cmark_{\kt} =  \bigcirc_{j = 1}^{\kt} C_j (\ketbra{\psi}{\psi}) , \Pi^\xmark_{\kt} = \mathds{1} - \bigcirc_{j = 1}^{\kt} C_j(\ketbra{\psi}{\psi})\}$
   			   			\State $V$ decides on the value $v^i$ ('0' reject, '1' accept)
   			\State \textbf{break}
   		\EndIf
   	\EndIf
   \EndFor
   \EndFor
   \end{algorithmic}
   \end{minipage}
\end{algorithm}

By using Def.~\ref{DEF:prob_succ_test1} with the set of Pauli states $\mathsf{X}$ and the set of Clifford gates $\mathsf{Cliff}$, the average probability of success for the teleportation-based ping-pong, Test \ref{PROT:teleport}, is
\begin{align}\label{EQ:avg_prob}
\begin{split}
\pavg & = \frac{1}{k}  \sum_\kt \frac{1}{|\mathsf{X}|}\sum_{\psi} \frac{1}{|\mathsf{Cliff}|^\kt} \sum_{\vec{c}_\kt} 
p_{\cmark|\psi,\vec{c}_\kt, \kt}.
\end{split}
\end{align}

Note that in Test \ref{PROT:teleport} the sampling of depths, gates and states is done uniformly at random. Using the definition of the expected value and the IID assumption ($\forall i,j ~\mathds{E}[v_i] = \mathds{E}[v_j]$), we can write that the winning rate has the expected value $\mathds{E}[R] = \frac{1}{k}  \sum_\kt \frac{1}{|\mathsf{X}|}\sum_{\psi} \frac{1}{|\mathsf{Cliff}|^\kt} \sum_{\vec{c}_\kt} p_{\cmark|\psi,\vec{c}_\kt, \kt}\cdot~1  + p_{\xmark|\psi,\vec{c}_\kt, \kt}\cdot 0$. 
\begin{lemma}
The expected value of the winning rate $R$ in Test \ref{PROT:teleport}, Eq. \eqref{EQ:rate}, is equal to the average probability of success $\pavg$,
\begin{align}
\mathds{E}[R] = \pavg.
\end{align}
\end{lemma}

\begin{corollary} [finite statistics]
The probability that the winning rate $R$ differs from the average probability of success $\pavg$ by more than $\epsilon$ is exponentially small in $\epsilon$,
\begin{align}\label{EQ:hoeffding}
\tn{Pr} \left[ | R-\pavg | \leq \epsilon \right] \geq 1 - 2e^{-2n\epsilon^2}.
\end{align}
Furthermore, let us set $\delta = 2e^{-2n\epsilon^2}$. If one fixes confidence $\delta$ and accuracy $\epsilon$, then the minimum number of rounds $n$ necessary to attain these parameters is given by
$
n \geq \frac{\ln(2\delta^{-1})}{2\epsilon^2}.
$
\end{corollary}

\section{Prover-verifier view} \label{sec:prover_verifier}

In this section we interpret our test, Test \ref{PROT:teleport} in the prover-verifier view. 
Specifically, we view our test as an interactive game played between a verifier $V$ (trusted third party), and two provers (the nodes $A$ and $B$)~\cite{Vidick2016}.
An interactive game is a situation where provers exchange a fixed-sized quantum register with the verifier $n$ times. The verifier is honest and wants to verify a certain statement, operating according to a defined set of rules.
However, potentially dishonest provers optimize towards a strategy that causes a verifier to output 1 (accept). 
We further assume assume a standard scenario, where the provers agree on their strategy prior to the beginning of the test and they do not communicate to readjust it during the execution, see Def.~\ref{DEF:m-cheating}. 
In contrast to the interactive proof literature, in our framework we consider finitely many test executions and therefore, we can also make non-asymptotic statistical statements.

In this view, performing Test \ref{PROT:teleport} allows us to certify that the provers have capabilities to perform $k$-round protocols. Indeed, if the provers follow the test then they can convince the verifier that they do so and achieve a high average probability of success. On the other hand, if the provers do not follow the test they cannot achieve a high probability of success and the verifier detects this behavior with high probability.
Formally, we require that the test satisfies:
\begin{itemize}[noitemsep,topsep=0pt]
\item \emph{completeness} -- if the provers are able to execute protocols that are certified by the test then they succeed in a game against the verifier, i.e. achieve a winning rate above a certain winning threshold $t$, $R>t$, see Eq.~\eqref{EQ:rate};
\item \emph{soundness} -- if the provers are not able to execute protocols certified by the test, then they can only achieve a winning rate $R\leq t$.
\end{itemize}

\subsection{Sending channel}
Let us now introduce a framework that formalizes what we mean by a round of a quantum communication. 
Whereas numerous schemes to describe local operations exist
\cite{Chuang1997,Poyatos1997,Merkel2013,BlumeKohout2013,Emerson2005,Knill2008,Magesan2011} it is not clear how to certify a round of quantum communication. To achieve this, we will assume that the provers are not honest, and might therefore employ an arbitrary strategy leading to a high probability of success. In particular, they might even try to not use a communication channel at all in some rounds of the protocol.  As a consequence, we have to specify what we mean by a round of communication.

For sending classical bits one typically considers the following scenario: $A$ chooses a random bit $b_{A_0} \in_R \{0,1 \}$ at time $t_0$ and wishes to send it to $B$. We then say that the nodes used a classical channel $\mathcal{E}_{cl}: A \rightarrow
B$ if the probability at time $t_1$ that $B$'s bit is the same as $A$'s, is equal to 1, $\tn{Pr}[b_{B_1 }= b_{A_0} ]=1$. 
In analogy, we could say that quantum communication through a quantum channel $\mathcal{E}: A_0 \rightarrow B_1$ occurred if at time $t_0$ a quantum state $\ket{\psi}_{A_0}$ was input on node $A$ and at time $t_1$ it appeared on node $B$ with probability $1$, $\tn{Pr}[\rho_{B_1} = \ketbra{\psi}_{A_0} ]=1$. 

Note that in the classical case, we can prove that the channel was used to send information about the bit only for one round, by giving a uniformly random bit to $A$ and ask $B$ to guess it. Indeed if $B$ guesses it with probability higher than $1/2$ then some information must have traveled from $A$ to $B$. Given a single bit as an input, one cannot generalize that to many rounds with a ``ping-pong'' type of protocol like Test \ref{PROT:teleport}. This is due to the fact that before $A$ sends information to $B$ in the first round, she can keep a copy of the bit. However, this issue can be avoided in a quantum setting due to the no-cloning theorem~\cite{Wootters1982}. Indeed, if $A$ gets a random unknown state and $B$ is able to output the exact same state (with probability $1$), then not only did all the (quantum) information about the state traveled from $A$ to $B$, but also $A$ could not have kept any information about the state to herself (see Thm.~\ref{THM:corr_uniq}). 

While the above definition provides a good intuition of what is going on, it becomes impractical when states do not have a unit probability of being transmitted through a channel (which in relation to our test means $t<1$). In such a scenario, classically, we can say that the nodes used a classical channel $\mathcal{E}_{cl}: A_0 \rightarrow B_1$ if the probability of correctly identifying $A$'s input bit on $B$'s side increased in time, $\tn{Pr}[\tn{out}_{B_1} = b_{A_0}] > \tn{Pr}[\tn{out}_{B_0} = b_{A_0}]$. This implies that some information about the bit must have been transferred from $A$ to $B$, see Fig. \ref{FIG:sendingchannel}. 
Our definition of quantum communication is, therefore, a generalization of the above to the quantum case. We say that quantum communication $\mathcal{E}: A_0 \rightarrow B_1$ occurred if the probability of correctly outputting $A$'s input quantum state on $B$'s side increased in time, see Def.~\ref{DEF:sending_channel_formal}. 

\begin{figure}[h]
\begin{minipage}{0.8\linewidth}
\centering 
\includegraphics[width=\linewidth]{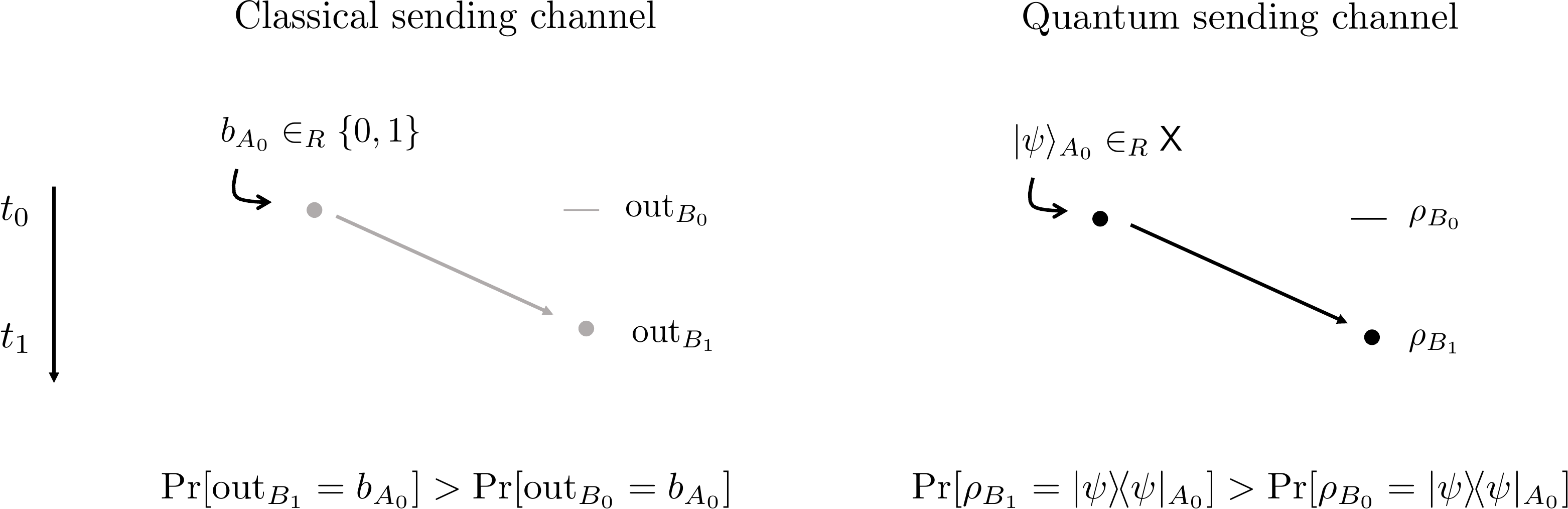}
\caption{Informal representation of a sending channel in a classical and quantum case.}\label{FIG:sendingchannel}
\end{minipage}
\end{figure}

In words, we say that a sending channel was used by the nodes if the fidelity averaged over all states, and optimized over all operations $\Gamma$ that the nodes can locally do, increased from instant $t_0$ to $t_1$.
Note that the above definition implies that any communication, quantum or classical, which increases fidelity of the state is considered a sending channel. As an example consider the following strategy. Node $A$ receives an unknown state from the verifier, measures it in the standard basis and sends the measurement outcome to $B$. Without loss of generality, let this measurement outcome be 0. Before receiving $A$'s measurement outcome, $B$ has average probability $\tfrac{1}{2}$ of correctly passing verifier's test. However, after receiving $A$'s measurement outcome, $B$ can locally prepare $\ket{0}$ state which increases the average probability of correctly identifying verifier's state to $\tfrac{2}{3}$. Therefore, there exists a purely classical strategy which satisfies our definition. As a consequence, we say that whenever the nodes do not use a sending channel $\mathcal{E}$, no communication (quantum or classical) occurred between them.

\begin{definition}[sending channel]\label{DEF:sending_channel_formal}
A channel $\mathcal{E}_{A_0 \rightarrow B_1}$ is a sending channel if there exists a CPTP map $\Omega_{A_0 \rightarrow A_0B_0}$ such that $\forall \ket{\psi}_{A_0}$  it creates a state $\rho_{A_0B_0}^{\psi} = \Omega(\ketbra{\psi}_{A_0})$ and
\begin{align}
\sup_{\Gamma_{B_0B_1}} \int \tn{d}\psi~ \Tr \left[ \Gamma_{B_0B_1}\left(\rho_{B_0B_1}^{\psi}\right) \cdot \ketbra{\psi}_{A_0} \right] > \sup_{\Gamma_{B_0}} \int \tn{d}\psi~ \Tr \left[ \Gamma_{B_0}\left(\rho_{A_0B_0}^{\psi}\right) \cdot \ketbra{\psi}_{A_0} \right],
\end{align}
where $\rho_{B_0B_1}^{\psi} = \mathcal{E}_{A_0 \rightarrow B_1}(\rho_{A_0B_0}^{\psi})$, $\Gamma$ is a CPTP map which traces out additional registers of $A$ and $B$ and outputs a qubit state. 
In particular, if 
\begin{align}
\sup_{\Gamma_{B_0B_1}} \int \tn{d}\psi~ \Tr \left[ \Gamma_{B_0B_1}\left(\rho_{B_0B_1}^{\psi}\right) \cdot \ketbra{\psi}_{A_0} \right] = 1 \quad \tn{and}\quad  \sup_{\Gamma_{B_0}} \int \tn{d}\psi~ \Tr \left[ \Gamma_{B_0}\left(\rho_{A_0B_0}^{\psi}\right)
 \cdot \ketbra{\psi}_{A_0} \right] = \frac{1}{2} 
\end{align}then we talk about an \emph{exact sending} channel.
\end{definition}

\begin{definition}[$m$-cheating]\label{DEF:m-cheating}
Provers $A$ and $B$ are $m$-cheating if their cheating strategy uses a sending channel $\mathcal{E}$ between them at most $m$ times. We assume that the provers choose a strategy -- in which round they use a sending channel and in which they do not --  prior to the beginning of the test.
\end{definition}

\subsection{Exact completeness and soundness}

To investigate the power of Test \ref{PROT:teleport} in verifying capabilities of the network, we first consider an instructive case when $\pavg = 1$. If the nodes are able to perfectly execute the test then they succeed with a unit probability, trivially satisfying the completeness, see Thm. \ref{THM:exact_corr}. On the other hand, if we demand that the nodes always succeed in the game, we can ask the question whether the nodes have the ability to perfectly execute protocols that have the form of Test \ref{PROT:teleport}, i.e., whether the test is sound. We answer this question positively in Thm. \ref{THM:corr_uniq} below.

\begin{theorem}[exact completeness]\label{THM:exact_corr}
If the provers are honest and they are able to perfectly execute Test \ref{PROT:teleport} then they succeed $\pavg = 1$.  
\end{theorem}

\begin{theorem}[exact soundness]\label{THM:corr_uniq}
{If the provers win the test with $\pavg = 1$ then they must be able to perfectly execute Test \ref{PROT:teleport} and they use an exact sending channel $\mathcal{E}$ between them $k$ times.}
\end{theorem}
\begin{proof}[Idea of the proof]
To prove the theorem, we argue that $\pavg = 1$ implies that the probability of winning $p_{\cmark|\psi,\vec{c}_\kt, \kt} $ for all states, all Clifford gates and all depths should be 1 (in particular, this implies that the provers are able to apply the required Clifford gates on the input state). Therefore, the average fidelity at every depth $\kt$ should be 1. 
That is, if at step $\kt-1$ $A$ has fidelity 1 it means that the state on A is pure, and by a purifying argument, $B$'s average fidelity at step $\kt-1$ must be $1/2$. At step $\kt$ $B$ has fidelity 1, which means that whatever channel $A$ and $B$ have used between step $\kt-1$ and $\kt$, it must be an exact sending channel (see Def. \ref{DEF:sending_channel_formal}). For more details see App.~\ref{APP:SUBSEC:exact_corr_sound}.
\end{proof} 

Note that in practice we are only able to observe the winning rate $R$ and, due to the finite statistics of our test, we cannot certify $\pavg = 1$.

\subsection{Completeness and soundness}

Therefore, let us explore the implications of Test \ref{PROT:teleport}, given that the winning rate $R>t$ is observed.  If the provers are honest and their devices are sufficiently good, their winning rate should be larger than threshold $t$ with high probability. More specifically, let memories and gates at every round $j$ be described in terms of the average fidelity. Assume that the quality of memory and gates is the same at every round $j$, i.e. for all $j$, the average fidelity $\bar{\mu} = \int \tn{d} \psi \Tr[ C_j \circ M^\mathcal{T}_j(\ketbra{\psi}) \cdot C_j(\ketbra{\psi})].$ Below we show that for honest provers, a certain fidelity of operations implies a bound on the winning rate. \vl{In order to satisfy both completeness and soundness we choose the winning threshold $t > \tfrac{5}{6}$, since the Test \ref{PROT:teleport} does not lead to any conclusion in the case when $t \leq \tfrac{5}{6}$, see Thm.~\ref{THM:eps-sound}.} Let $h_k(\bar{\mu}) = \frac{\bar{\mu}(\bar{\mu}^k-1)}{k(\bar{\mu}-1)}$.

\begin{figure}[t]
\begin{minipage}{0.6\linewidth}
\centering 
\includegraphics[width=\linewidth]{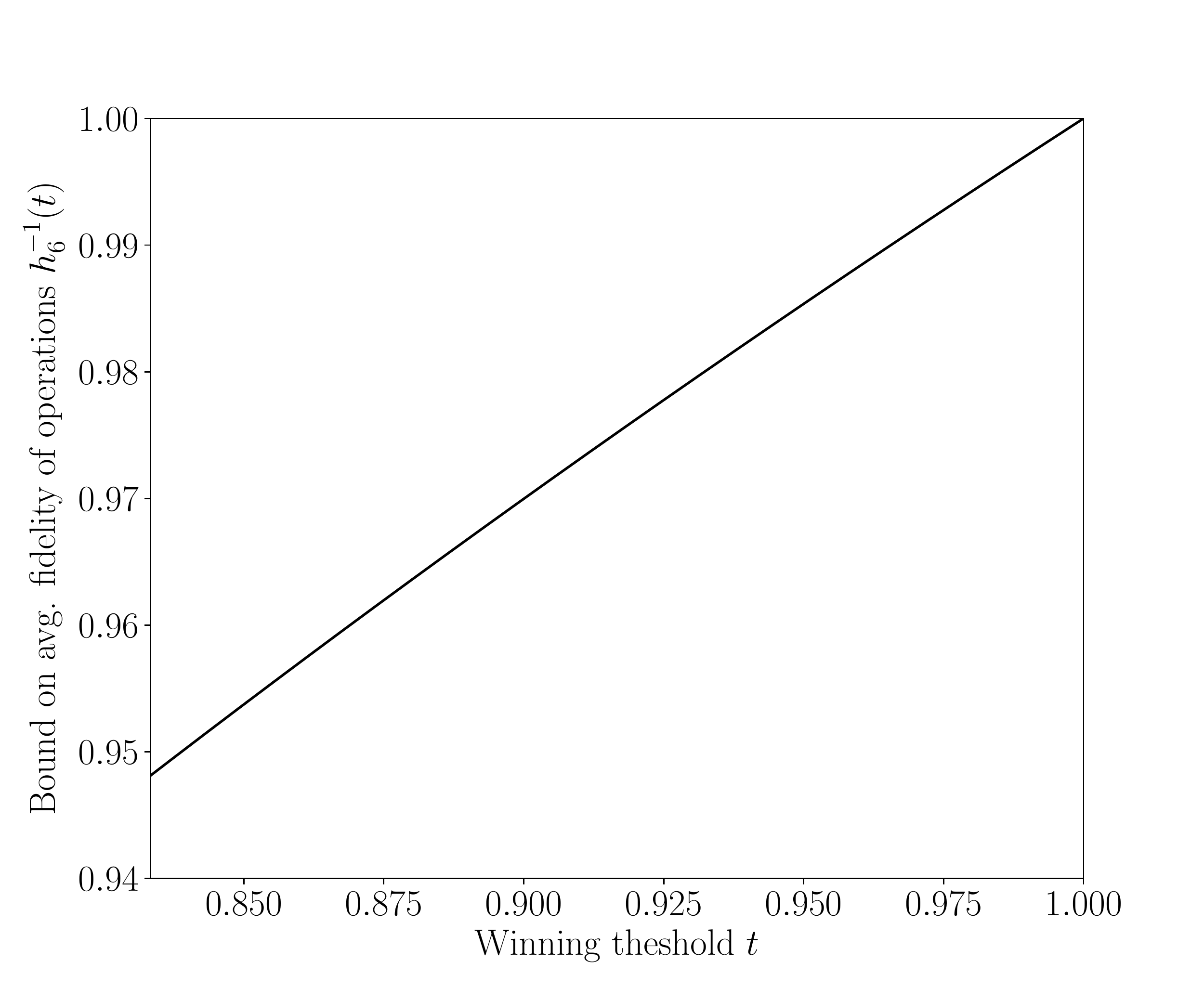}
\caption{\vl{The average fidelity of individual operations $\bar{\mu}$ as a function of the winning threshold $t$ (see Thm. \ref{THM:eps-corr} completeness). The plot shows the inverse of the function $h_k$, i.e. $h_k^{-1}(t)$ for $k=6$ and relevant values of $t$.}}\label{FIG:bound_h}
\end{minipage}
\end{figure}

\begin{theorem}[completeness]\label{THM:eps-corr}
If provers are honest and their individual operations satisfy $\bar{\mu} \geq h_k^{-1}(t) + \epsilon$, then the winning rate $R$ in Test \ref{PROT:teleport} is bounded by $R \geq t$ with probability at least $1 - e^{-n\epsilon^2}$, where \vl{$t \in(\tfrac{5}{6},1]$} is a winning threshold and $\epsilon $ is given by Eq. \eqref{EQ:hoeffding}.
\end{theorem}

\begin{proof}[Idea of the proof]
Using 2-design properties of the set of states $\mathsf{X}$ and the set of gates $\mathsf{Cliff}$, we show that in the regime where fidelity $\bar{\mu}$ is the same for every round $j$, we can express the average probability of success as a sum of powers of $\bar{\mu}$. That is, $\pavg = \frac{1}{k}\sum_{\kt = 1}^k \bar{\mu}^\kt = \frac{\bar{\mu}(\bar{\mu}^k-1)}{k(\bar{\mu}-1)} = h_k(\bar{\mu})$, see App. \ref{APP:SUBSEC:corr_sound} for details. Since we want the winning rate $R$ to be higher than the threshold $t$, we invert the function $h_k$ to obtain a bound on the fidelity of the devices $\bar{\mu}$. We plot the inverse $h^{-1}_k(t)$ in Fig. \ref{FIG:bound_h} \vl{for $t \in(\tfrac{5}{6},1]$.}
\end{proof}

Moreover, we can ask whether the converse of the above statement is true, i.e. whether a certain winning rate $R>t$ implies something about Test \ref{PROT:teleport}. When the provers are honest, we can reverse the completeness statement obtaining a bound on the quality of their devices. If the provers are dishonest ($m$-cheating) then they do not have to exactly follow the test. However, in this case we will show that the winning rate $R$ allows us to certify that the provers used a sending channel (Def. \ref{DEF:sending_channel_formal}) a certain number of times.

\begin{theorem}[soundness]\label{THM:eps-sound}
\vl{If the provers are $m$-cheating then the winning rate in Test \ref{PROT:teleport} is bounded by $R\leq \frac{1}{k} \left(m + \frac{5}{6}(k-m)\right) + \epsilon$, with probability exponentially close to 1, i.e. at least $1 - e^{-n\epsilon^2}$, where $\epsilon \in (0,1)$.}
\end{theorem}

\begin{proof}[Idea of the proof.]
\vl{In the case when the provers are $m$-cheating they can agree on a cheating strategy which uses a quantum channel $\mathcal{E}$ between them at most $m$ times, see Def. \ref{DEF:m-cheating}.} To prove soundness in this case we look at the average probability of winning for $A$ and $B$ at time steps $\kt-1$ and $\kt$. In App. \ref{APP:SUBSEC:corr_sound} we argue that whenever the provers use the channel $\mathcal{E}$, this probability is bounded by 1. On the other hand, whenever they do not use the channel and no communication occurred, we argue that the average probability of winning at both time steps is bounded by $\frac{5}{6}$ which is the bound provided by the approximate cloning theorem \cite{Gisin1997}. Since the nodes use the channel $\mathcal{E}$ at least $m$ times, their overall average probability of winning $\pavg$ is bounded by $\frac{1}{k} \left(m + \frac{5}{6}(k-m)\right)$. 
\end{proof}

The above theorem implies that in the situation when we do not trust the nodes, the higher $m$ we would like to certify, the higher the winning threshold should be. 
Indeed, for $\pavg \geq t $ we obtain $m \geq k(6t - 5) $. If we now set $ t = 1- \eta$, for some small $\eta$, then $m \geq k - 6k\eta$. For $m\sim k$, one should set at least $\eta = \mathcal{O}(k^{-1})$.

\vspace{1em}
\noindent \textbf{Remark.} Note that in Thm.~\ref{THM:corr_uniq} we are able to fully certify the action of the provers, even if they are not trusted. In particular, we know that they have perfectly sent the state to each other $k$ times. On the other hand, Thm.~\ref{THM:eps-sound} only certifies the use of some quantum or classical channel regardless of its quality. In particular, in the limit where $\pavg =1$, Thm.~\ref{THM:eps-sound} show that $m=k$ sending channels have been used, but we cannot explicitly certify the quality of the channel. However, the exact soundness statement, Thm.~\ref{THM:corr_uniq}, suggests that even in the imperfect case, the test should be able to certify the quality of each individual operation used by the provers.

\section{Estimation view}\label{sec:estimation}

In this section we interpret our test in the context of estimation in order to obtain measures of confidence in the nodes' ability to perform the test. We assume that the nodes $A$ and $B$ are honest and follow the protocol. 
Specifically, we use the winning rate $R$ in the teleportation-based ping-pong test, as a figure of merit to estimate the quality of the network. We then provide a consistency check which allows us to compare this to the quality one would expect from combining the individual devices. Furthermore, we use the statistics of the test to estimate the performance of $k$-round protocols.

Throughout this section we will use a tilde to denote noisy counterparts of operations, for example $\tilde{\mathcal{T}}_\kt$ will denote a noisy realization of the $\kt$-round teleportation-based ping-pong test $\mathcal{T}_\kt$, Test \ref{PROT:teleport}. 

\subsection{Preliminaries} \label{SUBSEC:math_tools}

In this section we introduce mathematical tools which will be useful for \emph{(i)} checking whether the test is consistent when the honest nodes use devices of a certain quality, Sec.~\ref{SUBSEC:consistency_check}, and \emph{(ii)} drawing conclusions about the performance of $k$-round protocols, Sec.~\ref{SUBSEC:performance} and \ref{SUBSEC:simulatedresults}.

We describe the quality of individual devices with a noise model.
Specifically, we assume that the individual operations used in the test, i.e. memories $M_j$ and gates $C_j$, have been tested individually for each round $j$, to obtain an estimate on their performance. More formally, let the quality of a noisy gate $\tilde{C}_j$ at round $j$, be described with the average fidelity, $\bar{F}(\tilde{C}_j) = \int \tn{d} \psi \Tr\left[ \tilde{C}_j (\ketbra{\psi}{\psi}) \cdot C_j(\ketbra{\psi}{\psi}) \right]$, for all $j=1,\dots,k$. Furthermore, let the average fidelity have an empirical estimate $r_{C_j}$, which is known with certain precision~\cite{Hoeffding1963}, such that
\begin{align}\label{EQ:hoeffding_C}
\tn{Pr}\left[ |r_{C_j} - \bar{F}(\tilde{C}_j)| \leq \epsilon_{C_j} \right] \geq 1 - \delta_{C_j},
\end{align} where $\delta_{C_j} = 2e^{-2n_{C_j}\epsilon_{C_j}^2}$. Here $n_{C_j}$ is the number of repetitions with which the estimate $r_{C_j}$ was obtained. Similarly, for  $\tilde{M}^\mathcal{T}_j$ a noisy quantum memory at round $j$,  average fidelity is $\bar{F}(\tilde{M}^\mathcal{T}_j) = \int \tn{d} \psi \Tr\left[ \tilde{M}^\mathcal{T}_j (\ketbra{\psi}{\psi}) \cdot M_j^\mathcal{T}(\ketbra{\psi}{\psi}) \right]$. This average fidelity has an empirical estimate $r_{M^\mathcal{T}_j}$ and a precision bound 
\begin{align}\label{EQ:hoeffding_M}
\tn{Pr}\left[ |r_{M^\mathcal{T}_j} - \bar{F}(\tilde{M}^\mathcal{T}_j)| \leq\epsilon_{M_j^\mathcal{T}} \right] \geq  1- \delta_{M_j^\mathcal{T}},
\end{align}where $\delta_{M_j^\mathcal{T}} = 2e^{-2n_{M_j^\mathcal{T}}\epsilon_{M_j^\mathcal{T}}^2}$. Furthermore, we assume that the nodes can locally and perfectly prepare and measure a quantum state.

\vspace{1em}
The teleportation-based ping-pong test, Test \ref{PROT:teleport}, is performed the total of $n$ times. Note that one can easily record which executions $i$ were performed for depths $\kt$, states $\psi$ and strings of Clifford gates $\vec{c}_\kt$. Then, in analogy to Eq. \eqref{EQ:rate}, we can define the winning rate for a \emph{fixed} depth $\kt$ and string $\vec{c}_\kt$,
\begin{align}
R_{\vec{c}_\kt,\kt} = \frac{1}{n_{\vec{c}_\kt,\kt} } \sum_{i} v^i_{\vec{c}_\kt,\kt} ,
\end{align}
where ${n_{\vec{c}_\kt,\kt} }$ is a total number of executions for fixed $\kt$ and $\vec{c}_\kt$, and $v^i_{\vec{c}_\kt,\kt}$ is a corresponding random variable assuming values 0 and 1 for 'lose' and 'win' events respectively. Analogously, we can record which executions correspond to a fixed depth $\kt$ only. We define
\begin{align}
R_\kt = \frac{1}{n_\kt} \sum_i v^i_\kt
\end{align}
as the winning rate for a fixed $\kt$. Here $n_\kt$ is a total number of executions for depth $\kt$ and $v^i_\kt$ is a corresponding random variable recording the wins in the test.

Now we will relate the above winning rates to the measures of quality of the test. Intuitively, the higher the winning rate the better the test performs and the less noise is present in the setup. In the remaining part of this section we make that statement rigorous.

\begin{lemma}\label{LEM:avgfid_expected}
Let the average fidelity of a noisy realization of Test \ref{PROT:teleport}, $\tilde{\mathcal{T}}_\kt$, for a fixed depth $\kt$ and a fixed string of Clifford gates $\vec{c}_\kt$ be defined as $\bar{F}_{\vec{c}_\kt,\kt}(\tilde{\mathcal{T}}_\kt) = \int \tn{d}\psi \Tr [\tilde{\mathcal{T}}_\kt(\ketbra{\psi}{\psi})\cdot  \Pi^\cmark_{\kt}]$, where $\mathcal{T}_\kt$ is defined as in Eq. \ref{EQ:teleport-based_test}.
The expected value of the winning rate $R_{\vec{c}_\kt,\kt}$ over the set of states $\mathsf{X}$, is equal to the average fidelity of the test $\tilde{\mathcal{T}}_\kt$,
\begin{align}
\mathds{E}\left[ R_{\vec{c}_\kt,\kt}\right]_\mathsf{X} = \bar{F}_{\vec{c}_\kt,\kt}(\tilde{\mathcal{T}}_\kt). 
\end{align}
\end{lemma}

\begin{proof}[Idea of the proof.]
The first step of the proof is to notice that the expected value of the variable $v^i_{\vec{c}_\kt,\kt}$ is the probability of success in a single round averaged over all the states in $\mathsf{X}$,
\begin{align}
\mathds{E}\left[ v^i_{\vec{c}_\kt,\kt} \right]_\mathsf{X} & = \frac{1}{|\mathsf{X}|} \sum_{\psi\in \mathsf{X}} \left( p_{\cmark|\psi,\vec{c}_\kt, \kt} \cdot 1 + p_{\xmark|\psi,\vec{c}_\kt, \kt} \cdot 0 \right)  = \frac{1}{|\mathsf{X}|} \sum_{\psi\in \mathsf{X}}\Tr [\tilde{\mathcal{T}}_\kt(\ketbra{\psi}{\psi})\cdot  \Pi^\cmark_{\kt}]
\end{align}
The second step is based on relating the above quantity to the average fidelity. Here the key idea is to observe that the expression under the trace contains only polynomials of degree 2 in $\ketbra{\psi}$. Therefore one can use the 2-design properties of the set $\mathsf{X}$ to equate the discrete averaging over the six Pauli states to the continuous Haar averaging over the whole state space in average fidelity. The details of the proof can be found in App. \ref{APP:SUBSEC:proof_avgfid_expected}.
\end{proof}
The above lemma has a simple useful corollary, namely, that the average fidelity and the winning rate $R_{\vec{c}_\kt,\kt}$ can be related through the Hoeffding inequality,
\begin{align}\label{EQ:hoeffding_avgfid}
\tn{Pr}\left[ |R_{\vec{c}_\kt,\kt} - \bar{F}_{\vec{c}_\kt,\kt}(\tilde{\mathcal{T}}_\kt)| \geq \epsilon_{\vec{c}_\kt,\kt} \right] \geq 1 - 2e^{-2n_{\vec{c}_\kt,\kt}\epsilon_{\vec{c}_\kt,\kt}^2}.
\end{align}

Before we make a similar connection for the rate $R_{\kt}$, let us define a useful quantity.

\begin{definition}[double-averaged fidelity]
Let $\bar{F}_{\vec{c}_\kt,\kt}(\tilde{\mathcal{T}}_\kt)$ be the average fidelity of a the teleportation-based ping-pong test, Test 2, defined for a fixed depth $\kt$ and a fixed sting of Clifford gates $\vec{c}_\kt$. We define the quantity
\begin{align}
\bar{\bar{F}}_{\kt}(\tilde{\mathcal{T}}_\kt) := \int \tn{d}C_1 ... \int \tn{d}C_\kt \bar{F}_{\vec{c}_\kt,\kt}(\tilde{\mathcal{T}}_\kt).
\end{align}
as double-averaged fidelity. The averaging for every gate $C_j$ is taken according to the Haar measure. 
\end{definition}

\begin{lemma}\label{LEM:doubleavgfid_expected}
The expected value of the winning rate $R_{\kt}$ in Test \ref{PROT:teleport}, for a fixed depth $\kt$, taken over the set of states $\mathsf{X}$ and set of Clifford gates, is equal to the double-averaged fidelity of the test $\tilde{\mathcal{T}}_\kt$,
\begin{align} \label{EQ:doubleavgfid_expected}
\mathds{E}\left[ R_{\kt}\right]_{\mathsf{X},\mathsf{Cliff}} = \bar{\bar{F}}_{\kt}(\tilde{\mathcal{T}}_\kt).
\end{align}
\end{lemma}
The intuition behind the above lemma is that discrete averaging in $\mathds{E}\left[ R_{\kt}\right]_{\mathsf{X},\mathsf{Cliff}}$ over the Clifford gates is equal to the continuous averaging in the definition of $\bar{\bar{F}}_\kt(\tilde{\mathcal{T}}_\kt)$. This statement follows from the unitary 2-design properties of the Clifford set, see App. \ref{APP:SUBSEC:proof_doubleavgfid_expected} for details.

Finally, the probability that the empirical data $R_\kt$ differs from double-averaged fidelity by more than $\epsilon_{\kt}$ is bounded by the Hoeffding inequality,
\begin{align}\label{EQ:hoeffding_doubleavg}
\tn{Pr}\left[ |R_\kt - \bar{\bar{F}}_\kt(\tilde{\mathcal{T}}_\kt)| \leq \epsilon_\kt \right] \geq 1 - 2e^{-2n_\kt \epsilon^2_{\kt} }.
\end{align}

\subsection{Consistency check}\label{SUBSEC:consistency_check}
\vl{In the following we demonstrate how to use the winning rates defined above to check for consistency, i.e. that devices with certain fidelities were used {\em together} in Test \ref{PROT:teleport}. Specifically, we provide a relation between the quality of the test in terms of $R_{\vec{c}_\kt,\kt}$ and what one may expect given individual devices with estimates of average fidelities $r_{M^\mathcal{T}_j}$ and $r_{C_j}$. Not satisfying this consistency-check relation implies that there is an internal contradiction in the reported values of individual average fidelities and observed rate $R_{\vec{c}_\kt,\kt}$.}

\begin{theorem}[consistency check]\label{THM:inconsistency_check} Let $r_{M^\mathcal{T}_j}$ and $r_{C_j}$, $j=1,\dots,k$, be empirical estimates of the average fidelity of all individual memories and gates respectively. Moreover let $R_{\vec{c}_\kt,\kt}$ be an empirical estimate of the average fidelity of the teleportation-based ping-pong test, Test \ref{PROT:teleport}, for a fixed depth $\kt$ and a fixed string of Clifford gates $\vec{c}_\kt$. Devices with estimates $r_{M^\mathcal{T}_j}$ and $r_{C_j}$ were used together in the test $\tilde{\mathcal{T}}_\kt$ if the following inequality is satisfied~\cite{Dugas2016},
\begin{widetext}
\begin{align}\label{EQ:bound_incons}
R_{\vec{c}_\kt,\kt} \geq \frac{2 \cos^2 \left( \sum_{j=1}^{\kt }
\acos \sqrt{\frac{3r_{M^\mathcal{T}_j}-1}{2}} + \acos \sqrt{\frac{3r_{C_j}-1}{2}}
 \right) + 1}{3} - \epsilon_{\vec{c}_\kt,\kt} 
\end{align}
\end{widetext}
The bound holds for any $2\kt$ quantum channels such that $\sum_{j=1}^{\kt }
\acos \sqrt{\frac{3r_{M^\mathcal{T}_j}-1}{2}} + \acos \sqrt{\frac{3r_{C_j}-1}{2}} \leq \frac{\pi}{2}$, and $\epsilon_{\vec{c}_\kt,\kt}$ is given by Eq. \eqref{EQ:hoeffding_avgfid}.
\end{theorem}

Recall that the individual estimates are known with certain confidence. That means that the above consistency check will be satisfied with a certain probability. We state it formally in the corollary below. 

\begin{corollary}\label{COR:probabilities_bound} Given the estimates of average fidelities for memories $r_{M^\mathcal{T}_j}$ and gates $r_{C_j}$ are known with confidence $\epsilon_{M_j^\mathcal{T}}$ and $\epsilon_{C_j}$ respectively, the bound from Thm.~\ref{THM:inconsistency_check} is satisfied by noisy devices with probability at least $1-2\sum_{j=1}^{\kt}\left(e^{-2n_{C_j}\epsilon_{C_j}^2} + e^{-2n_{M_j^\mathcal{T}}\epsilon_{M_j^\mathcal{T}}^2}\right)$.
\end{corollary}
\begin{proof}[Idea of the proof]
The probability that the bound \eqref{EQ:bound_incons} is satisfied is equal to the unity, minus the probability that at least one of the bounds for individual devices is not satisfied. By properties of probability one arrives at the statement above, see App. \ref{APP:SUBSEC:proof_cor_probabilties_bound} for details.
\end{proof}

\subsection{Performance of $k$-round protocols } \label{SUBSEC:performance}

In this section we investigate the implications of Test \ref{PROT:teleport} for the performance of more general $k$-round protocols $\tilde{\mathcal{P}}^{k}$, see Def.~\ref{DEF:class_of_protocols}. We show that their performance can be bounded using the winning rate $R_\kt$ (Sec.~\ref{SUBSEC:math_tools}) in the teleportation-based ping-pong test.

To explore the performance of protocols $\tilde{\mathcal{P}}^{k}$ we consider the diamond distance~\cite{Watrous2018}, $\parallel \Pi\circ \tilde{\mathcal{P}}^{k}\circ \tn{Prep} - \Pi\circ {\mathcal{P}}^{k}\circ \tn{Prep} \parallel_{\diamond}$. 
However, since $\tn{Prep}$ and $\Pi$ are perfect by assumption, the above diamond distance is upper-bounded by $\parallel \tilde{\mathcal{P}}^{k} - \mathcal{P}^{k} \parallel_\diamond$, which we fix to be the figure of merit in this section. It can be shown that the diamond distance is related to the average fidelity in the following way~\cite{Wallman2014},
\begin{align}\label{EQ:diamond_dist_prot}
\parallel \tilde{\mathcal{P}}^{k} - \mathcal{P}^{k} \parallel_\diamond \leq 2\sqrt{6} \sqrt{ \left(1-\bar{F}_{k,\vec{g}_k}(\tilde{\mathcal{P}}^{k})\right)}, 
\end{align}
where $\bar{F}_{k, \vec{g}_k}(\tilde{\mathcal{P}}^{k}) = \int \tn{d} \psi \Tr\left[ \tilde{\mathcal{P}}^{k} (\ketbra{\psi}) \cdot {\mathcal{P}}^{k}(\ketbra{\psi}) \right]$ is the average fidelity of a protocol $\tilde{\mathcal{P}}^{k}$ of a fixed depth $k$ and for a fixed string of gates $\vec{g}_k$. Note that the average fidelity differs depending on the sequence of gates one chooses to apply. Therefore, to estimate the behavior of protocol $\tilde{\mathcal{P}}^{k}$ one would have to know fidelities $\bar{F}_{k,\vec{g}_k}(\tilde{\mathcal{P}}^{k})$ for all possible gate sequences $G_1,...,G_k$, which is unfeasible in practice. 
For this reason, it is much more convenient to use double-averaged fidelity to bound the performance of a protocol $\tilde{\mathcal{P}}^{k}$. We formalize this argument in the following theorem.

\begin{theorem}[Performance of $k$-round protocols] \label{THM:bound_protocols}
The performance of single-qubit $k$-round protocols, Def. \ref{DEF:class_of_protocols}, can be bounded in terms of an estimate for the double-averaged fidelity $R_k$ of the $k$-round teleportation-based ping-pong test, Test \ref{PROT:teleport}, in the following way
\begin{align}\label{EQ:bound_protocols}
\parallel \tilde{\mathcal{P}}^{k} - \mathcal{P}^{k} \parallel_\diamond \leq 2 \sqrt{6} \sqrt{ |\mathsf{Cliff}|^{k}\left(1-R_k + \epsilon_k \right)}
\end{align}
where $|\mathsf{Cliff}|$ is the size of the Clifford group for dimension $2$ and $\epsilon_k$ is given by Eq. \eqref{EQ:hoeffding_doubleavg}. The bound is satisfied with probability $1 - e^{-2n_\kt \epsilon^2_{\kt}}$.
\end{theorem}
\begin{proof}[Idea of the proof]

To prove the theorem, one first needs to observe that the double-averaged fidelity, $\bar{\bar{F}}(\tilde{\mathcal{P}}^{k})$, can be lower-bounded by $\bar{F}_{\vec{g}_k}(\tilde{\mathcal{P}}^{k})$ minimized over all possible strings of gates $\vec{g}_\kt$, see App. \ref{APP:other_proofs} for details.

Moreover we have that $\bar{\bar{F}}_k(\tilde{\mathcal{P}}^{k}) = \bar{\bar{F}}_{\kt=k}(\tilde{\mathcal{T}}^{\kt= k})$. It follows from the fact that averaging over the Clifford group is equivalent to averaging over the entire unitary group, since the Clifford group forms a 2-design. Furthermore, the equality is possible, since we have put $M_j^\mathcal{T} \equiv M_j \circ \mathcal{E}_j$, and $M_j^\mathcal{T}$ encompasses operations associated with sending (in the test -- teleporting) and storing the qubit. Combining the above with Eqs. \eqref{EQ:doubleavgfid_expected} and  \eqref{EQ:diamond_dist_prot} yields the desired result.
\end{proof}

Finally, observe the above results can be straightforwardly generalized to bound the performance of protocols $\tilde{\mathcal{P}}^{K}$ for depth $K>k$. Since the teleportation-based ping-pong test is performed for all $1\leq \kt \leq k$, we can define a set $S$ such that  $\sum_{\kappa\in S} \kappa = K$. Then ${\tilde{\mathcal{P}}}^{K} = \bigcirc_{\kappa \in S}~  \tilde{{\mathcal{P}}}^{\kappa}$. Using the triangle inequality for the diamond distance, Thm. \ref{THM:bound_protocols} can be, therefore, rewritten as  
\begin{align}
\parallel \tilde{\mathcal{P}}^{K} - \mathcal{P}^{K} \parallel_\diamond \leq 2 \sqrt{6} \sum_{\kappa\in S} \sqrt{ |\mathsf{Cliff}|^\kappa\left(1-R_\kt + \epsilon_\kt\right)},
\end{align}
where $\epsilon_\kt$ is given by Eq. \eqref{EQ:hoeffding_doubleavg}.

\subsection{Simulated results} \label{SUBSEC:simulatedresults}
To gain intuition on how the test performs in this section we consider a few numerical examples. First, we discuss the implications of the consistency check, Thm. \ref{THM:inconsistency_check} and articulate the relation between the average fidelity of individual devices and the maximal depth of the test $k$. Second, we discuss the performance of the test under common noise models, depolarizing and dephasing noise. Finally, we comment on bounding the noisy protocols $\tilde{\mathcal{P}}^{k}$ based on numerical results from the teleportation-based ping-pong test.

Assume a test of maximum depth $k=2$, where we teleport a single qubit state at most two times between $A$ and $B$. Moreover, for simplicity say that $A$ and $B$ have access to memories and gates of equal fidelities, $r_{M^\mathcal{T}_j} = r_{C_j} = r$. 
\vl{Observe that the higher depth of the test $\kt$, i.e. the more devices one is testing, the higher individual fidelities should be, see Figure \ref{FIG:stuff_k}.}
\vl{Finally, note that the bound used for consistency check \eqref{EQ:bound_incons} was derived for a generic noise model and  it was shown to be tight~\cite{Dugas2016}. This means that if one does not have any additional knowledge about the noise present in the devices then the results presented here cannot be further improved. }

Let us now look at two specific noise models. Namely, let us model memories and gates to be \textit{(i)} single-qubit depolarizing channels, i.e. $\mathcal{D}(\rho) = p\rho + (1-p)\mathds{1}/2$ and \textit{(ii)} single-qubit dephasing channels, i.e. $\mathcal{F}(\rho) = q\rho + (1-q)(Z\rho Z^\dagger)/2$, where $Z$ is the Pauli $Z$ gate. Again, in these two cases let us fix the average fidelity estimate of individual devices $r$. Fig.~\ref{FIG:Simulation} presents the simulated behavior of the test as a function of individual estimates $r$ in the two cases. Observe that the test performs according to intuitive expectations -- if the noise is modeled as dephasing, the average fidelity of the test is higher than in the case of depolarizing noise, since the dephasing channel subjects any input state only to the $Z$ component of the Pauli noise, whereas depolarizing channel to all $X$, $Y$ and $Z$ components. Therefore, we expect ``more'' noise when the state is subjected to the depolarizing noise.

\begin{figure}
\begin{minipage}{0.6\linewidth}
\centering 
\includegraphics[width=\linewidth]{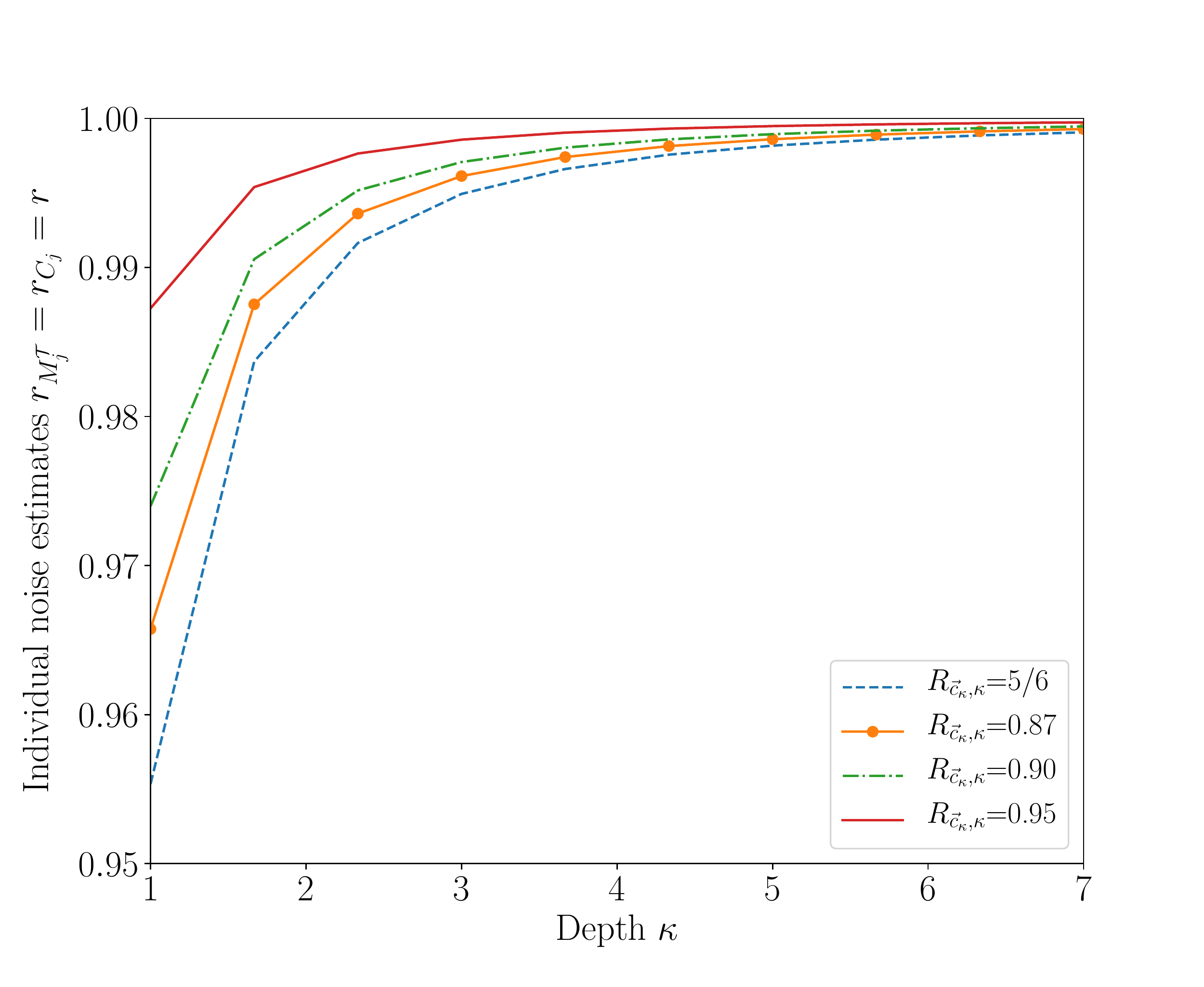}
\caption{Minimum average fidelity $r$  of individual operations (individual noise estimates) as a function of fixed depth $\kt$ for different winning rates $R_{\vec{c}_\kt,\kt}$. The plot shows the bound derived in Eq. \eqref{EQ:bound_incons}.}\label{FIG:stuff_k}
\end{minipage}
\end{figure}

\begin{figure}[!h]    
\centering
    \begin{minipage}{.45\textwidth}
        \centering
        \includegraphics[width=\textwidth]{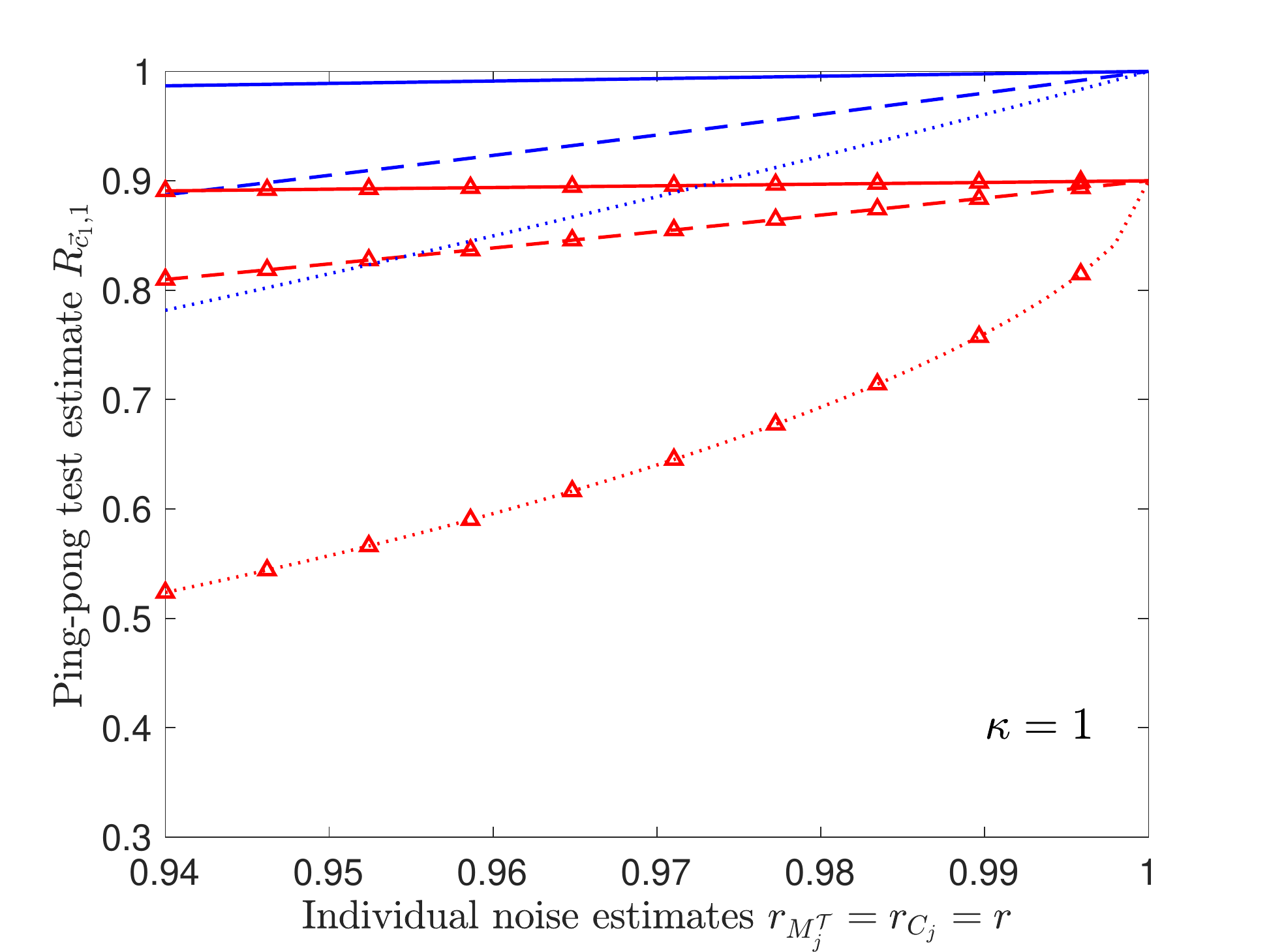}
        \label{fig:prob1_6_2}
    \end{minipage}
    \begin{minipage}{0.45\textwidth}
        \centering
        \includegraphics[width=\textwidth]{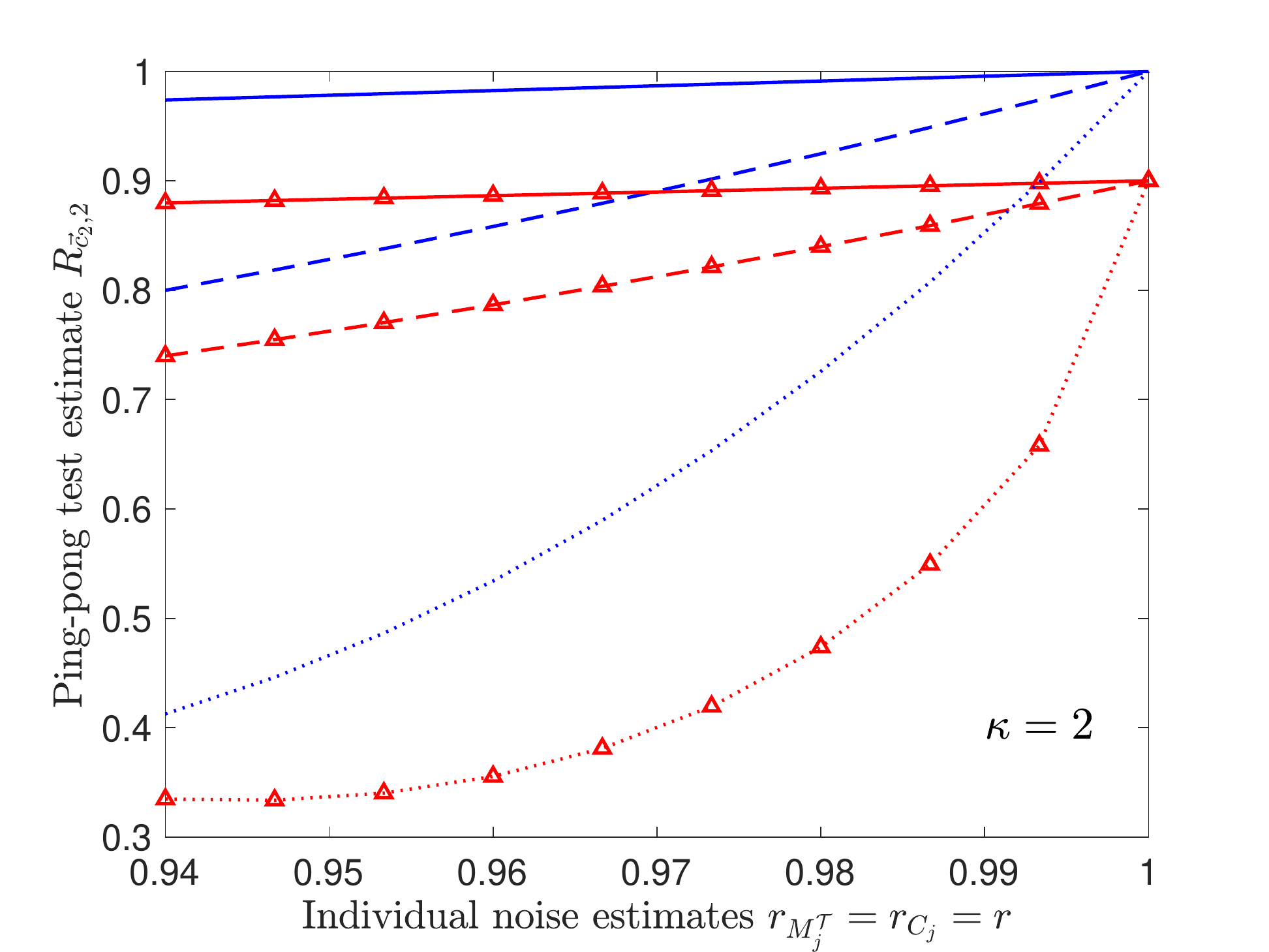}
        \label{fig:prob1_6_1}
    \end{minipage}
\caption{Average fidelity of the test as a function of the average fidelity of individual devices $r$, left plot for $\kt = 1$, right plot for $\kt=2$. Solid lines represent simulation of the test with the dephasing channel $\mathcal{D}$, dashed lines the simulation of the test with the depolarizing channel $\mathcal{F}$, and dotted lines the value of the bound \eqref{EQ:bound_incons}. No markers (blue color) correspond to the case where the input state is perfect, whereas the triangle markers (red color) to the case where the input state is dephased to initial fidelity 0.9.}\label{FIG:Simulation}
\end{figure}

Although in our network model we assume that the state preparation is perfect, it is interesting to see the behavior of the test once imperfect states are used. Fig.~\ref{FIG:Simulation} shows a result of simulation of the test when the initial state is submitted to a small dephasing noise, such that fidelity of the input state is 0.9. Note that if one has access to the average fidelity estimate of the noisy channel acting on the initial state, then one can use it in the consistency check \eqref{EQ:bound_incons}, simply treating the noise of the state as an additional channel in the protocol.

Let us also comment on the bound from Thm. \ref{THM:bound_protocols}. Already for a single qubit one obtains a constant prefactor of $2\sqrt{d(d+1)} \approx 4.9$. In addition to that, bound \eqref{EQ:bound_protocols} contains a factor associated with the size of the Clifford group -- for a single qubit $|\mathsf{Cliff}| = 24$. If one considers protocols of maximum depth $k=2$ then to obtain a non-trivial bound on the behavior of protocols in the class, the estimate of double-averaged fidelity must be of order $R_\kt = 1-10^{-5}$. This puts a very high precision requirement on double-averaged fidelity and, consequently, on individual devices. 

As an example consider \vl{the} quantum gambling (QG) protocol~\cite{Goldenberg1999}. In the protocol, $A$ chooses one of the states $\{\ket{0_z}, \ket{0_x}\}$ and sends it to $B$. After receiving the state $B$ stores the state and communicates classically his guess on the state sent by $A$. $A$ upon receiving the classical message from $B$, communicates back whether $B$ won or lost. After this round of communication $B$ measures the state either in $Z$ basis or $X$ basis. Let the protocol be described with a map $\mathcal{P}_{QG}$ which consists of local operations on the state (except measurement and state preparation, as before). Then $\mathcal{P}_{QG}$ consists of $k = 2$ rounds of communication during which $B$ has to store the state. Assume that in the protocol quantum memory is modeled as a depolarizing channel with fidelity \vl{$1-10^{-5}$}. Then explicit evaluation of the diamond distance $\parallel \tilde{\mathcal{P}}_{QG} - \mathcal{P}_{QG} \parallel_\diamond$ yields value $6 \cdot 10^{-5}$. On the other hand if one uses a two-round test to bound the behavior of the protocol, without explicit \textit{a priori} knowledge about the noise model of the memories then the bound from Thm. \ref{THM:bound_protocols} has the value 0.7436. \vl{However, note that in the quantum gambling protocol one does not perform any gates. Using this explicit knowledge about the protocol one could in principle tailor a ping-pong teleportation-based test without any gates. In this case, there would be no need to average over gates and therefore, the bound from Thm. \ref{THM:bound_protocols} would not carry the $|\mathsf{Cliff}|^k$ term. Consequently, the bound could be improved to value 0.0310.}

\subsection{More noise}\label{SEC:more_imperf}
In our network model we have assumed that \vl{state preparation, measurement}, sending qubits as well as preparing a EPR pair can be done perfectly. In particular, this implies that in our test teleportation is carried out perfectly. However, the test can still be performed without major changes if one wishes to take into account noisy teleportation. 

We consider two main noise sources arising in teleportation -- noise coming from performing imperfect Bell measurement and recovery operation, and noise originating from creation of a EPR pair. In App.~\ref{APP:absorbing_noise} we show that in a single round $j$ of our test both of these noise types can be absorbed into the noise coming from the memory $\tilde{M}^\mathcal{T}_j$, for all $j$. For the former noise source we assume a noise model where the imperfections follow the Bell measurement but precede recovery operation. For the latter noise source, we assume that noise is local for each half of the EPR pair and that it can be modeled as mixed-unitary noise. That is, each half of the EPR pair is subjected to $N(\cdot) = \sum_l p_l U_l^\dagger(\cdot)U_l$, where $U_l$ is a unitary operation, and $p_l$ is a probability. Then all the teleportation noise can be included in the noise of the memory and we can carry out the test as described above, i.e. by sending qubits via perfect teleportation.

\vl{Similarly to the analysis outlined in the previous paragraph, we can treat the noise of the state preparation as if it arose in the teleportation. Indeed, one can absorb the noise in the initial state similarly to the analysis in App. \ref{APP:absorbing_noise}. Note that in Fig.~\ref{FIG:Simulation} we indicate what one might expect from the test if the initial state is noisy. As for the noise in the final measurement, if we consider that the noisy measurement is described by a noise map $N$ followed by a perfect measurement, then $N$ can be treated as another noisy memory applied to the state before measuring. In this case, the analysis carried out in Lemma \ref{APP:LEM:cliff_2design} of App. \ref{APP:SEC:2designs} still holds.}

Finally, we remark that our test can be extended onto multi-qubit settings, where the number of qubits in the $k$-round protocol is $Q$. For a detailed description we refer the reader to App. \ref{APP:SEC:q_qubit_ext}.

\section{Conclusions and outlook}
In this work we considered the problem of certifying that a quantum network achieves the ability to perform a subset of protocols within a certain stage of development, i.e. a stage called quantum memory network. We designed the first testing protocol, which certifies that nodes have the capability to control and send qubits around the network $k$ times. We provided completeness and soundness statements for our protocol and expressed them in the interactive proof language. Moreover, in an honest implementation, we demonstrated that passing our test allows us to estimate statistical quantities about the devices used in the test and conclude about the performance of other $k$-round protocols in a quantum network.

An important question is how our estimate of performance for the class of multi-round protocols can be improved. Note that in our simple analysis we bound a very general class of protocols using a single test -- we bound the behavior of any unitary gate in terms of behavior of a small subset of gates. Therefore, it is not surprising that there must exist a trade-off between universality of the protocols and the precision of estimating their performance. One improvement could result from designing tests for a more specific  
(and therefore smaller) class of protocols. Alternatively, tailoring tests using additional knowledge of the underlying noise in a quantum network could improve the bound on the performance of $k$-round protocols.

Furthermore, as mentioned before, our test does not certify that any universal gate can be implemented. Due to the mathematical structures of unitary designs that we used, we can only make a statement about implementability of the gates from the Clifford set or \vl{any gate set with 2-design properties}. It is, therefore, an open problem how to test a quantum memory in the presence of the set powerful enough to generate any unitary operation. Such a universal set is, for example, a Clifford set extended with a $T$ gate~\cite{Nebe2001,Boykin1999}.

\section{Acknowledgments}
We thank J. Helsen, and G. Murta for inspiring discussions and useful comments on this work. We also thank B. Dirkse, T. Coopmans, M. Steudtner and K. Goodenough for feedback on the manuscript. This work was supported by STW Netherlands, NWO VIDI grant, ERC Starting grant and NWO Zwaartekracht QSC.

%

\appendix

\newpage

In the following we present technical details of our work. We first provide mathematical preliminaries necessary for our further considerations in App.~\ref{APP:SEC:preliminaries}. Then, in App.~\ref{APP:test_description} we give a detailed mathematical description of the general ping-pong test, Test \ref{PROT:generic}, and the teleportation-based ping-pong test, Test \ref{PROT:teleport}. In App.~\ref{APP:SEC:teleport_memory} we justify why in the teleportation-based ping-pong test, it is possible to absorb the (possibly noisy) teleportation channel into a memory $M_j^\mathcal{T}$. Next, we discuss 2-design properties of sets of Pauli states and Clifford gates in App.~\ref{APP:SEC:2designs}. In App.~\ref{APP:corr_sound} we prove completeness and soundness statements of our Test \ref{PROT:teleport}. Then, in App.~\ref{APP:other_proofs} we give proofs of statements discussed in the estimation view of our test. Finally, we discuss how to extend our results to $Q$-qubit protocols in App.~\ref{APP:SEC:q_qubit_ext}.

\begin{widetext}
\section{Preliminaries}\label{APP:SEC:preliminaries}

Communication between nodes of a quantum network can be described by quantum channels. A quantum channel can be described by a completely positive trace-preserving (CPTP) linear map $\Lambda: \mathds{D}(\mathcal{H}) \rightarrow \mathds{D}(\mathcal{H})$, where $\mathds{D}(\mathcal{H})$ denotes the space of density operators acting on Hilbert space $\mathcal{H}$. In a realistic setup, quantum channels are not perfect (or ideal) and instead of applying a perfect channel $\Lambda$ one applies its noisy counterpart $\tilde{\Lambda}$. If the perfect $\Lambda$ is unitary, then without loss of generality, a noisy channel $\tilde{\Lambda}$ can be written as a noise map $N$ followed by a perfect channel $\Lambda$, i.e. $\tilde{\Lambda} = \Lambda\circ N$. A sequence of $n$ operations can be represented as a composition of $n$ maps, $\tilde{\Lambda}_n \circ ... \circ \tilde{\Lambda}_1$.

One can quantify the difference between a noisy channel and its perfect implementation using the average fidelity.
\begin{definition}[Average fidelity]
The average fidelity of the channel $\tilde{\Lambda}$ (to $\Lambda$) is defined as
\begin{align}
\bar{F}({\tilde{\Lambda}}) = \int \tn{d}\psi \Tr\left[ \tilde{\Lambda}(\ketbra{\psi}{\psi}) ~ \Lambda(\ketbra{\psi}{\psi}) \right],
\end{align}
where $\tn{d}\psi$ is the Haar measure on pure states.
\end{definition}

Average fidelity is a quantity which can be accessed empirically and as such it is widely used as a parameter estimating the quality of a quantum channel. One cannot hope, however, to empirically average over the continuum of all pure states. Realistically, to access average fidelity one can use the properties of so called \textit{quantum state designs}. Intuitively, a quantum design is a probability distribution over pure states, which replicates the properties of the Haar averaging over the entire space of pure states.
\begin{definition}[Projective $t$-design] 
\label{DEF:2design}
A projective $t$-design is a distribution $\{ q_\psi, \psi \}$ over some finite set of states such that
\begin{equation}
\sum_{\psi} q_\psi \ketbra{\psi}{\psi}^{\otimes t} = \int \tn{d}\psi ~\ketbra{\psi}{\psi}^{\otimes t}\,.
\end{equation}
\end{definition}
An example of a projective 2-design for qubits is given by a set of six Pauli eigenstates, $\mathsf{X}$ chosen with equal probability $\frac{1}{6}$. A similar definition can be used when talking about averaging over the unitary group $U(d)$ of dimension $d$, see~\cite{Roy2009} for details. 
\begin{definition}[Unitary $2$-design]\label{APP:DEF:unit2design}
A set  $U(d)$ of unitary matrices is 2-design if  for any quantum channel $\Lambda$ holds that \cite{Gross2007}
\begin{align}
\frac{1}{|\mathcal{Y}|}\sum_{U_l \in \mathcal{Y}} U_l^\dagger \Lambda (U_l\rho U_l^\dagger) U_l = \int \tn{d}U~ U^\dagger \Lambda (U\rho U^\dagger) U
\end{align}
where $\tn{d}U$ denotes the Haar measure on $U(d)$. An example of a $2$-design for a unitary group $U(d)$ is the Clifford group $\mathsf{Cliff}(d)$ with uniform probability of each element.
\end{definition}

Another useful figure of merit for channels is the diamond distance~\cite{Watrous2018}. 
\begin{definition}[Diamond distance]
The diamond distance between two operators, $\tilde{\Lambda}$ and ${\Lambda}$, is defined through a distance measure on the space of density operators, maximized over all density operators $\rho$,
\begin{align} 
\parallel \tilde{\Lambda} - \Lambda \parallel_{\diamond} = \sup_{\rho } \parallel \tilde{\Lambda}\otimes \mathds{1}(\rho) - \Lambda\otimes \mathds{1}(\rho) \parallel_{1},
\end{align}
where $\parallel \cdot \parallel_1$ is the trace distance. The operational meaning behind the diamond distance definition is that it quantifies the worst-case distinguishability of any two quantum channels when one is given access to entanglement with an auxiliary system.
\end{definition}
From the properties of the diamond distance it follows that,
\begin{align} \label{EQ:diamond_sum_maps}
\parallel \tilde{\Lambda}_N \circ ... \circ \tilde{\Lambda}_1 - \Lambda_N \circ ... \circ \Lambda_1 \parallel_\diamond
 \leq \sum_{j = 1}^N \parallel  \tilde{\Lambda}_j - \Lambda_j \parallel_\diamond.
\end{align}
Note that such a relation cannot be easily found for average fidelity, since, unlike the diamond distance, fidelity is not a metric.

Although the diamond distance offers a convenient theoretical description, it is not as practical as average fidelity. But, since average fidelity and diamond distance both estimate the quality of a quantum channel, there exists a relation between the two. Indeed, it can be shown CIT that
\begin{align}\label{EQ:diamond_avg_fid}
\parallel \tilde{\Lambda} - \Lambda \parallel_{\diamond} \leq 2\sqrt{d(d+1)}\sqrt{1-\bar{F}({\tilde{\Lambda}}}),
\end{align}
where $d$ is the dimension of the underlying quantum system. 

While performing an experiment, for example estimating the average fidelity, one gathers empirical data. To compare the data with theoretical expectation one can use the Hoeffding's inequality~\cite{Hoeffding1963}. It states that the probability of the empirical mean and its expectation differing by more than $\epsilon$ is exponentially small in $n$.
\begin{lemma}[Hoeffding's inequality] If $v_1,...,v_n$ are independent random variables, $0\leq v_i \leq 1$, with empirical mean defined as
\begin{equation}
R = \frac{v}{n} = \frac{\sum_{i=1}^n v_i}{n},
\end{equation}
then an upper bound on the probability that the mean of random variables deviates from its expected value is given by
\begin{equation}
\tn{Pr}[|R - \mathds{E}[R]| \geq \epsilon] \leq 2e^{-2n\epsilon^2} .
\end{equation}
\end{lemma}

\begin{lemma}[Choi isomorphism]\label{APP:LEM:choi} For a map $\Omega: \mathcal{H}_{S_1} \rightarrow \mathcal{H}_{S_2}$ the following identity holds:
\begin{align}
&\Tr[\ketbra{\psi}{\psi}_{S_2}~\Omega_{S_1 \rightarrow S_2} (\ketbra{\psi}{\psi}_{S_1})] = \\
& = |S_1|\Tr[ \ketbra{\psi}{\psi}_{S_2} \otimes \ketbra{\psi}{\psi}_{S_1'} ~ \omega_{S_2S_1'}^{\Gamma} ],
\end{align}
where $\omega_{S_2S_1'}^{\Gamma}$ is a Choi state associated with the map $\Omega_{S_1 \rightarrow S_2}$ of the form $\omega_{S_2S_1'}^{\Gamma} = \Omega_{S_1 \rightarrow S_2} \otimes \mathds{1}_{S_1'} (\Phi)$, with $\Phi = \sum_{i,j}\ketbra{ii}{jj}/|S_1'|$ being the maximally entangled state. $\Gamma$ denotes partial transposition of $\omega$ on the system $S_1'$, and $|S_1|$ ($|S_1'|$) is a size of Hilbert space $\mathcal{H}_{S_1}$ ($\mathcal{H}_{S_1'}$).
\end{lemma}

\section{The test -- detailed description}\label{APP:test_description}
In this section we provide a mathematical detailed description of our tests. First we consider a general case of the ping-pong test, Test \ref{PROT:generic}. Then we discuss the specific case of the teleportation-based ping-pong test, Test \ref{PROT:teleport}.

\subsection{General ping-pong test}

We describe a general test Test \ref{PROT:generic} as a CPTP map which we will denote ${\mathcal{S}}_{\kt}$. We first consider all the registers available to the nodes. We call $\kt$ the \textit{depth} of the test and assume that $\kt$ is a natural number upper-bounded by given $k$. The time for performing one round $j = 1,...,\kt$ of the protocol is equal for all the rounds, i.e. $\Delta t = t_{j+1} - t_j = t_\tn{send} + t_M + \ell $.

We will describe a {\em round} where node $A$ initiates sending of the state, which implies that $j$ is odd. However, this description is fully symmetric and for even $j$ it is enough to interchange registers of $A$ with registers of $B$. $A$ sends the qubit $\ketbra{\psi}_{A_j^\inn}$ to node $B$ using channel $\mathcal{E}_{A_j^\inn \rightarrow B_{j}}$ which takes time upper-bounded by $t_\tn{send}$. After time $t_\tn{send} + t_M$ the verifier chooses a gate according to distribution $p_\mathsf{G}$ and gives its classical description $\ketbra{g_j}_{G_j}$ to $B$. $B$ applies the quantum gate that corresponds to the description and that we describe with a CPTP map $G_{G_jB_j \rightarrow G_j}$. This takes time $\ell$. After this, at time $t_\tn{send} + t_M + \ell$ the verifier distributes a challenge bit $\ketbra{f_j}_{F_j}$ chosen uniformly at random ($0$ means 'teleport back', $1$ means 'output'). Depending on the challenge, $B$ applies $IN_{F_jB_j^\epr \rightarrow B_{j+1}^\inn}$ for $f_j = 0$ and $OUT_{F_jB_j \rightarrow B_{j}^\out}$ for $f_j = 1$.

\begin{definition} (Honest round $j$) \label{DEF:round_j}
Round $j$ of a general test, where provers are honest, can be described as
\begin{align} 
\hat{\Lambda}_{A_j^\inn F_jG_j \rightarrow B_{j+1}^\inn} =  IN_{F_jB_j^\epr \rightarrow B_{j+1}^\inn} \circ G_{G_jB_j \rightarrow B_j} \circ M_{B_j \rightarrow B_j} \circ \mathcal{E}_{A_j^\inn \rightarrow B_j}
\end{align} whenever the challenge bit is 0, or
\begin{align} 
\hat{\Lambda}_{A_j^\inn  F_jG_j \rightarrow B_{j}^\out} = OUT_{F_jB_j \rightarrow B_{j}^\out} \circ G_{G_jB_j \rightarrow B_j} \circ M_{B_j \rightarrow B_j} \circ \mathcal{E}_{A_j^\inn \rightarrow B_j}
\end{align} whenever the challenge bit is 1. 
\end{definition}

Note that challenge bits form a string of length ${\kt}$, $f_1\dots f_{\kt}$, in registers $F_1\dots F_{\kt}$, consisting of ${\kt}-1$ ones and a single zero bit on ${\kt}$-th position. We denote such a string by $\vec{f}$, i.e. $\vec{f}_{\kt} = \underbrace{1\dots 10}_{\kt}$. For simplicity we will use a short notation for multiple registers, e.g. $F_{[1,{\kt}]}\equiv F_1\dots F_{\kt}$. Similarly, we will denote by $\vec{g}_{\kt}$ a sequence of ${\kt}$ gates chosen by the verifier, each of the gates chosen at the time step defined above. By $G_{[1,{\kt}]}\equiv G_1\cdots G_{k}$ we denote $k$ registers for the choice of a gate.

\begin{definition} 
The ping-pong testing protocol of depth $\kt$ for a state $\ketbra{\psi}\in \mathcal{X}$, a sequence of gates $\vec{g}_{\kt} \in \mathcal{G}^\kt$ and a string of challenges $\vec{f}_{\kt}=1\dots10$ is defined as a CPTP map ${\mathcal{S}}_{\kt}$ such that
\begin{align}
{\mathcal{S}}_{\kt} & \equiv \mathcal{S}_{A_1^\inn  F_{[1,\kt]} G_{[1,\kt]} \rightarrow B_{\kt}^\out}  
\left(\ketbra{\psi}{\psi}_{A_1^\inn} \otimes \ketbra{\vec{f}_\kt}_{F_{[1,\kt]}} \otimes \ketbra{\vec{g}_\kt}_{G_{[1,\kt]}} \right) = \\
 & = \hat{\Lambda}_{A_j^\inn  F_jG_j \rightarrow B_{j}^\out} \circ \bigcirc_{j = 1}^{\kt-1} \hat{\Lambda}_{A_j^\inn F_jG_j \rightarrow B_{j+1}^\inn}  \left(\ketbra{\psi}{\psi}_{A_1^\inn} \otimes \ketbra{\vec{f}_{\kt}}_{F_{[1,{\kt}]}} \otimes \ketbra{\vec{g}_{\kt}}_{G_{[1,{\kt}]}} \right)
\end{align} 
\end{definition}

\subsection{Teleportation-based test}

We now describe a single round $j$ of the test when the quantum communication is performed with teleportation, $\mathcal{T}_\kt$ (see Test \ref{PROT:teleport}). Let node $A$ initiate the teleportation, i.e. $j$ is odd. A round $j$ starts with placing the state to be teleported $\ketbra{\psi}{\psi}_{A_j^\inn}$ in an input register of $A$, $A_j^\inn$. Nodes generate an EPR pair between each other, $\Phi^+_{A_j^\epr B_j^\epr}$. $A$, using the generated pair, teleports the state $\ketbra{\psi}{\psi}_{A_j^\inn}$ to $B$ by performing a Bell state measurement and sending a classical message $m\in M$. This action is described by a CPTP map $\mathcal{B}_{A_j^\inn A_j^\epr \rightarrow M}$. At the same time $B$ applies quantum memory to his half of the EPR pair while waiting for the classical message from $A$ to arrive, which takes time $t_M$. We describe the action of the memory with a CPTP map $M_{B_j^\epr \rightarrow B_j^\epr}$. Upon receiving classical message $B$ now applies a recovery map $\mathcal{R}_{M B_j^\epr \rightarrow B_j^\epr}$ to recover the state which $A$ teleported. Then, $B$ applies a random gate chosen by the verifier from the set of Clifford gates $\mathsf{Cliff}(2)$. This choice is announced by the verifier with a classical register $\ketbra{c_j}_{C_j}$. $B$ then applies the gate to his recovered state, which we describe with a CPTP map $C_{C_jB_j^\epr \rightarrow B_j^\epr}$. This operation takes time $\ell$. 

Now, at time $t_M + \ell$ the verifier announces a flag in register $F_j$ with the challenge bit, $0$: teleport back, $1$: output. The choice of the challenge is uniform and random. Depending on the challenge, $B$ applies $IN_{F_jB_j^\epr \rightarrow B_{j+1}^\inn}$ for $f_j = 0$ and $OUT_{F_jB_j^\epr \rightarrow B_{j}^\out}$ for $f_j = 1$. The whole round $j$ takes time $\Delta t = t_M + \ell$.

\begin{definition} (Round $j$ of the teleportation-based test) \label{DEF:round_j}
We define a $j$-th round of teleportation as a sequence of following maps
\begin{align} \label{DEF:teleport_map0}
\Lambda_{A_j^\inn A_j^\epr B_j^\epr F_jC_j \rightarrow B_{j+1}^\inn} =  IN_{F_jB_j^\epr \rightarrow B_{j+1}^\inn} \circ C_{C_jB_j^\epr \rightarrow B_j^\epr} \circ \mathcal{R}_{M B_j^\epr \rightarrow B_j^\epr} \circ M_{B_j^\epr \rightarrow B_j^\epr} \circ \mathcal{B}_{A_j^\inn A_j^\epr \rightarrow M}
\end{align} whenever the challenge bit is 0, or
\begin{align} \label{DEF:teleport_map1}
\Lambda_{A_j^\inn A_j^\epr B_j^\epr F_jC_j \rightarrow B_{j}^\out} = OUT_{F_jB_j^\epr \rightarrow B_{j}^\out} \circ C_{C_jB_j^\epr \rightarrow B_j^\epr} \circ \mathcal{R}_{M B_j^\epr \rightarrow B_j^\epr} \circ M_{B_j^\epr \rightarrow B_j^\epr} \circ \mathcal{B}_{A_j^\inn A_j^\epr \rightarrow M}
\end{align} whenever the challenge bit is 1. 
\end{definition}
Note that, for simplicity, in the main text we denote $M^\mathcal{T}_j = \mathcal{R}_{M B_j^\epr \rightarrow B_j^\epr} \circ M_{B_j^\epr \rightarrow B_j^\epr} \circ \mathcal{B}_{A_j^\inn A_j^\epr \rightarrow M} $.

Having defined a single round of a protocol we describe the ping-pong teleportation protocol of depth ${\kt}$. Such a protocol is simply a ${\kt}$-round teleportation, where first ${\kt}-1$ maps have form \eqref{DEF:teleport_map0} and the last map outputs the state and so has the form \eqref{DEF:teleport_map1}. 

\begin{definition} \label{DEF:protocol_depth_ki}
The teleportation-based ping-pong testing protocol of depth ${\kt}$ for a state $\ketbra{\psi}\in \mathsf{X}$, a sequence of gates $\vec{c}_{\kt} \in \mathsf{Cliff}^\kt$ and a string of challenges $\vec{f}_{\kt} = 1\dots 10$ is defined as a CPTP map $\mathcal{T}_\kt$ such that
\begin{align}
\mathcal{T}_\kt  \equiv \mathcal{T}_{A_1^\inn A_{[1,{\kt}]}^\epr B_{[1,{\kt}]}^\epr  F_{[1,{\kt}]} C_{[1,{\kt}]} \rightarrow B_{\kt}^\out}  = \Lambda_{A_{{\kt}}^\inn A_{{\kt}}^\epr B_{{\kt}}^\epr F_{\kt} C_{\kt} \rightarrow B_{{\kt}}^\out} \circ \bigcirc_{j = 1}^{{\kt}-1} \Lambda_{A_j^\inn A_j^\epr B_j^\epr F_j C_j \rightarrow B_{j+1}^\inn}
\end{align} 
applied to the input state
\begin{align}
\ketbra{\psi}{\psi}_{A_1^\inn} \otimes \bigotimes_{j=1}^{{\kt}} \Phi^+_{A_j^\epr B_j^\epr} \otimes \ketbra{\vec{f}_{\kt}}_{F_{[1,{\kt}]}} \otimes \ketbra{\vec{c}_{\kt}}_{C_{[1,{\kt}]}}
\end{align}
where $\Lambda$'s are defined as in Def. \ref{DEF:round_j}.
\end{definition}

\subsection{Measurements}

Upon receiving requested state from either $A$ or $B$, $V$ must check its consistency with the distributed state, as well as confirm applying desired gates. This can be achieved by projecting outcomes onto the state $C_{{\kt}}\circ \dots \circ C_{1} (\ketbra{\psi}{\psi})_{B_{\kt}^\out}$, which is the original state rotated with ${\kt}$ Clifford channels.

\begin{definition}[POVM elements for the node $V$]
\label{obs:povm}
Measurements performed by $V$ in the teleportation-based ping-pong test can be described by POVM elements,
\begin{align}
& \Pi_{\cmark}^{\kt} =  C_{\kt}\circ \dots \circ C_{1} (\ketbra{\psi}{\psi}) _{B_{\kt}^\out} \\
& \Pi_{\xmark}^{\kt} =  \mathds{1} - C_{\kt}\circ \dots \circ C_{1} (\ketbra{\psi}{\psi}) _{B_{\kt}^\out} 
\end{align}
for all ${\kt} = 1,\dots ,k$. ${\kt}$ denotes here the output register of the ${\kt}$-th party, depending on the parity either $A$ or $B$. 
\end{definition}

\subsection{Renaming teleportation channel}

Now that we have formalized the testing protocol in detail, we will justify using notation for a teleportation channel used in the main text. That is, we will show that  a teleportation channel with noisy memory acting on $\ketbra{\psi}\otimes \Phi^+$ can be viewed as a channel $M^\mathcal{T}_j$ acting only on $\ketbra{\psi}$.

Recall Definition \ref{DEF:round_j}. In a single round $j$ of the protocol $A$ performs a Bell measurement on the state $\ketbra{\psi}{\psi}_{A_j^\inn}$ and her part of EPR pair. This action is described  described by an operator $\mathcal{B}_{A_j^\inn A_j^\epr \rightarrow M}$, acting on two registers on $A$'s side and producing a classical message $m\in M$ which is then sent to $B$. The initial state $\ketbra{\psi}{\psi}_{A_j^\inn} \otimes \Phi^+_{A_j^\epr B_j^\epr}$ becomes 
\begin{equation}
\begin{split}
\mathcal{B}_{A_j^\inn A_j^\epr \rightarrow M} 
(\ketbra{\psi}{\psi}_{A_j^\inn} \otimes \Phi^+_{A_j^\epr B_j^\epr}) & = 
\Tr_{A_j^\inn A_j^\epr}\left[
\sum_{m\in M} p_m {\Psi'}_{A_j^\inn A_j^\epr, m} \otimes (U_m \ketbra{\psi}{\psi} U_m^\dagger)_{B_j^\epr} \otimes \ketbra{m}{m}_M,
\right] = \\
& = \sum_{m\in M} p_m  (U_m \ketbra{\psi}{\psi} U_m^\dagger)_{B_1^\epr} \otimes \ketbra{m}{m}_M
\end{split}
\end{equation}
where ${\Psi'}_{A_j^\inn A_j^\epr, m}$ is one of four Bell states resulting from the Bell measurement, $p_m \geq 0$, $\sum_mp_m=1$ is a probability of an outcome $m$ occurring. $(U_m \ketbra{\psi}{\psi} U_m^\dagger)_{B_j^\epr}$ is a state on $B$'s register after the Bell measurement. Note that this is simply a unitary applied to the initial state. $\ketbra{m}{m}$ is a classical message register, which $A$ sends to $B$ in order for him to correct the state. Next, $B$ applies a memory $M_{B_j^\epr \rightarrow B_j^\epr}$ to his share of the state:
\begin{align}
	M_{B_j^\epr \rightarrow B_j^\epr} \circ \mathcal{B}_{A_j^\inn A_j^\epr \rightarrow M}
(\ketbra{\psi}{\psi}_{A_j^\inn} \otimes \Phi^+_{A_j^\epr B_j^\epr}) & = 
M_{B_j^\epr \rightarrow B_j^\epr}\sum_{m\in M} p_m  (U_m \ketbra{\psi}{\psi} U_m^\dagger)_{B_j^\epr} \otimes \ketbra{m}{m}_M \\
	& = \sum_{m\in M} p_m M_{B_j^\epr \rightarrow B_j^\epr} (U_m \ketbra{\psi}{\psi} U_m^\dagger)_{B_j^\epr} \otimes \ketbra{m}{m}_M
\end{align}
Upon receiving a classical message $m$ $B$ undoes the unitary operations to recover the teleported state. This operation is described by a map $\mathcal{R}_{M B_j^\epr \rightarrow B_j^\epr}(\cdot) = \Tr_{M}[\sum_m U_m (\cdot) U_m^\dagger \otimes \ketbra{m}{m}_M]$,
\begin{equation}
\begin{split}
& \mathcal{R}_{M B_j^\epr \rightarrow B_j^\epr}\circ M_{B_j^\epr \rightarrow B_j^\epr} \circ \mathcal{B}_{A_j^\inn A_j^\epr \rightarrow M}(\ketbra{\psi}{\psi}_{A_j^\inn} \otimes \Phi^+_{A_j^\epr B_j^\epr}) \\
& = 
\Tr_M \left[
\sum_{m\in M} p_m U_m^\dagger M_{B_j^\epr \rightarrow B_j^\epr} (U_m \ketbra{\psi}{\psi} U_m^\dagger)_{B_j^\epr} U_m \otimes \ketbra{m}{m}_M
\right]\\
& = \sum_{m\in M} p_m U_m^\dagger M_{B_j^\epr \rightarrow B_j^\epr} (U_m \ketbra{\psi}{\psi} U_m^\dagger)_{B_j^\epr} U_m\\
& = \sum_{m\in M} p_m \mathcal{U}_m^\dagger \circ M_{B_j^\epr \rightarrow B_j^\epr} \circ \mathcal{U}_m (\ketbra{\psi}{\psi}) =: M^\mathcal{T}_j (\ketbra{\psi})
\end{split}
\end{equation}

Then, the test of depth $\kt$ can be described as in the main text
\begin{align}
{\mathcal{T}}^{\kt} & =  {\mathcal{T}}_{A_1^\inn A_{[1,\kt]}^\epr B_{[1,\kt]}^\epr  F_{[1,\kt]} C_{[1,\kt]} \rightarrow B_{\kt}^\out}  \left(\ketbra{\psi}{\psi}_{A_1^\inn} \otimes \bigotimes_{j=1}^{\kt} \Phi^+_{A_j^\epr B_j^\epr} \otimes \ketbra{\vec{f}_\kt}_{F_{[1,k]}} \otimes \ketbra{\vec{c}_{\kt}}_{C_{[1,k]}} \right) \\
& \equiv \bigcirc_{j = 1}^{\kt} {\Lambda}_j = \bigcirc_{j = 1}^{\kt}~ {C}_{j} \circ M^\mathcal{T}_{j} \left( \ketbra{\psi}{\psi} \right)
\end{align}

\section{Teleportation and quantum memory}\label{APP:SEC:teleport_memory}

\subsection{Absorbing teleportation noise into the memory} \label{APP:absorbing_noise}
As it is often done in the estimation literature for quantum computing, see e.g.~\cite{Emerson2005,Knill2008,Magesan2011} , we will model teleportation as a perfect operation followed (or preceded) by noise. This will allow us to consider teleportation as a perfect operation i.e.~with perfect Bell measurement and recovery operation as well as perfect EPR pair, and absorb all the associated noise into the quantum memory. 

\subsubsection{Noisy operations}
Assume a Bell state measurement is followed by a local noise, $\mathcal{B}_{A_j^\inn A_j^\epr \rightarrow M} \equiv N^\mathcal{B}\circ\mathcal{B}_{A_j^\inn A_j^\epr \rightarrow M}$. Assume further that the recovery operation is also noisy, but in this case the map is preceded by the noise, $\mathcal{R}_{M B_j^\epr \rightarrow B_j^\epr} \equiv \mathcal{R}_{M B_j^\epr \rightarrow B_j^\epr} \circ N^\mathcal{R}$. Looking at Definition \ref{DEF:round_j} it is now clear that one can redefine $M_{B_j^\epr \rightarrow B_j^\epr}' = N^\mathcal{R} \circ M_{B_j^\epr \rightarrow B_j^\epr} \circ N^\mathcal{B}$ and use $M'$ as a new memory channel.

\subsubsection{Noisy EPR pair}
The situation is similar for a noisy EPR pair. Assume a EPR pair is affected by local noise, i.e. teleportation occurs on a state $N_{A_j^\epr} \otimes N_{B_j^\epr} \left(\Phi^+_{A_j^\epr B_j^\epr}\right)$. Here the maps $N$ are mixed-unitary channels, i.e. have the form $N(\cdot)=\sum_l p_l U_l (\cdot) U_l^\dagger$ , with $p_l$ being a probability and $U_l$ a unitary. Note that this is not the most general type of noise, however the most common ones (e.g. depolarizing, dephasing) can be modeled this way.  
Moreover, note that for an EPR pair it holds that
\begin{align}
U_{A_j^\epr} \otimes U_{B_j^\epr} \left(\Phi^+_{A_j^\epr B_j^\epr}\right) = \tn{id}_{A_j^\epr} \otimes U_{B_j^\epr} U_{A_j^\epr}^T\left(\Phi^+_{A_j^\epr B_j^\epr}\right).
\end{align}
Therefore, using an explicit form of maps $N$ and the above statement, we can write,
\begin{align}
N_{A_j^\epr} \otimes N_{B_j^\epr} \left(\Phi^+_{A_j^\epr B_j^\epr}\right) &= \sum_{l,l'} p_{l,l'} U_{A_j^\epr} \otimes U_{B_j^\epr} \left(\Phi^+_{A_j^\epr B_j^\epr}\right) \\
&= \sum_{l,l'} p_{l,l'} \tn{id}_{A_j^\epr} \otimes U_{B_j^\epr}U_{A_j^\epr} ^T \left(\Phi^+_{A_j^\epr B_j^\epr}\right) \\
& =:  \tn{id}_{A_j^\epr} \otimes N'_{B_j^\epr} \left(\Phi^+_{A_j^\epr B_j^\epr}\right)
\end{align}
In particular, this means that noise acting on the EPR pair, which has the mixed-unitary form, can be absorbed into the memory map,
\begin{align}
&M_{B_j^\epr \rightarrow B_j^\epr}' \circ \mathcal{B}_{A_j^\inn A_j^\epr \rightarrow M} \circ \left( N_{A_j^\epr} \otimes N_{B_j^\epr}\right) \left(\Phi^+_{A_j^\epr B_j^\epr}\right) \\
& = M_{B_j^\epr \rightarrow B_j^\epr}' \circ \mathcal{B}_{A_j^\inn A_j^\epr \rightarrow M} \circ \left( \tn{id}_{A_j^\epr} \otimes N_{B_j^\epr}' \right) \left(\Phi^+_{A_j^\epr B_j^\epr}\right) \\
& = M_{B_j^\epr \rightarrow B_j^\epr}'\circ N_{B_j^\epr}' \circ \mathcal{B}_{A_j^\inn A_j^\epr \rightarrow M}    \left(\Phi^+_{A_j^\epr B_j^\epr}\right) \\
& \equiv M_{B_j^\epr \rightarrow B_j^\epr}''\circ \mathcal{B}_{A_j^\inn A_j^\epr \rightarrow M}    \left(\Phi^+_{A_j^\epr B_j^\epr}\right) 
\end{align}

\section{2-designs}\label{APP:SEC:2designs}

In this appendix we show that for the ping-pong test,  the average of the probability $p_{\cmark|\psi,\vec{c}_{\kt},{\kt}}$ over the six Pauli states is equal to its average over the whole state space according to the Haar measure. To do so, we use the fact that the uniform distribution over set X is a 2-design \cite{Roy2009} and $p_{\cmark|\psi,\vec{c}_{\kt},{\kt}}$ contains a polynomial of degree 2 in $\ket{\psi}$. 
Next, we prove a similar statement when averaging over the Clifford group.
\subsection{Pauli states}

\begin{lemma}\label{APP:LEM:paulistates_2design}
Averaging the probability of success for a single execution of Test \ref{PROT:teleport}, $p_{\cmark|\psi,\vec{c}_{\kt},{\kt}}$, over Pauli states is equal to averaging over all qubit states according to the Haar measure,

\begin{align}
\frac{1}{|\mathsf{X}|} \sum_{\psi \in \mathsf{X}} p_{\cmark|\psi,\vec{c}_{\kt},{\kt}} = \int \tn{d} \psi~ p_{\cmark|\psi,\vec{c}_{\kt},{\kt}}.
\end{align}

\end{lemma}

\begin{proof}[Proof of Lem. \ref{LEM:avgfid_expected}] We can write the left-hand side explicitly as,
\begin{align}
\frac{1}{|\mathsf{X}|} \sum_{\psi \in \mathsf{X}} p_{\cmark|\psi,\vec{c}_{\kt},{\kt}} & = \frac{1}{|\mathsf{X}|}\sum_\psi \Tr [{\mathcal{T}}_\kt(\ketbra{\psi}{\psi}_{A_1^{\inn}})\cdot \bigcirc_{j = 1}^{\kt} C_j (\ketbra{\psi}{\psi})_{B_{\kt}^\out}] 
\end{align}
where we explicitly write $A$ and $B$'s input and output registers. Let $X,Y \in \mathcal{L}(\mathcal{H})$ be linear operators over the Hilbert space. The inner product $\braket{X}{Y}:=\Tr[X^\dagger Y]$ is invariant $\forall U: \mathcal{L}(\mathcal{H}) \rightarrow \mathcal{L}(\mathcal{H})$, i.e. $\braket{U(X)}{U(Y)} = \Tr[(U(X))^\dagger \cdot U(Y)] = \Tr[X^\dagger \cdot Y] = \braket{X}{Y}$. Note that $\bigcirc_{j = 1}^{\kt} C_j$ is a unitary channel and therefore, we can write,
\begin{align}
\frac{1}{|\mathsf{X}|} \sum_{\psi \in \mathsf{X}} p_{\cmark|\psi,\vec{c}_{\kt},{\kt}} 
& = \frac{1}{|\mathsf{X}|}\sum_\psi \Tr [\left(\bigcirc_{j = 1}^{\kt} C_j\right)^{-1} \circ {\mathcal{T}}_\kt(\ketbra{\psi}{\psi}_{A_1^{\inn}})\cdot \left(\bigcirc_{j = 1}^{\kt} C_j\right)^{-1} \circ \bigcirc_{j = 1}^{\kt} C_j (\ketbra{\psi}{\psi})_{B_{\kt}^\out}] \\
& = \frac{1}{|\mathsf{X}|}\sum_\psi \Tr [\left(\bigcirc_{j = 1}^{\kt} C_j\right)^{-1} \circ {\mathcal{T}}_\kt(\ketbra{\psi}{\psi}_{A_1^{\inn}})\cdot \ketbra{\psi}{\psi}_{B_{\kt}^\out}] 
\end{align}
Now, using Choi-Jamiolkowski theorem, see Lem. \ref{APP:LEM:choi}, we can write 
\begin{align}
\frac{1}{|\mathsf{X}|} \sum_{\psi \in \mathsf{X}} p_{\cmark|\psi,\vec{c}_{\kt},{\kt}}  = \frac{1}{|\mathsf{X}|}\sum_\psi |A_1^{\inn'}| \Tr [\ketbra{\psi}{\psi}_{A_1^{\inn}}\otimes \ketbra{\psi}{\psi}_{B_{\kt}^\out} \omega^\Gamma_{A_1^{\inn'} B_{\kt}^\out}]
\end{align}
$\omega^\Gamma_{A_1^{\inn'} B_{\kt}^\out}$ is a Choi-Jamiolkowski state associated with the map $\left(\bigcirc_{j = 1}^{\kt} C_j\right)^{-1} \circ \tilde{\mathcal{T}}_\kt$. It is now clear that averaging is taken over a polynomial of degree 2 under the trace and we can use properties of a 2-design. Therefore,
\begin{align}
\frac{1}{|\mathsf{X}|} \sum_{\psi \in \mathsf{X}} p_{\cmark|\psi,\vec{c}_{\kt},{\kt}} & = \int \tn{d}\psi |A_1^{\inn'}| \Tr [\ketbra{\psi}{\psi}_{A_1^{\inn}}\otimes \ketbra{\psi}{\psi}_{B_{\kt}^\out} \omega^\Gamma_{A_1^{\inn'} B_{\kt}^\out}] \\
& = \int \tn{d}\psi
\Tr [{\mathcal{T}}_\kt(\ketbra{\psi}{\psi}_{A_1^{\inn}})\cdot \bigcirc_{j = 1}^{\kt} C_j (\ketbra{\psi}{\psi})_{B_{\kt}^\out}] \\
& = \int \tn{d}\psi~ p_{\cmark|\psi,\vec{c}_{\kt},{\kt}}
\end{align}
where we used Choi-Jamiolkowski isomorphism and properties of the trace again. We define $\bar{F}_{\vec{c}_{\kt},{\kt}} = \int \tn{d}\psi~ p_{\cmark|\psi,\vec{c}_{\kt},{kt}}$ as the average fidelity.
\end{proof}

\subsection{Clifford gates}

Now we will prove that averaging $p_{\cmark|\psi,\vec{c}_{\kt},{\kt}}$ over the Clifford set reproduces averaging over the whole unitary set taken according to the Haar measure.

\begin{lemma}\label{APP:LEM:cliff_2design}
Averaging the probability of success for a single execution of Test \ref{PROT:teleport}, $p_{\cmark|\psi,\vec{c}_{\kt},{\kt}}$, over Pauli states and over Clifford gates is equal to averaging over all qubit states and all 2-qubit unitary gates according to the Haar measure,

\begin{align}
\frac{1}{|\mathsf{X}|}  \sum_{\psi \in \mathsf{X}}\frac{1}{|\mathsf{Cliff}|^\kt} \sum_{\vec{c}_\kt} p_{\cmark|\psi,\vec{c}_{\kt},{\kt}} = \int \tn{d} \psi \int\tn{d} C_1 \dots \int\tn{d} C_\kt ~p_{\cmark|\psi,\vec{c}_{\kt},{\kt}}.
\end{align}
\end{lemma}

\begin{proof}
Just like in the previous lemma, let us first use cyclicity of the trace, 
\small
\begin{align}
\tn{LHS} & = \int \tn{d}\psi \frac{1}{|\mathsf{Cliff}|^\kt} \sum_{\vec{c}_\kt} ~\Tr [\left(\bigcirc_{j = 1}^{\kt} C_j\right)^{-1} \circ {\mathcal{T}}_\kt(\ketbra{\psi}{\psi})\cdot \ketbra{\psi}{\psi}] \\
& = \int \tn{d}\psi \frac{1}{|\mathsf{Cliff}|^\kt} \sum_{\vec{c}_\kt} ~\Tr [C_1^\dagger \circ \dots \circ C_\kt^\dagger \circ C_\kt \circ M_\kt^\mathcal{T}\circ C_{\kt-1} \circ M_{\kt-1}^\mathcal{T} \circ \dots \circ C_1 \circ M_1^\mathcal{T} (\ketbra{\psi}{\psi})\cdot \ketbra{\psi}{\psi}] \\
& = \int \tn{d}\psi \frac{1}{|\mathsf{Cliff}|^{\kt-1}} \sum_{C_1,...,C_{\kt-1}\in \mathsf{Cliff}} ~\Tr [C_1^\dagger \circ \dots \circ C_{\kt-1}^\dagger \circ M_\kt^\mathcal{T}\circ C_{\kt-1} \circ M_{\kt-1}^\mathcal{T} \circ \dots \circ C_1 \circ M_1^\mathcal{T} (\ketbra{\psi}{\psi})\cdot \ketbra{\psi}{\psi}] \\
& = \int \tn{d}\psi \frac{1}{|\mathsf{Cliff}|^{\kt-1}} \sum_{C_1,...,C_{\kt-2}\in \mathsf{Cliff}} ~\Tr [C_1^\dagger \circ \dots  \circ  C_{\kt-2}^\dagger
\left(\sum_{C_{\kt-1}\in \mathsf{Cliff}}  C_{\kt-1}^\dagger \circ M_\kt^\mathcal{T}\circ C_{\kt-1} \right)
\circ M_{\kt-1}^\mathcal{T} \circ \dots \circ C_1 \circ M_1^\mathcal{T} (\ketbra{\psi}{\psi})\cdot \ketbra{\psi}{\psi}] ,
\end{align}
\normalsize
where in the last step we pulled the summation over $\kt-1$ under the trace. Note that $\left(\sum_{C_{\kt-1}}  C_{\kt-1}^\dagger \circ M_\kt^\mathcal{T}\circ C_{\kt-1} \right)$ is an unnormalized twirl over $\mathsf{Cliff}$ and therefore it commutes with all Clifford gates  $C \in \mathsf{Cliff}$. By repeating pulling the summation under the trace, we can write,
\small
\begin{align}
\tn{LHS} & = \int \tn{d}\psi \frac{1}{|\mathsf{Cliff}|^{\kt-1}} \sum_{C_1,...,C_{\kt-3}\in \mathsf{Cliff}} ~\Tr \left[C_1^\dagger \circ \dots  \circ  C_{\kt-3}^\dagger
\left(\sum_{C_{\kt-2}\in \mathsf{Cliff}}  C_{\kt-2}^\dagger \circ M_{\kt-1}^\mathcal{T}\circ C_{\kt-2} \right) \circ
\left(\sum_{C_{\kt-1}\in \mathsf{Cliff}}  C_{\kt-1}^\dagger \circ M_\kt^\mathcal{T}\circ C_{\kt-1} \right)\right.\\
& \qquad \qquad \qquad \qquad \qquad \qquad \circ M_{\kt-2}^\mathcal{T} \circ \dots \circ C_1 \circ M_1^\mathcal{T} (\ketbra{\psi}{\psi})\cdot \ketbra{\psi}{\psi} \Bigg] \\
& = \int \tn{d}\psi \frac{1}{|\mathsf{Cliff}|^{\kt-1}}  ~\Tr [\bigcirc_{j = 1}^{\kt-1} \left(\sum_{C_{j}\in \mathsf{Cliff}}  C_{j}^\dagger \circ M_{j+1}^\mathcal{T}\circ C_{j} \right) \circ M_{1}^\mathcal{T} (\ketbra{\psi}) \cdot \ketbra{\psi}].
\end{align}
\normalsize
Now we are left with a rather awkward map $M_{1}^\mathcal{T}$ which is not twirled. However, since the Haar measure is invariant under unitary transformations, for all $E \in \mathsf{Cliff}$ it holds that
\begin{align}
\tn{LHS} & = \int \tn{d}\psi \frac{1}{|\mathsf{Cliff}|^{\kt-1}}  ~\Tr [\bigcirc_{j = 1}^{\kt-1} \left(\sum_{C_{j}\in \mathsf{Cliff}}  C_{j}^\dagger \circ M_{j+1}^\mathcal{T}\circ C_{j} \right) \circ M_{1}^\mathcal{T} \circ E (\ketbra{\psi}) \cdot E (\ketbra{\psi})] \\
& = \int \tn{d}\psi \frac{1}{|\mathsf{Cliff}|^{\kt}} \sum_{E \in \mathsf{Cliff}} ~\Tr [\bigcirc_{j = 1}^{\kt-1} \left(\sum_{C_{j}\in \mathsf{Cliff}}  C_{j}^\dagger \circ M_{j+1}^\mathcal{T}\circ C_{j} \right)  \circ M_{1}^\mathcal{T} \circ E (\ketbra{\psi}) \cdot E(\ketbra{\psi})]
\end{align}
where in the last line we used the fact that the value of the expression does not depend on $E$. Now, using cyclicity of the trace and commutativity properties of $E$, we get
\begin{align}
\tn{LHS}& = \int \tn{d}\psi \frac{1}{|\mathsf{Cliff}|^{\kt}}  ~\Tr [\bigcirc_{j = 1}^{\kt-1} \left(\sum_{C_{j}\in \mathsf{Cliff}}  C_{j}^\dagger \circ M_{j+1}^\mathcal{T}\circ C_{j} \right) \circ\left( \sum_{E \in \mathsf{Cliff}} E^\dagger \circ M_{1}^\mathcal{T} \circ E \right)(\ketbra{\psi}) \cdot (\ketbra{\psi})].
\end{align}
Now we can change discrete averaging to the continuous one by definition of the unitary 2-design, see Def.~\ref{APP:DEF:unit2design} and \cite{Gross2007}. We have
\begin{align}
\tn{LHS}& = \int \tn{d}\psi  ~\Tr [\bigcirc_{j = 1}^{\kt-1} \left(\int \tn{d}C_j~  C_{j}^\dagger \circ M_{j+1}^\mathcal{T}\circ C_{j} \right) \circ\left( \int \tn{d}E~ E^\dagger \circ M_{1}^\mathcal{T} \circ E \right)(\ketbra{\psi}) \cdot (\ketbra{\psi})].
\end{align}
To get back to the expression for $p_{\cmark|\psi,\vec{c}_{\kt},{\kt}}$, we can invert the procedure we just applied, i.e.
\small
\begin{align}
\tn{LHS} & = \int \tn{d}\psi \int \tn{d}E ~\Tr [\bigcirc_{j = 1}^{\kt-1} \left(\int \tn{d}C_j~  C_{j}^\dagger \circ M_{j+1}^\mathcal{T}\circ C_{j} \right) \circ\left( E^\dagger \circ M_{1}^\mathcal{T} \circ E \right)(\ketbra{\psi}) \cdot (\ketbra{\psi})] \\
& = \int \tn{d}\psi \int \tn{d}E ~\Tr [\bigcirc_{j = 1}^{\kt-1} \left(\int \tn{d}C_j~  C_{j}^\dagger \circ M_{j+1}^\mathcal{T}\circ C_{j} \right) \circ M_{1}^\mathcal{T} \circ E (\ketbra{\psi}) \cdot E(\ketbra{\psi})] \\
& = \int \tn{d}\psi \int \tn{d}E \int \tn{d}C_1 \dots \int \tn{d}C_{\kt-1} ~\Tr [C_1^\dagger \circ \dots \circ C_{\kt-1}^\dagger \circ M_\kt^\mathcal{T}\circ C_{\kt-1} \circ M_{\kt-1}^\mathcal{T} \circ \dots \circ C_1 \circ M_1^\mathcal{T} (\ketbra{\psi}{\psi})\cdot \ketbra{\psi}{\psi}] \\
& = \int \tn{d}\psi \int \tn{d}E \int \tn{d}C_1 \dots \int \tn{d}C_{\kt-1} ~\Tr [C_1^\dagger \circ \dots \circ C_{\kt-1}^\dagger \circ E^\dagger \circ E \circ M_\kt^\mathcal{T}\circ C_{\kt-1} \circ M_{\kt-1}^\mathcal{T} \circ \dots \circ C_1 \circ M_1^\mathcal{T} (\ketbra{\psi}{\psi})\cdot \ketbra{\psi}{\psi}] 
\end{align}
\normalsize
If now we put $E = C_\kt$ we obtain the desired result. We define $\bar{\bar{F}}_\kt = \int \tn{d} \psi \int\tn{d} C_1 \dots \int\tn{d} C_\kt ~p_{\cmark|\psi,\vec{c}_{\kt},{\kt}}$ as double-averaged fidelity.
\end{proof}

\section{Completeness and soundness} \label{APP:corr_sound}

\subsection{Exact completeness and soundness}\label{APP:SUBSEC:exact_corr_sound}

To keep this section more compact, we use notation from the main text. That is we express Test \ref{PROT:teleport} as  $\mathcal{T}_\kt = \bigcirc_{j = 1}^{\kt} ~ C_{j} \circ M_{j}^\mathcal{T}$, see Eq. \eqref{EQ:teleport-based_test}.

\begin{proof}[Proof of Thm. \ref{THM:exact_corr}]

First, we prove that Test \ref{PROT:teleport} is exactly correct when the winning threshold $\pavg=1$. That is, for honest $A$ and $B$ and for any $1\leq {\kt} \leq k $ after ${\kt}$ rounds the state that the verifier obtains at output $\kt$ is $\bigcirc_{j=1}^\kt  C_j (\ketbra{\psi}{\psi})$.
To prove this, we need to make sure that for all the rounds preceding ${\kt}$ the state at outputs $j=1,\dots,{\kt}$ are correct.
The above can be proven by induction. For $\kt=1$ the verifier measures $ C_{1} (\ketbra{\psi}{\psi})$. 
On the other hand, $C_{1} \circ M_{1}^\mathcal{T} =  C_{1} (\ketbra{\psi}{\psi})$, since the setup is perfect. Repeating this step inductively we get for all $\kt$,
\begin{align}
\bigcirc_{j = 1}^{\kt} ~ C_{j} \circ M_{j}^\mathcal{T} (\ketbra{\psi}{\psi}) = \bigcirc_{j=1}^\kt  C_j (\ketbra{\psi}{\psi})
\end{align}
Hence, $\pavg = 1$. 
\end{proof}

Before proving Thm. \ref{THM:corr_uniq} we formally prove a known fact related to no-cloning theorem \cite{Wootters1982}.
\begin{lemma}
	Let $V_{A \to A'B}$ be an arbitrary isometry, and let for any qubit state $\ket{\psi}_A$, $\ket{\Psi}_{A'B}:=V \ket{\psi}_A$. 
	If for all $\ket{\psi}$, $\Tr_B(\ketbra{\Psi}_{A'B}) = \ketbra{\psi}_{A'}$, then $\ket{\Psi}_{A'B}=\ket{\psi}_{A'}\otimes \ket{{\rm junk}}_B$, where $\ket{{\rm junk}}$ is a pure state independent of $\ket{\psi}$.
\end{lemma}
\begin{proof}
	If the above is true for all $\ket{\psi}$ it is in particular true for $\ket{0}$, namely, $V\ket{0}=\ket{0}\otimes \ket{\sigma_0}$. Similarly $V\ket{1}=\ket{1} \otimes \ket{\sigma_1}$. When now computing the action of $V$ on the state $\ket{+}$, we have $V \ket{+}= \frac{\ket{0}\otimes \ket{\sigma_0} + \ket{1} \otimes \ket{\sigma_1}}{\sqrt{2}}$. But since $\Tr_B(V\ketbra{+}V^{\dagger})=\ketbra{+}$, we must have $\ket{\sigma_0}=\ket{\sigma_1}=:\ket{\rm junk}$
\end{proof}

\begin{corollary} \label{APP:COR:purity}
	Let $\Omega_{A \to A'B}$ be an arbitrary CPTP map, and let for any qubit state $\ketbra{\psi}_A$, $\rho_{A'B}:=\Omega( \ketbra{\psi}_A)$. 
	If for all $\ket{\psi}$, $\Tr_B(\rho_{A'B}) = \ketbra{\psi}_{A'}$, then $\rho_{A'B}=\ketbra{\psi}_{A'}\otimes {\rm junk}_B$, where ${\rm junk}$ is a state independent of $\ketbra{\psi}$.
\end{corollary}
\begin{proof}
	By Stinespring dilation $\exists V_{A \to A' B E}$ $\Omega(\cdot)=\Tr_E(V (\cdot) V^\dagger)$. Since $\Tr_B(\rho_{A'B}) = \ketbra{\psi}_{A'}=\tr_{BE}(V \ketbra{\psi} V^\dagger)$ we must have by the above lemma that $V \ketbra{\psi} V^\dagger= \ketbra{\psi}_{A'} \otimes \ketbra{\rm junk'}_{BE}$, and therefore $\rho_{A'B}=\ketbra{\psi}_{A'} \otimes \Tr_{E}(\ketbra{\rm junk'}_{BE})=\ketbra{\psi}_{A'}\otimes {\rm junk}_B$.
\end{proof}

\vspace{1em}
\begin{proof}[Proof of Thm. \ref{THM:corr_uniq}]
Now we prove that Test \ref{PROT:teleport} is exactly sound. That is, if the average probability of success $\pavg = 1$, then nodes $A$ and $B$ have the ability to correctly execute Test \ref{PROT:teleport}. The intuition behind our proof is that challenges given by the verifier impose a certain structure on the provers strategy. We first show that if the nodes win the test with probability 1, then their strategy must produce the correct state at each time step $\kt$. Then, we argue that this implies that the nodes must have passed the state around, and therefore use a quantum channel between them exactly $\kt$ times.

\begin{lemma}\label{APP:LEM:Omega_corrstate} Let $\mathcal{Q}_{\vec{c}_\kt,\kt}$ be an arbitrary strategy of the provers, which can depend on the information available throughout the protocol, i.e. depth $\kt$ and Clifford string $\vec{c}_\kt$. If the average probability of success in Test \ref{PROT:teleport} is $\pavg = 1$, then for all depths $\kt = 1,\dots,k$, all Clifford strings $\vec{c}_\kt$ and all input states $\psi \in \mathsf{X}$, $\mathcal{Q}_{\vec{c}_\kt,\kt}$ outputs the correct state.
\end{lemma}

\begin{proof}
This statement is essentially the inverse of the exact completeness statement. Let us explicitly write the average probability of success,
\begin{align}
\pavg = \frac{1}{k} \sum_{\kt} \frac{1}{|\mathsf{X}|} \sum_{\psi \in \mathsf{X}} \frac{1}{|\mathsf{Cliff}|^{\kt}} \sum_{\vec{c}_{\kt}\in \mathsf{Cliff}^{\kt}}
\Tr\left[ \mathcal{Q}_{\vec{c}_\kt,\kt} \left( \ketbra{\psi} \right) \cdot  \Pi_\cmark^\kt \right]= 1.
\end{align}
This implies that for all states, gates and depths the trace must be equal to 1,
\begin{align}
\forall \psi \in \mathsf{X},~ \forall \vec{c}_{\kt}\in \mathsf{Cliff}^{\kt},~ \forall \kt = 1,\dots,k: \quad
\Tr\left[ \mathcal{Q}_{\vec{c}_\kt,\kt} \left( \ketbra{\psi} \right) \cdot  \Pi_\cmark^\kt \right] = 1.
\end{align}
Therefore, \begin{align}
\forall \psi \in \mathsf{X},~ \forall \vec{c}_{\kt}\in \mathsf{Cliff}^{\kt},~ \forall \kt = 1,\dots,k: \quad
\mathcal{Q}_{\vec{c}_\kt,\kt} \left( \ketbra{\psi} \right) = C_{\kt}\circ \dots \circ C_{1} (\ketbra{\psi}{\psi})
\end{align}
and the state at every $\kt$ must be exactly the one requested by the verifier.
\end{proof}

\begin{lemma}\label{APP:LEM:exactsound_corrstate}
If the average probability of success in Test \ref{PROT:teleport} is $\pavg = 1$, then for all depths $\kt = 1,\dots,k$, all Clifford strings $\vec{c}_\kt$ and all input states $\psi \in \mathsf{X}$, $\mathcal{Q}_{\vec{c}_\kt,\kt}$ uses an exact sending channel $\kt$ times and apply an operation equivalent to the one described by $\vec{c}_\kt$.
\end{lemma}

\begin{proof}

In our test at every time step $\kt$ the provers must produce some state. Since at every time step a state has to be defined, $\mathcal{Q}_{\vec{c}_\kt,\kt}$ can be described by
\begin{align}
\mathcal{Q}_{\vec{c}_\kt,\kt} = \bigcirc_{j=1}^\kt \mathcal{E}_{\vec{c}_j,j}
\end{align}
Let $\hat{\Gamma}_{A_{\kt-1}}$ and $\hat{\Gamma}_{B_{\kt-1},B_\kt}$ denote CPTP maps which act on registers of $A$ and $B$ respectively, and output qubit states, and let $\Gamma_{A_{\kt-1}}$ and $\Gamma_{B_{\kt-1},B_\kt}$ be $ \Tr_{B_{\kt-1},B_{\kt}}[\hat{\Gamma}_{A_{\kt-1}} (\cdot)]$ and $\Tr_{A_{\kt-1}}[\hat{\Gamma}_{B_{\kt-1},B_\kt} (\cdot)]$ respectively.
The above fact together with Lem. \ref{APP:LEM:paulistates_2design} and \ref{APP:LEM:exactsound_corrstate} implies that at time steps $\kt$ and $\kt-1$
\begin{align}
&\int \tn{d}\psi
\Tr\left[ \mathcal{E}_{\vec{c}_\kt,\kt} \circ \mathcal{Q}_{\vec{c}_{\kt-1},\kt-1}
  \left(\ketbra{\psi}{\psi}\right) \cdot \Pi_\cmark^{\kt} \right] 
  = 1 \Rightarrow \sup_{\Gamma_{B_{\kt-1},B_\kt}} \int \tn{d}\psi
\Tr\left[ \Gamma_{B_\kt-1,B_\kt} \left( \mathcal{E}_{\vec{c}_\kt,\kt} 
  \left(\rho^{\psi}_{\vec{c}_{\kt-1},\kt-1} \right) \right) \cdot \Pi_\cmark^{\kt} \right] = 1 \label{APP:EQ:supBk}\\
&\int \tn{d}\psi
\Tr\left[ \mathcal{Q}_{\vec{c}_{\kt-1},\kt-1}
 \left(\ketbra{\psi}{\psi} \right) \cdot \Pi_\cmark^{\kt-1} \right] 
 = 1 \Rightarrow \sup_{\Gamma_{A_{\kt-1}}}\int \tn{d}\psi
\Tr\left[ \Gamma_{A_{\kt-1}}
  \left(\rho^{\psi}_{\vec{c}_{\kt-1},\kt-1} \right)  \cdot \Pi_\cmark^{\kt-1} \right] = 1
\end{align}
 Moreover, we've put $\rho^{\psi}_{\vec{c}_{\kt-1},\kt-1}:= \mathcal{Q}_{\vec{c}_{\kt-1},\kt-1} \left(\ketbra{\psi}{\psi} \right)$ to denote a joint state of $A$ and $B$ at time step $\kt-1$. Observe that for all $\kt$, $\Pi_\cmark^{\kt} = C_{\kt}\circ \dots \circ C_{1} (\ketbra{\psi}{\psi})$ projects onto a pure state. Therefore, for all $\kt$, the states at output registers $\kt-1$ for $A$, and $\kt$ for $B$, must be pure, 
\begin{align}
\Gamma_{B_{\kt-1},B_\kt} \left( \mathcal{E}_{\vec{c}_\kt,\kt} 
  \left(\rho^{\psi}_{\vec{c}_{\kt-1},\kt-1} \right) \right) = C_{\kt}\circ \dots \circ C_{1} (\ketbra{\psi}{\psi})_{B_\kt}\\
\Gamma_{A_{\kt-1}} \left(\rho^{\psi}_{\vec{c}_{\kt-1},\kt-1} \right) = C_{\kt-1}\circ \dots \circ C_{1} (\ketbra{\psi}{\psi})_{A_{\kt-1}} \label{APP:EQ:Gamma_A}
\end{align}
Let $\sigma^{\psi} = (\hat{\Gamma}_{A_{\kt-1}} \otimes \mathds{1}) (\rho^{\psi}_{\vec{c}_{\kt-1},\kt-1})$ be the joint state of $A$ and $B$ at time step $\kt-1$, and after applying $\hat{\Gamma}_{A_{\kt-1}} $ on $A$. Using Eq. \eqref{APP:EQ:Gamma_A} we have that,
\begin{align}
 \Tr_{B_{\kt-1}} (\sigma^{\psi}) = C_{\kt-1}\circ \dots \circ C_{1} (\ketbra{\psi}{\psi})_{A_{\kt-1}},
\end{align}
which is a pure state on $A$, and therefore any extension of this state has tensor product form across $A$ and $B$, in particular,
\begin{align}
	\sigma^\psi =  C_{\kt-1}\circ \dots \circ C_{1} (\ketbra{\psi}{\psi})_{A_{\kt-1}} \otimes \sigma_{B_{\kt-1}}, \label{APP:EQ:sigma_form}
\end{align}
where $\sigma_{B_{\kt-1}}$ is a state on $B$ independent of $\psi$ by Corr.~\ref{APP:COR:purity}. Therefore the (maximum) average fidelity of the state on $B$'s side is,
\begin{align}
\sup_{\Gamma_{B_{\kt-1}}}\int \tn{d}\psi
\Tr\left[ \Gamma_{B_{\kt-1}}
  \left(\rho^{\psi}_{\vec{c}_{\kt-1},\kt-1} \right)  \cdot \Pi_\cmark^{\kt-1} \right] = \sup_{\Gamma_{B_{\kt-1}}}\int \tn{d}\psi
\Tr\left[ \sigma_{B_{\kt-1}} \cdot\Pi_\cmark^{\kt-1} \right] = \sup_{\Gamma_{B_{\kt-1}}}
\Tr\left[ \sigma_{B_{\kt-1}} \cdot \frac{\mathds{1}}{2} \right] = \frac{1}{2}
\end{align}
This, together with Eq. \eqref{APP:EQ:supBk}, implies that $\mathcal{E}_{\vec{c}_\kt,\kt}$ is an exact sending channel at time step $\kt$. Since the statement holds for all $\kt$, the provers necessarily use the exact sending channel $k$ times. 

\end{proof}

\end{proof}

\subsection{Completeness and soundness}\label{APP:SUBSEC:corr_sound}

\begin{proof}[Proof of Thm. \ref{THM:eps-corr}. Completeness] 

As stated in the main text, we assume that the quality of operations is quantified by average fidelity and that at every round $j$ the quality of operations is the same, i.e. for all $j$, $\bar{\mu} = \int \tn{d} \psi \Tr[ C_j \circ M^\mathcal{T}_j (\ketbra{\psi}) \cdot C_j(\ketbra{\psi})]$. In the following we bound the average probability of success $\pavg$ in terms of $\bar{\mu}$. 
Let us write $\pavg$ explicitly,
\begin{align}
\pavg = \frac{1}{k}  \sum_\kt \frac{1}{|\mathsf{X}|}\sum_{\psi} \frac{1}{|\mathsf{Cliff}|^\kt} \sum_{\vec{c}_\kt} \Tr \left[ \bigcirc_{j=1}^\kt C_j \circ  M^\mathcal{T}_j (\ketbra{\psi}) \cdot \Pi_\kt^\cmark \right]
\end{align}
From Lem. \ref{APP:LEM:cliff_2design} we have that 
\begin{align}
\pavg & = \frac{1}{k}  \sum_\kt \int \tn{d}\psi \int \tn{d}C_1\dots\int \tn{d}C_\kt \Tr \left[ \left( \bigcirc_{j=1}^\kt C_j \right)^{-1} \circ \bigcirc_{j=1}^\kt C_j \circ M^\mathcal{T}_j (\ketbra{\psi}) \cdot   \ketbra{\psi} \right], \\
& = \frac{1}{k}  \sum_\kt \int \tn{d}\psi \Tr [\bigcirc_{j = 1}^{\kt} ~ \left( \int \tn{d} C_j ~C_{j}^\dagger \circ  M^\mathcal{T}_j \circ C_{j} \right) (\ketbra{\psi}{\psi})\cdot \ketbra{\psi}{\psi}].
\end{align}
Observe that $(M^\mathcal{T}_j)_{twirl} = \int \tn{d} C_j ~C_{j}^\dagger \circ  M^\mathcal{T}_j\circ C_{j}$ is a twirl of the operator $M^\mathcal{T}_j$ \cite{Nielsen2002}. Furthermore, twirling any map is equivalent to the action of a depolarizing channel, i.e. $(M^\mathcal{T}_j)_{twirl}(\rho) = \mathcal{D}_j(\rho) = p  \rho + (1-p) \mathds{1}/2$, for some parameter $p$ and any state $\rho$. Using properties of the depolarizing channel we can write,
\begin{align}
\pavg & = \frac{1}{k}  \sum_\kt \int \tn{d}\psi \Tr [\bigcirc_{j = 1}^{\kt} (M^\mathcal{T}_j)_{twirl}(\ketbra{\psi}{\psi})\cdot \ketbra{\psi}{\psi}] \\
 & = \frac{1}{k}  \sum_\kt \int \tn{d}\psi \Tr [\bigcirc_{j = 1}^{\kt} \mathcal{D}_j (\ketbra{\psi}{\psi})\cdot \ketbra{\psi}{\psi}] \\
 & = \frac{1}{k}\sum_\kt \prod_{j=1}^\kt \bar{F}\left(\mathcal{D}_j\right)\\
 & = \frac{1}{k}\sum_\kt \prod_{j=1}^\kt \bar{F}\left((M^\mathcal{T}_j)_{twirl}\right) \\
\end{align} 
Additionally, the average fidelity of a twirled map is equal to the average fidelity of the same map without a twirl \cite{Nielsen2002}, therefore,
$
\pavg  = \frac{1}{k}\sum_\kt \prod_{j=1}^\kt \bar{F}\left(M^\mathcal{T}_j)\right)
$.
By assumption, $\forall j~ \bar{F}\left(M^\mathcal{T}_j\right) = \bar{\mu}$, and
\begin{align}
\pavg  = \frac{1}{k}\sum_\kt \bar{\mu}^\kt = \frac{1}{k} \frac{\bar{\mu}(\bar{\mu}^k-1)}{k(\bar{\mu}-1)} = h_k(\bar{\mu}).
\end{align}
If we demand that $\pavg \geq t$ then $\bar{\mu} \geq h_k^{-1}(t)$.
\end{proof}

\begin{proof}[Proof of Thm. \ref{THM:eps-sound}. Soundness] 
In the case when the nodes $A$ and $B$ are honest, the soundness statement is the converse of the completeness, see the proof above. Here we prove soundness of Test \ref{PROT:teleport} in the case when the nodes are dishonest ($m$-cheating). Just like before, we will assume that output for a fixed $\kt$ happens at node $B$. 

The idea behind this proof is that we bound the average probability of success of the provers when they use a quantum channel between them, and when they do not. More specifically, let $\rho_{A^\out_{\kt-1}}$ be a state available at $A$'s output at time step $\kt-1$ and $\rho^\out_{B_{\kt}}$ be a state available at $B$'s at time step $\kt$. We show that whenever the provers use the channel, the average fidelity between these two states is bounded by 1. However, whenever they do not use the channel, the average fidelity between these two states is at most as large as the average fidelity between the states at time step $\kt-1$, i.e. $\rho_{A^\out_{\kt-1}}$ and $\rho_{B^\out_{\kt-1}}$. This average fidelity is intrinsically bounded by the approximate cloning theorem \cite{Gisin1997}, and here takes value $\frac{5}{6}$. If the provers are $m$-cheating, they use the channel between at least $m$ times. We prove that, as a consequence, their overall average probability of winning $\pavg$ is upper-bounded by $\frac{1}{k} (m + \frac{5}{6}(k-m))$.

\vspace{1em}
When the provers are $m$-cheating they adapt an arbitrary strategy $\mathcal{Q}^m_{\vec{c}_\kt,\kt}$ which depends on the maximum number of channel uses $m$ between the nodes. It can also depend on all the information available throughout the protocol, i.e. the challenges and gates distributed by the verifier. We assume that the executions of the test are IID (independent and identically distributed) and the probability of winning a single execution $i = 1,\dots,n$ is expressed as
\begin{align}
\forall~ \vec{c}_\kt,\kt,\psi \quad 
p_{\cmark|\psi,\vec{c}_\kt,\kt} = \Tr\left[ \mathcal{Q}^m_{\vec{c}_\kt,\kt} \left(\ketbra{\psi}{\psi}\right) \cdot  \Pi_\cmark^\kt \right].
\end{align}
The configuration of channel uses, i.e. at which time step the provers use the channel between them, does not need to be fixed. At each execution, the provers can choose a particular strategy $\mathcal{Q}^{m,\nu}$ which describes a configuration $\nu$ of channel uses between the nodes. We assume that the provers are \emph{non-adaptive} and throughout an execution $i$ their strategy does not change. Therefore, the fact  whether the provers choose to send the state or not, is independent of the information available throughout the protocol. I.e. $\nu$ is independent of $\kt$ and $\vec{c}_\kt$, and we have $q_\nu \geq 0$, $\sum_\nu q_\nu = 1$, such that 
\begin{align}
\mathcal{Q}^m_{\vec{c}_\kt,\kt} = \sum_\nu q_\nu \mathcal{Q}^{m,\nu}_{\vec{c}_\kt,\kt}.
\end{align}
Note that there are $k\choose m$ such strategies. Furthermore, let us define 
\begin{align}
p_{\cmark|\nu,\psi,\vec{c}_\kt,\kt} := \Tr\left[ \mathcal{Q}^{m,\nu}_{\vec{c}_\kt,\kt} \left( \ketbra{\psi} \right) \cdot  \Pi_\cmark^\kt \right].
\end{align}
Let us rewrite the average probability of success, Eq. \eqref{EQ:avg_prob},
\begin{align}
\pavg & = \frac{1}{2}\left( 
\frac{1}{k}  \sum_{\kt=1}^k \frac{1}{|\mathsf{X}|}\sum_{\psi} \frac{1}{|\mathsf{Cliff}|^\kt} \sum_{\vec{c}_\kt} 
p_{\cmark|\psi,\vec{c}_\kt, \kt} + 
\frac{1}{k}  \sum_{\kt=1}^{k} \frac{1}{|\mathsf{X}|}\sum_{\psi} \frac{1}{|\mathsf{Cliff}|^\kt} \sum_{\vec{c}_\kt} 
p_{\cmark|\psi,\vec{c}_\kt, \kt}  \right) .
\end{align}
Now, we will move the summation in the second component of the sum over $\kt$ -- instead of going through $(1,2,\dots, k-1,k)$ we will set it to go $(k,1,2,\dots,k-1)$,
\begin{align}
\pavg &  = \frac{1}{2}\left( 
\frac{1}{k}  \sum_{\kt=1}^k \frac{1}{|\mathsf{X}|}\sum_{\psi} \frac{1}{|\mathsf{Cliff}|^\kt} \sum_{\vec{c}_\kt} 
p_{\cmark|\psi,\vec{c}_\kt, \kt} + 
\frac{1}{k}  \sum_{\kt=k}^{k-1} \frac{1}{|\mathsf{X}|}\sum_{\psi} \frac{1}{|\mathsf{Cliff}|^\kt} \sum_{\vec{c}_\kt} 
p_{\cmark|\psi,\vec{c}_\kt, \kt}  \right) .
\end{align}
Let us define $p_{\cmark|\psi,\vec{c}_0, 0}:=1$ for round $\kt=0$, which one can interpret as simply giving the state to node $A$ and immediately requesting it back. Now since for $\kt=k$ it holds that $p_{\cmark|\psi,\vec{c}_k, k} \leq 1$, we have $p_{\cmark|\psi,\vec{c}_k, k} \leq p_{\cmark|\psi,\vec{c}_0, 0}$. Therefore,
\begin{align}
\pavg &  \leq \frac{1}{2}\left( 
\frac{1}{k}  \sum_{\kt=1}^k \frac{1}{|\mathsf{X}|}\sum_{\psi} \frac{1}{|\mathsf{Cliff}|^\kt} \sum_{\vec{c}_\kt} 
p_{\cmark|\psi,\vec{c}_\kt, \kt} + 
\frac{1}{k}  \sum_{\kt=0}^{k-1} \frac{1}{|\mathsf{X}|}\sum_{\psi} \frac{1}{|\mathsf{Cliff}|^\kt} \sum_{\vec{c}_\kt} 
p_{\cmark|\psi,\vec{c}_\kt, \kt}  \right) .
\end{align}
Now we write the expression as a single summation,
\begin{align}
\pavg  & \leq \frac{1}{k}  \sum_{\kt=1}^k \frac{1}{|\mathsf{X}|}\sum_{\psi} \frac{1}{|\mathsf{Cliff}|^\kt} \sum_{\vec{c}_\kt} 
\frac{p_{\cmark|\psi,\vec{c}_\kt, \kt} +  p_{\cmark|\psi,\vec{c}_{\kt-1}, \kt-1}}{2}\\
\label{APP:EQ:sound_proof_lin} & = \sum_\nu q_\nu \frac{1}{k}  \sum_{\kt=1}^k \frac{1}{|\mathsf{X}|}\sum_{\psi} \frac{1}{|\mathsf{Cliff}|^\kt} \sum_{\vec{c}_\kt} 
\frac{1}{2}\left(p_{\cmark|\nu, \psi,\vec{c}_\kt, \kt} +  p_{\cmark|\nu, \psi,\vec{c}_{\kt-1}, \kt-1}\right) \\
\label{APP:EQ:sound_proof_2des} & = \sum_\nu q_\nu \frac{1}{k}  \sum_{\kt=1}^k \frac{1}{|\mathsf{Cliff}|^\kt} \sum_{\vec{c}_\kt} 
\frac{1}{2} \left( \bar{F}_{\nu, \vec{c}_\kt, \kt} +  \bar{F}_{\nu,\vec{c}_{\kt-1}, \kt-1}\right)
\end{align}
In line \eqref{APP:EQ:sound_proof_lin} we used the linearity property of the trace, and in line \eqref{APP:EQ:sound_proof_2des} we used 2-design properties of the set $\mathsf{X}$ (see argument in Sec. \ref{APP:LEM:paulistates_2design}) together with the fact that $\mathcal{Q}^{m,\nu}_{\vec{c}_\kt,\kt}$ does not depend on the state.

In our test at every time step $\kt$ the provers must produce some state. Since at every time step a state has to be defined, $\mathcal{Q}^{m,\nu}_{\vec{c}_\kt,\kt}$ can be described by
\begin{align}
\mathcal{Q}^{m,\nu}_{\vec{c}_\kt,\kt} = \bigcirc_{j=1}^\kt \mathcal{E}^{m,\nu}_{\vec{c}_j,j}
\end{align}
Now our goal is to bound the probability of winning $\pavg$ if $\mathcal{Q}^{m,\nu}_{\vec{c}_\kt,\kt}$ has exactly $m$ sending channels $\mathcal{E}$, as defined in Def. \ref{DEF:sending_channel_formal}. We will consider two cases: when the channel $\mathcal{E}$ is a sending channel and when it is not.

1. \emph{$\mathcal{E}$ is a sending channel.} In this case the provers can output the correct state at both time steps, $\kt-1$ and $\kt$. Therefore, in this case we use the trivial bound that each of the fidelities is upper-bounded by 1, and 
\begin{align} \label{APP:EQ:sound_bound_1}
\frac{1}{2} \left(\bar{F}_{\nu, \vec{c}_\kt, \kt} +  \bar{F}_{\nu,\vec{c}_{\kt-1}, \kt-1}\right) \leq 1.
\end{align}

2. \emph{$\mathcal{E}$ is not a sending channel.} 
Consider average fidelity expressions at time steps $\kt$ and $\kt-1$ for the same execution $i$, and assume that $\kt$ is odd and output is requested at $B$'s side. Each of the fidelities can be upper-bounded by its supremum,
\begin{align}
\bar{F}_{\nu, \vec{c}_\kt, \kt} & =  \int \tn{d}\psi
\Tr\left[ \mathcal{E}^{m,\nu}_{\vec{c}_\kt,\kt} \circ \mathcal{Q}^{m,\nu}_{\vec{c}_{\kt-1},\kt-1}
  \left(\ketbra{\psi}{\psi}\right) \cdot \Pi_\cmark^{\kt} \right] \\ 
  & \leq \sup_{\Gamma_{B_{\kt-1},B_\kt}} \int \tn{d}\psi
\Tr\left[ \Gamma_{B_\kt-1,B_\kt} \left( \mathcal{E}^{m,\nu}_{\vec{c}_\kt,\kt} 
  \left(\rho^{m,\nu,\psi}_{\vec{c}_{\kt-1},\kt-1} \right) \right) \cdot \Pi_\cmark^{\kt} \right] 
\end{align}
and
\begin{align}
\bar{F}_{\nu,\vec{c}_{\kt-1}, \kt-1} & = \int \tn{d}\psi\Tr\left[ \mathcal{Q}^{m,\nu}_{\vec{c}_{\kt-1},\kt-1}
 \left(\ketbra{\psi}{\psi} \right) \cdot \Pi_\cmark^{\kt-1} \right] \\ 
 & \leq \sup_{\Gamma_{A_{\kt-1}}}\int \tn{d}\psi
\Tr\left[ \Gamma_{A_{\kt-1}}
  \left(\rho^{m,\nu,\psi}_{\vec{c}_{\kt-1},\kt-1} \right)  \cdot \Pi_\cmark^{\kt-1} \right].
\end{align}
Here $\Gamma_{A_{\kt-1}}$ and $\Gamma_{B_{\kt-1},B_\kt}$ denote CPTP maps which trace out additional registers of $A$ and $B$ and output a qubit state. Moreover, we've put $\rho^{m,\nu,\psi}_{\vec{c}_{\kt-1},\kt-1}:= \mathcal{Q}^{m,\nu}_{\vec{c}_{\kt-1},\kt-1}
 \left(\ketbra{\psi}{\psi} \right)$ to denote a joint state of $A$ and $B$ at time step $\kt-1$.

According to Def. \ref{DEF:sending_channel_formal} if the channel is not sending then we can bound 
\begin{align}
\sup_{\Gamma_{B_{\kt-1},B_\kt}} \int \tn{d}\psi
\Tr\left[ \Gamma_{B_{\kt-1},B_\kt} \left( \mathcal{E}^{m,\nu}_{\vec{c}_\kt,\kt} 
  \left(\rho^{m,\nu,\psi}_{\vec{c}_{\kt-1},\kt-1} \right) \right) \cdot \Pi_\cmark^{\kt-1} \right] \leq 
  \sup_{\Gamma_{B_{\kt-1}}} \int \tn{d}\psi
\Tr\left[ \Gamma_{B_{\kt-1}}  \left(\rho^{m,\nu,\psi}_{\vec{c}_{\kt-1},\kt-1} \right) \cdot \Pi_\cmark^{\kt} \right]
\end{align}
and hence,
\begin{align}
\frac{1}{2} \left(\bar{F}_{\nu, \vec{c}_\kt, \kt} +  \bar{F}_{\nu,\vec{c}_{\kt-1}, \kt-1}\right) \leq \frac{1}{2} \left( \sup_{\Gamma_{B_{\kt-1}}} \int \tn{d}\psi
\Tr\left[ \Gamma_{B_{\kt-1}}  \left(\rho^{m,\nu,\psi}_{\vec{c}_{\kt-1},\kt-1} \right) \cdot \Pi_\cmark^{\kt-1} \right] +  \sup_{\Gamma_{A_{\kt-1}}}\int \tn{d}\psi
\Tr\left[ \Gamma_{A_{\kt-1}}
  \left(\rho^{m,\nu,\psi}_{\vec{c}_{\kt-1},\kt-1} \right)  \cdot \Pi_\cmark^{\kt-1} \right]\right)
\end{align}
The right-hand side of the above equation is bounded by $\frac{5}{6}$ due to the approximate cloning theorem \cite{Gisin1997}. Therefore, we have
\begin{align}\label{APP:EQ:sound_bound_approx_clone}
\frac{1}{2}\left( \bar{F}_{\nu, \vec{c}_\kt, \kt} +  \bar{F}_{\nu,\vec{c}_{\kt-1}, \kt-1} \right) \leq \frac{5}{6}.
\end{align}

There are at most $k-m$ time steps $\kt = 1,\dots,k$ such that the channel $\mathcal{E}$ is not sending. Therefore, using Eqs. \eqref{APP:EQ:sound_bound_1} and \eqref{APP:EQ:sound_bound_approx_clone} we can write \eqref{APP:EQ:sound_proof_2des}
\begin{align}
\pavg & \leq \sum_\nu q_\nu \frac{1}{k}  \sum_{\kt=1}^k \frac{1}{|\mathsf{Cliff}|^\kt} \sum_{\vec{c}_\kt} 
\frac{\bar{F}_{\nu, \vec{c}_\kt, \kt} +  \bar{F}_{\nu,\vec{c}_{\kt-1}, \kt-1}}{2} \\
& \leq \sum_\nu q_\nu \frac{1}{k}\left(m \cdot 1 + (k-m) \cdot \frac{5}{6} \right) \\
& = \frac{1}{k}\left(m + \frac{5}{6}(k-m)  \right)
\end{align}
where in the last line we used the fact that $\sum_\nu q_\nu = 1.$
\end{proof}

\section{Other proofs} \label{APP:other_proofs}

In this appendix we present remaining proofs from Sec. \ref{SUBSEC:performance}. First we prove two statements about expected value of the rate of wins and average fidelity. Then we calculate the probability that our consistency check is satisfied. Finally, we derive a bound on the performance of $k$-round protocols in terms of double-average fidelity.

\subsection{Proof of Lem. \ref{LEM:avgfid_expected}}

\label{APP:SUBSEC:proof_avgfid_expected}

Here we prove that the expected value of rate $R_{\vec{c}_\kt,\kt}$, for specific depth $\kt$ and string of Clifford gates $\vec{c}_\kt$, is equal to fidelity $ \bar{F}_{\vec{c}_\kt,\kt}(\tilde{\mathcal{T}}_\kt)$ averaged over the state space.

\begin{proof}[Proof of Lem. \ref{LEM:avgfid_expected}] By definition, the expected value of $v^i_{\vec{c}_\kt,\kt}$ as,
\begin{align}
\mathds{E}\left[ v^i_{\vec{c}_\kt,\kt} \right]_{\mathsf{X}} & = \frac{1}{|\mathsf{X}|}\sum_\psi \left( p_{\cmark|\psi,\vec{c}_{\kt},{\kt}} \cdot 1 + p_{\xmark|\psi,\vec{c}_{\kt},{\kt}} \cdot 0 \right) = \frac{1}{|\mathsf{X}|}\sum_\psi p_{\cmark|\psi,\vec{c}_{\kt},{\kt}}
\end{align}
From the 2-design properties of set $\mathsf{X}$, Lem. \ref{APP:LEM:paulistates_2design}, we have that 
\begin{align}
\mathds{E}\left[ v^i_{\vec{c}_\kt,\kt} \right]_{\mathsf{X}} = \int \tn{d} \psi ~ p_{\cmark|\psi,\vec{c}_{\kt},{\kt}} = \bar{F}_{\vec{c}_\kt,\kt}(\tilde{\mathcal{T}}_\kt)
\end{align}
By linearity property of expected value,
\begin{equation}
\mathds{E}\left[ \sum_i v^i_{\vec{c}_\kt,\kt} \right]_{\mathsf{X}} =  \sum_i \mathds{E}\left[ v^i_{\vec{c}_\kt,\kt} \right]_{\mathsf{X}} = \sum_i \bar{F}_{\vec{c}_\kt,\kt}(\tilde{\mathcal{T}}_\kt) = n_{\vec{c}_\kt,\kt}\bar{F}_{\vec{c}_\kt,\kt}(\tilde{\mathcal{T}}_\kt).
\end{equation}
And therefore,
\begin{align}
\mathds{E}\left[R_{\vec{c}_\kt,\kt}\right]_{\mathsf{X}} = \mathds{E}\left[ \frac{\sum_i v^i_{\vec{c}_\kt,\kt}} {n_{\vec{c}_\kt,\kt}} \right]_{\mathsf{X}}
= \bar{F}_{\vec{c}_\kt,\kt}(\tilde{\mathcal{T}}_\kt)
\end{align}
\end{proof}

\subsection{Proof of Lem. \ref{LEM:doubleavgfid_expected}}
Here we present a proof that is analogous to the previous one, and shows that expected value of rate $R_\kt$ for a fixed depth $\kt$, is equal to double-average fidelity.

\label{APP:SUBSEC:proof_doubleavgfid_expected}

\begin{proof}[Proof of Lem. \ref{LEM:doubleavgfid_expected}] 
By definition, the expected value $\mathds{E}\left[v^i_\kt \right]_{\mathsf{X},\mathsf{Cliff}}$ as
\begin{align}
\mathds{E}\left[ v^i_\kt \right]_{\mathsf{X},\mathsf{Cliff}} & = \frac{1}{|\mathsf{Cliff}|^{\kt}} \frac{1}{|\mathsf{X}|}
\sum_{C_1,\dots,C_{\kt} \in \mathsf{Cliff}} \sum_{\psi \in \mathsf{X}}
\left( p_{\cmark|\psi,\vec{c}_{\kt},{\kt}} \cdot 1 + p_{\xmark|\psi,\vec{c}_{\kt},{\kt}} \cdot 0 \right) \\
& =\frac{1}{|\mathsf{Cliff}|^{\kt}} \frac{1}{|\mathsf{X}|}
\sum_{C_1,\dots,C_{\kt} \in \mathsf{Cliff}} \sum_{\psi \in \mathsf{X}} p_{\cmark|\psi,\vec{c}_{\kt},{\kt}}
\end{align}
By 2-design properties of the Clifford set, Lem. \ref{APP:LEM:cliff_2design}, we have that 
\begin{align}
\mathds{E}\left[ v^i_\kt \right]_{\mathsf{X},\mathsf{Cliff}} = \int \tn{d} \psi \int C_1 \dots \int C_\kt ~p_{\cmark|\psi,\vec{c}_{\kt},{\kt}} = \bar{\bar{F}}(\tilde{\mathcal{T}}_\kt).
\end{align}
Since $\mathds{E}\left[R_\kt \right]_{\mathsf{X},\mathsf{Cliff}}= \frac{\sum_i^{n_{\kt}} \mathds{E}\left[ v^i_\kt\right]_{\mathsf{X},\mathsf{Cliff}}}{n_{\kt}}$, with $n_{\kt}$ being a total number of executions for a fixed $\kt$, we have that $\mathds{E}\left[ R_\kt \right]_{\mathsf{X},\mathsf{Cliff}} = \bar{\bar{F}}_\kt(\tilde{\mathcal{T}}_\kt)$.
\end{proof}

\subsection{Proof of Cor. \ref{COR:probabilities_bound}}\label{APP:SUBSEC:proof_cor_probabilties_bound}
Next, we prove that  the consistency check, Thm.~\ref{THM:inconsistency_check},  is satisfied with a certain probability, determined by the estimates on the performance of individual devices.

\begin{proof}[Proof of Cor. \ref{COR:probabilities_bound}]
The probability that the bound \eqref{EQ:bound_incons} is satisfied is equal to probability that all the individual bounds are satisfied, i.e. $\tn{Pr}\left[ \left(\bigwedge_{j=1}^{\kt} \left( 
|r_{M^{\mathcal{T}}_j} - \bar{F}(\tilde{M}^{\mathcal{T}}_j)| > \epsilon_{M^{\mathcal{T}}_j}
\wedge
 |r_{C_j} - \bar{F}(\tilde{C}_j)| > \epsilon_{C_j}
\right)\right) \right]$. This is equal to 
\begin{align}
&\tn{Pr}\left[ \left(\bigwedge_{j=1}^{\kt} \left( 
|r_{M^{\mathcal{T}}_j} - \bar{F}(\tilde{M}^{\mathcal{T}}_j)| > \epsilon_{M^{\mathcal{T}}_j}
\wedge
 |r_{C_j} - \bar{F}(\tilde{C}_j)| > \epsilon_{C_j}
\right)\right) \right] = \\ 
& \quad =  1 - \tn{Pr}\left[ \left(\bigvee_{j=1}^{\kt} \left( 
|r_{M^{\mathcal{T}}_j} - \bar{F}(\tilde{M}^{\mathcal{T}}_j)| \leq \epsilon_{M^{\mathcal{T}}_j}
\vee
|r_{C_j} - \bar{F}(\tilde{C}_j)| \leq \epsilon_{C_j}
\right)\right) \right].
\end{align}

Since $\tn{Pr}[A \vee B] \leq \Pr[A] + \Pr[B]$, we can write that 
\begin{align}
&\tn{Pr}\left[ \left(\bigvee_{j=1}^{\kt} \left( 
|r_{M^{\mathcal{T}}_j} - \bar{F}(\tilde{M}^{\mathcal{T}}_j)| \leq \epsilon_{M^{\mathcal{T}}_j}
\vee
|r_{C_j} - \bar{F}(\tilde{C}_j)| \leq \epsilon_{C_j}
\right)\right) \right] \leq \\
&\leq \sum_{j=1}^{\kt} \tn{Pr} \left[ |r_{M^{\mathcal{T}}_j} - \bar{F}(\tilde{M}^{\mathcal{T}}_j)| \leq \epsilon_{M^{\mathcal{T}}_j} \right] + \tn{Pr} \left[ |r_{C_j} - \bar{F}(\tilde{C}_j)| \leq \epsilon_{C_j} \right] \\
& \leq 2\sum_{j=1}^{\kt} \left( e^{-2n_{C_j}\epsilon_{C_j}^2} + e^{-2n_{M_j^\mathcal{T}}\epsilon_{M_j^\mathcal{T}}^2}\right)
\end{align}
where in the last line we used Hoeffding inequality, Eq. \eqref{EQ:hoeffding_C} and \eqref{EQ:hoeffding_M}.
Hence, we can write that \\$\tn{Pr}\left[ \left(\bigwedge_{j=1}^{\kt} \left( 
|r_{M^{\mathcal{T}}_j} - \bar{F}(\tilde{M}^{\mathcal{T}}_j)| > \epsilon_{M^{\mathcal{T}}_j}
\wedge
 |r_{C_j} - \bar{F}(\tilde{C}_j)| > \epsilon_{C_j}
\right)\right) \right] \geq 1 - 2\sum_{j=1}^{\kt}\left( e^{-2n_{C_j}\epsilon_{C_j}^2} + e^{-2n_{M_j^\mathcal{T}}\epsilon_{M_j^\mathcal{T}}^2}\right)$.
\end{proof}

\subsection{Proof of Thm. \ref{THM:bound_protocols}} \label{APP:SUBSEC:proof_bound_protocols}

Here we prove our bound on the performance of $k$-round protocols in terms of winning rate $R_\kt$ in Test \ref{PROT:teleport}. The core of this theorem is the following lemma, which relates the diamond distance between the ideal and real implementation of a $k$-round protocol, and the double-average fidelity.

\begin{lemma}\label{APP:LEM:protocol_double_avg_fid}

The performance of a $k$-round protocol can be bounded by the double-averaged fidelity in the following way
\begin{align}
\parallel \tilde{\mathcal{P}}^{k} - \mathcal{P}^{k} \parallel_\diamond \leq 2 \sqrt{d(d+1) |\mathsf{Cliff}|^k\left(1-\bar{\bar{F}}(\tilde{\mathcal{P}}^{k})\right)}
\end{align}
where $d$ is the dimension of the underlying Hilbert space, and $|\mathsf{Cliff}|$ is a size of the Clifford group for dimension d.
\end{lemma}

\begin{proof}[Proof of Lem. \ref{APP:LEM:protocol_double_avg_fid}] \label{APP:proof_protocol_double_avg_fidelity}
To prove the inequality from Lem. \ref{APP:LEM:protocol_double_avg_fid} one needs to show dependence between $\bar{F}_{\vec{g}_{k}}(\tilde{\mathcal{P}}^k)$ and  $\bar{\bar{F}}(\tilde{\mathcal{P}}^k)$. In particular, to preserve the direction of inequality we want that $\bar{F}_{\vec{g}_{k}}(\tilde{\mathcal{P}}^k) \geq \bar{\bar{F}}(\tilde{\mathcal{P}}^k)$. Firstly, we trivially have that 
\begin{align}
\bar{F}_{\vec{g}_{k}}(\tilde{\mathcal{P}}^k) \geq \min_{\vec{g}_{k}} \bar{F}_{\vec{g}_{k}}(\tilde{\mathcal{P}}^k)
\end{align}
On the other hand,
\begin{align}
\bar{\bar{F}}(\tilde{\mathcal{P}}^k) & = \int \tn{d} G_1 \dots \int \tn{d} G_k \bar{F}_{\vec{g}_{k}}(\tilde{\mathcal{P}}^k)\\
\label{EQ:proof_line2}& = \int \tn{d} C_1 \dots \int \tn{d} C_k \bar{F}_{\vec{c}_{k}}(\tilde{\mathcal{P}}^k)\\
\label{EQ:proof_line3}& = \frac{1}{|\tn{Cliff}|^k}\sum_{\vec{c}_{k} \in \tn{Cliff}} \bar{F}_{\vec{c}_{k}}(\tilde{\mathcal{P}}^k) \\
\label{EQ:proof_line4} & = \frac{1}{|\tn{Cliff}|^k} \left( \bar{F}_{\vec{c}_{k}^{\min}}(\tilde{\mathcal{P}}^k) + \sum_{ \vec{c}_{k}\neq \vec{c}_{k}^{\min}} \underbrace{\bar{F}_{\vec{c}_{k}}(\tilde{\mathcal{P}}^k)}_{\leq 1} \right) \\
\label{EQ:proof_line5}& \leq \frac{1}{|\tn{Cliff}|^k} \bar{F}_{\vec{c}_{k}^{\min}}(\tilde{\mathcal{P}}^k) + 1 - \frac{1}{|\tn{Cliff}|^k} = 1 - \frac{1}{|\tn{Cliff}|^k} \left( 1 - \bar{F}_{\vec{c}_{k}^{\min}}(\tilde{\mathcal{P}}^k) \right),
\end{align}
where in lines \eqref{EQ:proof_line2} and \eqref{EQ:proof_line3} we used Lem. \ref{APP:LEM:cliff_2design}, and in line \eqref{EQ:proof_line4} we separated the minimum element out of the summation and in line \eqref{EQ:proof_line5} we bounded each element under the sum by 1. 

Now let us relate the minimum over the Clifford group to a minimum over the whole unitary group. Note that the Clifford group rotated by any unitary $\mathcal{U}$ remains a Clifford group. Therefore, let us rotate every $C_j$ by a constant $\mathcal{U}_j$, $j=1,\dots,k$, such that the minimum over the Clifford sets corresponds to the minimum over the whole unitary group. Let us write $\vec{u}_k = \mathcal{U}_1,\dots \mathcal{U}_k$, 
\begin{align}
\bar{F}_{\vec{u}_k\vec{c}_{\kt}^{\min}}(\tilde{\mathcal{P}}^k) = \min_{\vec{g}_k} \bar{F}_{\vec{g}_k}(\tilde{\mathcal{P}}^k).
\end{align}
We obtain 
\begin{align}
\bar{\bar{F}}(\tilde{\mathcal{P}}^k) \leq 1 - \frac{1}{|\tn{Cliff}|^k} \left( 1 - \min_{\vec{g}_k} \bar{F}_{\vec{g}_k}(\tilde{\mathcal{P}}^k) \right)
\end{align}
and so
\begin{align}
\parallel \tilde{\mathcal{P}}^{k} - \mathcal{P}^{k} \parallel_\diamond \leq 2 \sqrt{d(d+1)} \sqrt{|\tn{Cliff}|^k\left(1-\bar{\bar{F}}(\tilde{\mathcal{P}}^k)\right)}.
\end{align}

\end{proof}
Now we will relate the double-average fidelity of a $k$-round protocol to the double-average fidelity of the test. Indeed, we will show that these quantities are equal.

\begin{lemma}
Double-averaged fidelity of a $k$-round protocol $\tilde{\mathcal{P}}^{k}$ of depth $k$ is equal to double averaged fidelity of the test $\tilde{\mathcal{T}}_\kt$ of the same depth, $\kt = k$,
\begin{align}
\bar{\bar{F}}(\tilde{\mathcal{P}}^{k}) = \bar{\bar{F}}(\tilde{\mathcal{T}}_\kt) 
\end{align}
\end{lemma}
\begin{proof}
As stated in the main text, the proof of this lemma follows from noticing that the expression for double-averaged fidelity contains only polynomials of degree 2 in every Clifford gate $C_j$. Therefore, here averaging over the Clifford group is equivalent to averaging over the entire unitary group, since the Clifford group forms a 2-design. Furthermore, the equality is possible, since we have put $M_j^\mathcal{T} \equiv M_j \circ \mathcal{E}_j$, and $M_j^\mathcal{T}$ encompasses operations associated with sending and storing the qubit. 
\end{proof}
The above two lemmas, combined with the Hoeffding bound on $R_\kt$, Eq. \ref{EQ:hoeffding_doubleavg} complete the proof Thm.~\ref{THM:bound_protocols}.

\section{$Q$-qubit protocols}\label{APP:SEC:q_qubit_ext}
In this section we provide a description of a $Q$-qubit extension of our class of protocols. The structure of our description is exactly the same as the one from Section \ref{SEC:functionality} with the difference that all the operations are carried out on more than one qubit.

In a $Q$-qubit $k$-round protocol nodes have a total of $Q$ qubits available. At each round $j=1,\dots, k$ of the protocol nodes $A$ and $B$ can send any subset of the local qubits to one another. We denote all of the sending operations in round $j$ by $\mathcal{E}_{j}$. Moreover, the nodes can store local qubits in the quantum memory and apply local gates. We denote these operations by $M_{A_j}^{(q_j)}$ and ${G}_{A_j}^{(q_j)}$ for node $A$, and ${G}_{B_j}^{(Q - q_j)}$ and $M_{B_j}^{(Q - q_j)}$ for node $B$. Here $q_j$ and $Q-q_j$ denote the number of local qubits at $A$ or $B$'s side at round $j$, respectively, \emph{after} the sending operation $\mathcal{E}_{j}$.
Therefore, we describe a $Q$-qubit $k$-round protocol can with a map
\begin{align}
{\mathcal{P}}^{k,(Q)}  =  \bigcirc_{j = 1}^{k}  ~ &\Big[\left( {G}_{A_j}^{(q_j)} \circ M_{A_j}^{(q_j)}\right) \otimes \\ & \otimes \left({G}_{B_j}^{(Q - q_j)} \circ M_{B_j}^{(Q - q_j)} \right)\Big] \circ \mathcal{E}_{j}
\end{align}
In the presence of noise, we assume the following noise model for $Q$-qubit $k$-round protocols: 
\begin{itemize}
\item the noise on gates is independent of the applied gate;
\item the noise from memories, gates and transmission channels acts individually on each qubit;
\item for each round $j$, the qubits are submitted to the same kind of noise on node $A$ and node $B$ (noise can differ from round to round).
\end{itemize}
Formally, we assume the following
\begin{align}\label{EQ:noise_Qqubit}
\begin{split}
\tilde{M}_{A_j}^{(q_j)} = \tilde{M}_{j}^{\otimes q_j},& \quad \tilde{M}_{B_j}^{(Q-q_j)} = \tilde{M}_{j}^{\otimes Q-q_j} \\
\tilde{G}_{A_j}^{(q_j)} = G_{A_j}^{(q_j)} \circ {N}_{j}^{\otimes q_j},& \quad \tilde{G}_{B_j}^{(Q-q_j)} = G_{B_j}^{(q_j)} \circ {N}_{j}^{\otimes Q-q_j} 
\end{split}
\end{align}

A bipartite $Q$-qubit, $k$-round protocol between any two nodes $A$ and $B$ consists of the following operations:
\begin{enumerate}
\item Local preparation of $Q$ perfect qubit states, $\ket{\psi}_{A} \in \mathds{D}(\mathcal{H}_A^{\otimes q})$ and $\ket{\psi}_B\in \mathds{D}(\mathcal{H}_B^{\otimes Q-q})$. Here the superscript denotes the number of qubits on $A$'s or $B$'s side.
\item Sending any subset of local qubits from node $A$ to node $B$ and vice versa. We denote all exchanging of qubits in a round $j$ by $\mathcal{E}_j$
\item Storing all local qubits, $M_A^{(q)} = M_A^{\otimes q}$, where $M_A\in U(\mathcal{H}_{A}) $, and $M_B^{(Q-q)} = M_A^{\otimes Q-q}$, where $M_B\in U(\mathcal{H}_{B}) $. Again, the superscript denotes the number qubits on $A$ or $B$'s side. Storage can take up to $k t_M$, where $t_M$ is time necessary for creating one EPR pair and communicating classically between two most distant nodes. A noisy memory is denoted by a tilde, $\tilde{M}_A^{(q)}$ and accordingly for $B$. 
\item Applying an arbitrary local operation by any node on any subset of local qubits. We describe this operation by a unitary gate $G_A^{(q)} \in U(\mathcal{H}_A^{\otimes q})$ and $G_B^{(Q-q)} \in U(\mathcal{H}_A^{\otimes Q-q})$. 
$\tilde{G}_A^{(q)} = G_A^{(q)} \circ N_A^{(q)}$ denotes the noisy counterpart, where $G_A^{(q)}$ is a perfect gate and $N_A^{(q)} = N_A^{\otimes q}$ is a noise map independent of the applied gate. Similarly for gates on $B$'s side. Applying any gate takes a known finite time $\ell\ll t_M$.
\item Local measurement of all local qubits at the end of the protocol, $\Pi_A^{(q)} \in \tn{Proj}(\mathcal{H}_A^{\otimes q})$ and $\Pi_B^{(Q-q)} \in \tn{Proj}(\mathcal{H}_B^{\otimes Q-q})$. As stated before, we assume that the measurement can be performed perfectly.
\end{enumerate}
Steps 2. -- 4. are performed in rounds $j=1,...,k_{\mathcal{P}}$ a total of $k$ times. We denote memories and gates that are used by $A$ and $B$ at a $j$-th round by $M_{A_j}^{(q_j)},\tilde{G}_{A_j}^{(q_j)}$ and $M_{B_j}^{(Q-q_j)},\tilde{G}_{B_j}^{(Q-q_j)}$ respectively. Such a protocol operates on a total number of $Q$ qubits. Note that we model noise map as a product for each of $Q$ qubits.

\begin{definition}[$k$-round protocols]\label{app:DEF:class_of_protocols}
Let $\mathcal{H}^{\otimes Q}$ be the Hilbert space of a two-partite quantum network. We define a $k$-round protocol as a CPTP map of the form $\Pi^{(Q)}\circ \tilde{\mathcal{P}}^{k,(Q)}\circ \tn{Prep}^{(Q)}$, where:
\begin{itemize}
\item $\tn{Prep}^{(Q)}$ corresponds to preparation of $Q$ local qubits $\ket{\psi}_{A} \in \mathds{D}(\mathcal{H}_A^{\otimes q_1})$ and $\ket{\psi}_B\in \mathds{D}(\mathcal{H}_B^{\otimes Q-q_1})$ (Step 1.).

\item $\tilde{\mathcal{P}}^{k,(Q)}$ is a map describing $k$ rounds of local operations -- memories and gates, as well as sending qubits from $A$ to $B$ (Step 2. -- 4.),
\begin{align}
\tilde{\mathcal{P}}^{k,(Q)} & =  \bigcirc_{j = 1}^{k_{\mathcal{P}}}  ~ \Big[\left( \tilde{G}_{A_j}^{(q_j)} \circ M_{A_j}^{(q_j)}\right) \otimes \left(\tilde{G}_{B_j}^{(Q - q_j)} \circ M_{B_j}^{(Q - q_j)} \right)\Big] \circ \mathcal{E}_j.
\end{align}

\item $\Pi_A^{(q)} \otimes \Pi_B^{(Q - q)}$ is a local measurement of all the local qubits. (Step 5.)
\end{itemize}
\end{definition}

Now let us describe a test that certifies the above functionality. It is a straightforward extension of the ping-pong test we have discussed before. The idea of the $Q$-qubit teleportation-based ping-pong test is instead of teleporting a single qubit, to teleport all $Q$ qubits back and forth between nodes $A$ and $B$ and sample a random $Q$-qubit Clifford gate ($\mathsf{Cliff}(2^Q)$) at line 3: of Test \ref{PROT:teleport}. The initial state of $Q$ qubits $\ketbra{\psi}{\psi}^{(Q)}$ is chosen uniformly at random from a 2-design of $Q$-qubit states. In this case the test can be described with a map 
\begin{align}
{\mathcal{T}}^{\kt,(Q)} = \bigcirc_{j=1}^{\kt} {C}_{j}^{(Q)} \circ M_{j}^{\mathcal{T}(Q)},
\end{align}
where, as before, the parity of $j$ indicates on which side all the qubits are, and $M_{j}^{\mathcal{T}(Q)} \equiv M_j^{(Q)} \circ \mathcal{E}_j$ accounts for operations associated with transmission (here teleportation). Note that $M_{j}^{\mathcal{T}(Q)}$ can still be written in the product form, i.e. acting individually on each qubit, since we have assumed that $M_j^{(Q)}$ and $\mathcal{E}_j$ are both in the product form. Therefore, in the presence of noise, we employ the same model as for $k$-round $Q$-qubit protocols, i.e. $\tilde{M}_{j}^{\mathcal{T}(Q)} = (\tilde{M}_{j}^{\mathcal{T}}) ^{\otimes Q}$ and $\tilde{C}_{j}^{(Q)} = {C}_{j}^{(Q)} \circ N_{j}^{\otimes Q}$.

Based on the average fidelity estimate of this test $r(\tilde{\mathcal{T}}^{\kt,(Q)})$ one can, again, check whether memories and gates were used together by satisfying an analog of the bound \eqref{EQ:bound_incons}. This can be done provided one has access to estimates of quality of memories $
\bar{F}({\tilde{M}_{j}^{(Q)}}) = \int \tn{d} \psi^{(Q)} \Tr \left[ (\tilde{M}_{j}^{\mathcal{T}}) ^{\otimes Q} (\ketbra{\psi}{\psi}^{(Q)}) \cdot \ketbra{\psi}{\psi}^{(Q)} \right]$ and gates $\bar{F}({N_{j}^{(Q)}}) = \int \tn{d} \psi^{(Q)} \Tr \left[ N_{j}^{\otimes Q} (\ketbra{\psi}{\psi}^{(Q)}) \cdot \ketbra{\psi}{\psi}^{(Q)} \right]$,
for all $j$. Note that here we necessarily use the fidelity of ${\tilde{M}_{j}^{(Q)}}$ and ${N_{j}^{(Q)}}$ evaluated on the space of all $Q$-qubit states.

Now we can extend Thm. \ref{THM:bound_protocols} onto $Q$-qubit protocols using the noise assumptions on the class of protocols. We arrive at the following statement.

\begin{theorem}[Bounding the behavior of $Q$-qubit $k$-round protocols] \label{THM:bound_Q_protocols}
Given the noise model is the same for all $Q$ qubits at each round $j$, the performance of any $Q$-qubit $Q$-qubit $k$-round protocol, can be bounded in terms of an estimate for the double-averaged fidelity $R(\tilde{\mathcal{T}}^{\kt,(Q)})$ of the $Q$-qubit test in the following way
\begin{align}\label{EQ:bound_Q_protocols}
\begin{split}
&\parallel \tilde{\mathcal{P}}^{k,(Q)} - \mathcal{P}^{k,(Q)} \parallel_\diamond  \leq \\
&\quad \leq 2 \sqrt{d(d+1)} \sum_{\kt}  \sqrt{|\mathsf{Cliff}(d)|^{\kt}\left(1-R(\tilde{\mathcal{T}}^{\kt,(Q)})\right)}
\end{split}
\end{align}
where $d= 2^q$ is the dimension of the underlying Hilbert space, and $|\mathsf{Cliff}(d)|$ is a size of the Clifford group for dimension $d$.
\end{theorem}
The proof of that statement is analogous to the single-qubit case, with the difference that here one uses the properties of the unitary 2-design given by the Clifford group of dimension $2^Q$.

\end{widetext}


\begin{thebibliography}{28}%
\makeatletter
\providecommand \@ifxundefined [1]{%
 \@ifx{#1\undefined}
}%
\providecommand \@ifnum [1]{%
 \ifnum #1\expandafter \@firstoftwo
 \else \expandafter \@secondoftwo
 \fi
}%
\providecommand \@ifx [1]{%
 \ifx #1\expandafter \@firstoftwo
 \else \expandafter \@secondoftwo
 \fi
}%
\providecommand \natexlab [1]{#1}%
\providecommand \enquote  [1]{``#1''}%
\providecommand \bibnamefont  [1]{#1}%
\providecommand \bibfnamefont [1]{#1}%
\providecommand \citenamefont [1]{#1}%
\providecommand \href@noop [0]{\@secondoftwo}%
\providecommand \href [0]{\begingroup \@sanitize@url \@href}%
\providecommand \@href[1]{\@@startlink{#1}\@@href}%
\providecommand \@@href[1]{\endgroup#1\@@endlink}%
\providecommand \@sanitize@url [0]{\catcode `\\12\catcode `\$12\catcode
  `\&12\catcode `\#12\catcode `\^12\catcode `\_12\catcode `\%12\relax}%
\providecommand \@@startlink[1]{}%
\providecommand \@@endlink[0]{}%
\providecommand \url  [0]{\begingroup\@sanitize@url \@url }%
\providecommand \@url [1]{\endgroup\@href {#1}{\urlprefix }}%
\providecommand \urlprefix  [0]{URL }%
\providecommand \Eprint [0]{\href }%
\providecommand \doibase [0]{http://dx.doi.org/}%
\providecommand \selectlanguage [0]{\@gobble}%
\providecommand \bibinfo  [0]{\@secondoftwo}%
\providecommand \bibfield  [0]{\@secondoftwo}%
\providecommand \translation [1]{[#1]}%
\providecommand \BibitemOpen [0]{}%
\providecommand \bibitemStop [0]{}%
\providecommand \bibitemNoStop [0]{.\EOS\space}%
\providecommand \EOS [0]{\spacefactor3000\relax}%
\providecommand \BibitemShut  [1]{\csname bibitem#1\endcsname}%
\let\auto@bib@innerbib\@empty
\bibitem [{\citenamefont {Bennett}\ and\ \citenamefont
  {Brassard}(2014)}]{Bennett1984}%
  \BibitemOpen
  \bibfield  {author} {\bibinfo {author} {\bibfnamefont {C.~H.}\ \bibnamefont
  {Bennett}}\ and\ \bibinfo {author} {\bibfnamefont {G.}~\bibnamefont
  {Brassard}},\ }\href {\doibase https://doi.org/10.1016/j.tcs.2014.05.025}
  {\bibfield  {journal} {\bibinfo  {journal} {Theoretical Computer Science}\
  }\textbf {\bibinfo {volume} {560}},\ \bibinfo {pages} {7 } (\bibinfo {year}
  {2014})},\ \bibinfo {note} {theoretical Aspects of Quantum Cryptography –
  celebrating 30 years of BB84}\BibitemShut {NoStop}%
\bibitem [{\citenamefont {Ekert}(1991)}]{Ekert1991}%
  \BibitemOpen
  \bibfield  {author} {\bibinfo {author} {\bibfnamefont {A.~K.}\ \bibnamefont
  {Ekert}},\ }\href {\doibase 10.1103/PhysRevLett.67.661} {\bibfield  {journal}
  {\bibinfo  {journal} {Phys. Rev. Lett.}\ }\textbf {\bibinfo {volume} {67}},\
  \bibinfo {pages} {661} (\bibinfo {year} {1991})}\BibitemShut {NoStop}%
\bibitem [{\citenamefont {Wehner}\ \emph {et~al.}(2018)\citenamefont {Wehner},
  \citenamefont {Elkouss},\ and\ \citenamefont {Hanson}}]{Wehner2018}%
  \BibitemOpen
  \bibfield  {author} {\bibinfo {author} {\bibfnamefont {S.}~\bibnamefont
  {Wehner}}, \bibinfo {author} {\bibfnamefont {D.}~\bibnamefont {Elkouss}}, \
  and\ \bibinfo {author} {\bibfnamefont {R.}~\bibnamefont {Hanson}},\ }\href
  {\doibase 10.1126/science.aam9288} {\bibfield  {journal} {\bibinfo  {journal}
  {Science}\ }\textbf {\bibinfo {volume} {362}} (\bibinfo {year} {2018}),\
  10.1126/science.aam9288}\BibitemShut {NoStop}%
\bibitem [{\citenamefont {Chuang}\ and\ \citenamefont
  {Nielsen}(1997)}]{Chuang1997}%
  \BibitemOpen
  \bibfield  {author} {\bibinfo {author} {\bibfnamefont {I.~L.}\ \bibnamefont
  {Chuang}}\ and\ \bibinfo {author} {\bibfnamefont {M.~A.}\ \bibnamefont
  {Nielsen}},\ }\href {\doibase 10.1080/09500349708231894} {\bibfield
  {journal} {\bibinfo  {journal} {Journal of Modern Optics}\ }\textbf {\bibinfo
  {volume} {44}},\ \bibinfo {pages} {2455} (\bibinfo {year}
  {1997})}\BibitemShut {NoStop}%
\bibitem [{\citenamefont {Poyatos}\ \emph {et~al.}(1997)\citenamefont
  {Poyatos}, \citenamefont {Cirac},\ and\ \citenamefont
  {Zoller}}]{Poyatos1997}%
  \BibitemOpen
  \bibfield  {author} {\bibinfo {author} {\bibfnamefont {J.~F.}\ \bibnamefont
  {Poyatos}}, \bibinfo {author} {\bibfnamefont {J.~I.}\ \bibnamefont {Cirac}},
  \ and\ \bibinfo {author} {\bibfnamefont {P.}~\bibnamefont {Zoller}},\ }\href
  {\doibase 10.1103/PhysRevLett.78.390} {\bibfield  {journal} {\bibinfo
  {journal} {Phys. Rev. Lett.}\ }\textbf {\bibinfo {volume} {78}},\ \bibinfo
  {pages} {390} (\bibinfo {year} {1997})}\BibitemShut {NoStop}%
\bibitem [{\citenamefont {Merkel}\ \emph {et~al.}(2013)\citenamefont {Merkel},
  \citenamefont {Gambetta}, \citenamefont {Smolin}, \citenamefont {Poletto},
  \citenamefont {C\'orcoles}, \citenamefont {Johnson}, \citenamefont {Ryan},\
  and\ \citenamefont {Steffen}}]{Merkel2013}%
  \BibitemOpen
  \bibfield  {author} {\bibinfo {author} {\bibfnamefont {S.~T.}\ \bibnamefont
  {Merkel}}, \bibinfo {author} {\bibfnamefont {J.~M.}\ \bibnamefont
  {Gambetta}}, \bibinfo {author} {\bibfnamefont {J.~A.}\ \bibnamefont
  {Smolin}}, \bibinfo {author} {\bibfnamefont {S.}~\bibnamefont {Poletto}},
  \bibinfo {author} {\bibfnamefont {A.~D.}\ \bibnamefont {C\'orcoles}},
  \bibinfo {author} {\bibfnamefont {B.~R.}\ \bibnamefont {Johnson}}, \bibinfo
  {author} {\bibfnamefont {C.~A.}\ \bibnamefont {Ryan}}, \ and\ \bibinfo
  {author} {\bibfnamefont {M.}~\bibnamefont {Steffen}},\ }\href {\doibase
  10.1103/PhysRevA.87.062119} {\bibfield  {journal} {\bibinfo  {journal} {Phys.
  Rev. A}\ }\textbf {\bibinfo {volume} {87}},\ \bibinfo {pages} {062119}
  (\bibinfo {year} {2013})}\BibitemShut {NoStop}%
\bibitem [{\citenamefont {Blume-Kohout}\ \emph {et~al.}(2013)\citenamefont
  {Blume-Kohout}, \citenamefont {Gamble}, \citenamefont {Nielsen},
  \citenamefont {Mizrahi}, \citenamefont {Sterk},\ and\ \citenamefont
  {Maunz}}]{BlumeKohout2013}%
  \BibitemOpen
  \bibfield  {author} {\bibinfo {author} {\bibfnamefont {R.}~\bibnamefont
  {Blume-Kohout}}, \bibinfo {author} {\bibfnamefont {J.~K.}\ \bibnamefont
  {Gamble}}, \bibinfo {author} {\bibfnamefont {E.}~\bibnamefont {Nielsen}},
  \bibinfo {author} {\bibfnamefont {J.}~\bibnamefont {Mizrahi}}, \bibinfo
  {author} {\bibfnamefont {J.~D.}\ \bibnamefont {Sterk}}, \ and\ \bibinfo
  {author} {\bibfnamefont {P.}~\bibnamefont {Maunz}},\ }\href@noop {} {\enquote
  {\bibinfo {title} {Robust, self-consistent, closed-form tomography of quantum
  logic gates on a trapped ion qubit},}\ } (\bibinfo {year} {2013}),\ \Eprint
  {http://arxiv.org/abs/arXiv:quant-ph/1310.4492} {arXiv:quant-ph/1310.4492}
  \BibitemShut {NoStop}%
\bibitem [{\citenamefont {Emerson}\ \emph {et~al.}(2005)\citenamefont
  {Emerson}, \citenamefont {Alicki},\ and\ \citenamefont
  {Życzkowski}}]{Emerson2005}%
  \BibitemOpen
  \bibfield  {author} {\bibinfo {author} {\bibfnamefont {J.}~\bibnamefont
  {Emerson}}, \bibinfo {author} {\bibfnamefont {R.}~\bibnamefont {Alicki}}, \
  and\ \bibinfo {author} {\bibfnamefont {K.}~\bibnamefont {Życzkowski}},\
  }\href {http://stacks.iop.org/1464-4266/7/i=10/a=021} {\bibfield  {journal}
  {\bibinfo  {journal} {Journal of Optics B: Quantum and Semiclassical Optics}\
  }\textbf {\bibinfo {volume} {7}},\ \bibinfo {pages} {S347} (\bibinfo {year}
  {2005})}\BibitemShut {NoStop}%
\bibitem [{\citenamefont {Knill}\ \emph {et~al.}(2008)\citenamefont {Knill},
  \citenamefont {Leibfried}, \citenamefont {Reichle}, \citenamefont {Britton},
  \citenamefont {Blakestad}, \citenamefont {Jost}, \citenamefont {Langer},
  \citenamefont {Ozeri}, \citenamefont {Seidelin},\ and\ \citenamefont
  {Wineland}}]{Knill2008}%
  \BibitemOpen
  \bibfield  {author} {\bibinfo {author} {\bibfnamefont {E.}~\bibnamefont
  {Knill}}, \bibinfo {author} {\bibfnamefont {D.}~\bibnamefont {Leibfried}},
  \bibinfo {author} {\bibfnamefont {R.}~\bibnamefont {Reichle}}, \bibinfo
  {author} {\bibfnamefont {J.}~\bibnamefont {Britton}}, \bibinfo {author}
  {\bibfnamefont {R.~B.}\ \bibnamefont {Blakestad}}, \bibinfo {author}
  {\bibfnamefont {J.~D.}\ \bibnamefont {Jost}}, \bibinfo {author}
  {\bibfnamefont {C.}~\bibnamefont {Langer}}, \bibinfo {author} {\bibfnamefont
  {R.}~\bibnamefont {Ozeri}}, \bibinfo {author} {\bibfnamefont
  {S.}~\bibnamefont {Seidelin}}, \ and\ \bibinfo {author} {\bibfnamefont
  {D.~J.}\ \bibnamefont {Wineland}},\ }\href {\doibase
  10.1103/PhysRevA.77.012307} {\bibfield  {journal} {\bibinfo  {journal} {Phys.
  Rev. A}\ }\textbf {\bibinfo {volume} {77}},\ \bibinfo {pages} {012307}
  (\bibinfo {year} {2008})}\BibitemShut {NoStop}%
\bibitem [{\citenamefont {Magesan}\ \emph {et~al.}(2011)\citenamefont
  {Magesan}, \citenamefont {Gambetta},\ and\ \citenamefont
  {Emerson}}]{Magesan2011}%
  \BibitemOpen
  \bibfield  {author} {\bibinfo {author} {\bibfnamefont {E.}~\bibnamefont
  {Magesan}}, \bibinfo {author} {\bibfnamefont {J.~M.}\ \bibnamefont
  {Gambetta}}, \ and\ \bibinfo {author} {\bibfnamefont {J.}~\bibnamefont
  {Emerson}},\ }\href {\doibase 10.1103/PhysRevLett.106.180504} {\bibfield
  {journal} {\bibinfo  {journal} {Phys. Rev. Lett.}\ }\textbf {\bibinfo
  {volume} {106}},\ \bibinfo {pages} {180504} (\bibinfo {year}
  {2011})}\BibitemShut {NoStop}%
\bibitem [{\citenamefont {Pfister}\ \emph {et~al.}(2018)\citenamefont
  {Pfister}, \citenamefont {Rol}, \citenamefont {Mantri}, \citenamefont
  {Tomamichel},\ and\ \citenamefont {Wehner}}]{Pfister2018}%
  \BibitemOpen
  \bibfield  {author} {\bibinfo {author} {\bibfnamefont {C.}~\bibnamefont
  {Pfister}}, \bibinfo {author} {\bibfnamefont {M.~A.}\ \bibnamefont {Rol}},
  \bibinfo {author} {\bibfnamefont {A.}~\bibnamefont {Mantri}}, \bibinfo
  {author} {\bibfnamefont {M.}~\bibnamefont {Tomamichel}}, \ and\ \bibinfo
  {author} {\bibfnamefont {S.}~\bibnamefont {Wehner}},\ }\href {\doibase
  10.1038/s41467-017-00961-2} {\bibfield  {journal} {\bibinfo  {journal}
  {Nature Communications}\ }\textbf {\bibinfo {volume} {9}},\ \bibinfo {pages}
  {27} (\bibinfo {year} {2018})}\BibitemShut {NoStop}%
\bibitem [{\citenamefont {Montanaro}\ and\ \citenamefont
  {de~Wolf}(2016)}]{Montanaro2016}%
  \BibitemOpen
  \bibfield  {author} {\bibinfo {author} {\bibfnamefont {A.}~\bibnamefont
  {Montanaro}}\ and\ \bibinfo {author} {\bibfnamefont {R.}~\bibnamefont
  {de~Wolf}},\ }\href@noop {} {\bibfield  {journal} {\bibinfo  {journal}
  {Theory of Computing, Graduate Surveys}\ }\textbf {\bibinfo {volume} {7}},\
  \bibinfo {pages} {1} (\bibinfo {year} {2016})}\BibitemShut {NoStop}%
\bibitem [{\citenamefont {{Bancal}}\ \emph {et~al.}(2018)\citenamefont
  {{Bancal}}, \citenamefont {{Redeker}}, \citenamefont {{Sekatski}},
  \citenamefont {{Rosenfeld}},\ and\ \citenamefont {{Sangouard}}}]{Bancal2018}%
  \BibitemOpen
  \bibfield  {author} {\bibinfo {author} {\bibfnamefont {J.-D.}\ \bibnamefont
  {{Bancal}}}, \bibinfo {author} {\bibfnamefont {K.}~\bibnamefont {{Redeker}}},
  \bibinfo {author} {\bibfnamefont {P.}~\bibnamefont {{Sekatski}}}, \bibinfo
  {author} {\bibfnamefont {W.}~\bibnamefont {{Rosenfeld}}}, \ and\ \bibinfo
  {author} {\bibfnamefont {N.}~\bibnamefont {{Sangouard}}},\ }\href@noop {}
  {\bibfield  {journal} {\bibinfo  {journal} {arXiv e-prints}\ ,\ \bibinfo
  {eid} {arXiv:1812.09117}} (\bibinfo {year} {2018})},\ \Eprint
  {http://arxiv.org/abs/1812.09117} {arXiv:1812.09117 [quant-ph]} \BibitemShut
  {NoStop}%
\bibitem [{\citenamefont {Reichardt}\ \emph {et~al.}(2013)\citenamefont
  {Reichardt}, \citenamefont {Unger},\ and\ \citenamefont
  {Vazirani}}]{Reichardt2013}%
  \BibitemOpen
  \bibfield  {author} {\bibinfo {author} {\bibfnamefont {B.~W.}\ \bibnamefont
  {Reichardt}}, \bibinfo {author} {\bibfnamefont {F.}~\bibnamefont {Unger}}, \
  and\ \bibinfo {author} {\bibfnamefont {U.}~\bibnamefont {Vazirani}},\ }\href
  {\doibase 10.1038/nature12035} {\bibfield  {journal} {\bibinfo  {journal}
  {Nature}\ }\textbf {\bibinfo {volume} {496}},\ \bibinfo {pages} {456}
  (\bibinfo {year} {2013})}\BibitemShut {NoStop}%
\bibitem [{\citenamefont {Mahadev}(2018)}]{Mahadev2018}%
  \BibitemOpen
  \bibfield  {author} {\bibinfo {author} {\bibfnamefont {U.}~\bibnamefont
  {Mahadev}},\ }\href {\doibase 10.1109/FOCS.2018.00033} {\bibfield  {journal}
  {\bibinfo  {journal} {2018 IEEE 59th Annual Symposium on Foundations of
  Computer Science}\ } (\bibinfo {year} {2018}),\
  10.1109/FOCS.2018.00033}\BibitemShut {NoStop}%
\bibitem [{\citenamefont {Vidick}\ and\ \citenamefont
  {Watrous}(2016)}]{Vidick2016}%
  \BibitemOpen
  \bibfield  {author} {\bibinfo {author} {\bibfnamefont {T.}~\bibnamefont
  {Vidick}}\ and\ \bibinfo {author} {\bibfnamefont {J.}~\bibnamefont
  {Watrous}},\ }\href {https://ieeexplore.ieee.org/document/8186772} {\emph
  {\bibinfo {title} {Quantum Proofs}}}\ (\bibinfo  {publisher} {now},\ \bibinfo
  {year} {2016})\BibitemShut {NoStop}%
\bibitem [{\citenamefont {Wootters}\ and\ \citenamefont
  {Zurek}(1982)}]{Wootters1982}%
  \BibitemOpen
  \bibfield  {author} {\bibinfo {author} {\bibfnamefont {W.~K.}\ \bibnamefont
  {Wootters}}\ and\ \bibinfo {author} {\bibfnamefont {W.~H.}\ \bibnamefont
  {Zurek}},\ }\href {\doibase 10.1038/299802a0} {\bibfield  {journal} {\bibinfo
   {journal} {Nature}\ }\textbf {\bibinfo {volume} {299}},\ \bibinfo {pages}
  {802} (\bibinfo {year} {1982})}\BibitemShut {NoStop}%
\bibitem [{\citenamefont {Gisin}\ and\ \citenamefont
  {Massar}(1997)}]{Gisin1997}%
  \BibitemOpen
  \bibfield  {author} {\bibinfo {author} {\bibfnamefont {N.}~\bibnamefont
  {Gisin}}\ and\ \bibinfo {author} {\bibfnamefont {S.}~\bibnamefont {Massar}},\
  }\href {\doibase 10.1103/PhysRevLett.79.2153} {\bibfield  {journal} {\bibinfo
   {journal} {Phys. Rev. Lett.}\ }\textbf {\bibinfo {volume} {79}},\ \bibinfo
  {pages} {2153} (\bibinfo {year} {1997})}\BibitemShut {NoStop}%
\bibitem [{\citenamefont {Hoeffding}(1963)}]{Hoeffding1963}%
  \BibitemOpen
  \bibfield  {author} {\bibinfo {author} {\bibfnamefont {W.}~\bibnamefont
  {Hoeffding}},\ }\href {\doibase 10.1080/01621459.1963.10500830} {\bibfield
  {journal} {\bibinfo  {journal} {Journal of the American Statistical
  Association}\ }\textbf {\bibinfo {volume} {58}},\ \bibinfo {pages} {13}
  (\bibinfo {year} {1963})}\BibitemShut {NoStop}%
\bibitem [{\citenamefont {Carignan-Dugas}\ \emph {et~al.}(2019)\citenamefont
  {Carignan-Dugas}, \citenamefont {Wallman},\ and\ \citenamefont
  {Emerson}}]{Dugas2016}%
  \BibitemOpen
  \bibfield  {author} {\bibinfo {author} {\bibfnamefont {A.}~\bibnamefont
  {Carignan-Dugas}}, \bibinfo {author} {\bibfnamefont {J.~J.}\ \bibnamefont
  {Wallman}}, \ and\ \bibinfo {author} {\bibfnamefont {J.}~\bibnamefont
  {Emerson}},\ }\href {\doibase 10.1088/1367-2630/ab1800} {\bibfield  {journal}
  {\bibinfo  {journal} {New Journal of Physics}\ }\textbf {\bibinfo {volume}
  {21}},\ \bibinfo {pages} {053016} (\bibinfo {year} {2019})}\BibitemShut
  {NoStop}%
\bibitem [{\citenamefont {Watrous}(2018)}]{Watrous2018}%
  \BibitemOpen
  \bibfield  {author} {\bibinfo {author} {\bibfnamefont {J.}~\bibnamefont
  {Watrous}},\ }\href {\doibase 10.1017/9781316848142} {\emph {\bibinfo {title}
  {The Theory of Quantum Information}}}\ (\bibinfo  {publisher} {Cambridge
  University Press},\ \bibinfo {year} {2018})\BibitemShut {NoStop}%
\bibitem [{\citenamefont {Wallman}\ and\ \citenamefont
  {Flammia}(2014)}]{Wallman2014}%
  \BibitemOpen
  \bibfield  {author} {\bibinfo {author} {\bibfnamefont {J.~J.}\ \bibnamefont
  {Wallman}}\ and\ \bibinfo {author} {\bibfnamefont {S.~T.}\ \bibnamefont
  {Flammia}},\ }\href {http://stacks.iop.org/1367-2630/16/i=10/a=103032}
  {\bibfield  {journal} {\bibinfo  {journal} {New Journal of Physics}\ }\textbf
  {\bibinfo {volume} {16}},\ \bibinfo {pages} {103032} (\bibinfo {year}
  {2014})}\BibitemShut {NoStop}%
\bibitem [{\citenamefont {Goldenberg}\ \emph {et~al.}(1999)\citenamefont
  {Goldenberg}, \citenamefont {Vaidman},\ and\ \citenamefont
  {Wiesner}}]{Goldenberg1999}%
  \BibitemOpen
  \bibfield  {author} {\bibinfo {author} {\bibfnamefont {L.}~\bibnamefont
  {Goldenberg}}, \bibinfo {author} {\bibfnamefont {L.}~\bibnamefont {Vaidman}},
  \ and\ \bibinfo {author} {\bibfnamefont {S.}~\bibnamefont {Wiesner}},\ }\href
  {\doibase 10.1103/PhysRevLett.82.3356} {\bibfield  {journal} {\bibinfo
  {journal} {Phys. Rev. Lett.}\ }\textbf {\bibinfo {volume} {82}},\ \bibinfo
  {pages} {3356} (\bibinfo {year} {1999})}\BibitemShut {NoStop}%
\bibitem [{\citenamefont {Nebe}\ \emph {et~al.}(2001)\citenamefont {Nebe},
  \citenamefont {Rains},\ and\ \citenamefont {Sloane}}]{Nebe2001}%
  \BibitemOpen
  \bibfield  {author} {\bibinfo {author} {\bibfnamefont {G.}~\bibnamefont
  {Nebe}}, \bibinfo {author} {\bibfnamefont {E.~M.}\ \bibnamefont {Rains}}, \
  and\ \bibinfo {author} {\bibfnamefont {N.~J.~A.}\ \bibnamefont {Sloane}},\
  }\href {\doibase 10.1023/A:1011233615437} {\bibfield  {journal} {\bibinfo
  {journal} {Designs, Codes and Cryptography}\ }\textbf {\bibinfo {volume}
  {24}},\ \bibinfo {pages} {99} (\bibinfo {year} {2001})}\BibitemShut {NoStop}%
\bibitem [{\citenamefont {Boykin}\ \emph {et~al.}(1999)\citenamefont {Boykin},
  \citenamefont {Mor}, \citenamefont {Pulver}, \citenamefont {Roychowdhury},\
  and\ \citenamefont {Vatan}}]{Boykin1999}%
  \BibitemOpen
  \bibfield  {author} {\bibinfo {author} {\bibfnamefont {P.}~\bibnamefont
  {Boykin}}, \bibinfo {author} {\bibfnamefont {T.}~\bibnamefont {Mor}},
  \bibinfo {author} {\bibfnamefont {M.}~\bibnamefont {Pulver}}, \bibinfo
  {author} {\bibfnamefont {V.}~\bibnamefont {Roychowdhury}}, \ and\ \bibinfo
  {author} {\bibfnamefont {F.}~\bibnamefont {Vatan}},\ }in\ \href {\doibase
  10.1109/SFFCS.1999.814621} {\emph {\bibinfo {booktitle} {40th Annual
  Symposium on Foundations of Computer Science (Cat. No.99CB37039)}}}\
  (\bibinfo  {publisher} {{IEEE} Comput. Soc},\ \bibinfo {year}
  {1999})\BibitemShut {NoStop}%
\bibitem [{\citenamefont {Roy}\ and\ \citenamefont {Scott}(2009)}]{Roy2009}%
  \BibitemOpen
  \bibfield  {author} {\bibinfo {author} {\bibfnamefont {A.}~\bibnamefont
  {Roy}}\ and\ \bibinfo {author} {\bibfnamefont {A.~J.}\ \bibnamefont
  {Scott}},\ }\href {\doibase 10.1007/s10623-009-9290-2} {\bibfield  {journal}
  {\bibinfo  {journal} {Designs, Codes and Cryptography}\ }\textbf {\bibinfo
  {volume} {53}},\ \bibinfo {pages} {13} (\bibinfo {year} {2009})}\BibitemShut
  {NoStop}%
\bibitem [{\citenamefont {Gross}\ \emph {et~al.}(2007)\citenamefont {Gross},
  \citenamefont {Audenaert},\ and\ \citenamefont {Eisert}}]{Gross2007}%
  \BibitemOpen
  \bibfield  {author} {\bibinfo {author} {\bibfnamefont {D.}~\bibnamefont
  {Gross}}, \bibinfo {author} {\bibfnamefont {K.}~\bibnamefont {Audenaert}}, \
  and\ \bibinfo {author} {\bibfnamefont {J.}~\bibnamefont {Eisert}},\ }\href
  {\doibase 10.1063/1.2716992} {\bibfield  {journal} {\bibinfo  {journal}
  {Journal of Mathematical Physics}\ }\textbf {\bibinfo {volume} {48}},\
  \bibinfo {pages} {052104} (\bibinfo {year} {2007})}\BibitemShut {NoStop}%
\bibitem [{\citenamefont {Nielsen}(2002)}]{Nielsen2002}%
  \BibitemOpen
  \bibfield  {author} {\bibinfo {author} {\bibfnamefont {M.~A.}\ \bibnamefont
  {Nielsen}},\ }\href {\doibase https://doi.org/10.1016/S0375-9601(02)01272-0}
  {\bibfield  {journal} {\bibinfo  {journal} {Physics Letters A}\ }\textbf
  {\bibinfo {volume} {303}},\ \bibinfo {pages} {249 } (\bibinfo {year}
  {2002})}\BibitemShut {NoStop}%
\end{thebibliography}
\end{document}